\newif\ifjournal\journaltrue 

\documentclass{LMCS}

\def\dOi{9(3:11)2013}
\lmcsheading%
{\dOi}
{1--39}
{}
{}
{Jan.~29, 2012}
{Sep.~\phantom09, 2013}
{}

\usepackage{enumerate}

\theoremstyle{plain}

\usepackage{amssymb,amsfonts,latexsym,alltt,xspace}
\newcommand{\cal}{\mathcal}
\newcommand{\citeT}[2]{#1~\cite{#2}}
\newcommand{\citep}[1]{\cite{#1}}

\usepackage{mathptmx}

\usepackage{eucal}

\usepackage{graphicx}
\usepackage{url}

\usepackage{color}
\definecolor{dkblue}{rgb}{0,0.1,0.5}
\definecolor{dkgreen}{rgb}{0,0.4,0}
\definecolor{dkred}{rgb}{0.6,0,0}
\definecolor{linkColor}{rgb}{0,0,0.5}

\usepackage[
  colorlinks=true,
  linkcolor=linkColor,
  citecolor=linkColor,
  pagecolor=linkColor,
  urlcolor=linkColor,
  pdftitle={Measure Transformer Semantics for Bayesian Machine Learning},
  pdfsubject={},
  pdfauthor={Johannes Borgstrom, Andrew D. Gordon, Michael Greenberg, James Margetson, Jurgen Van Gael}
  pdfkeywords={}
]{hyperref}

\newcommand\maybecolor[1]{\color{#1}}

\usepackage{url}
\usepackage{listings}
\usepackage{version}
\usepackage{extarrows}

\newenvironment{prog}{\begin{array}[t]{@{}l@{}}}{\end{array}}

\excludeversion{HIDE}
\newcommand{\displaycomment}[1]{}
\excludeversion{DRAFT}

\includeversion{FULL}
\excludeversion{SHORT}
\newcommand{\MySubsection}[1]{\subsection{#1}}
\newcommand{\MySubSubsection}[1]{\subsubsection*{#1}}

\newcommand{\Note}[3][dkgreen]{{\color{#1}{\bf #2: }#3}}
\newcommand{\todonote}[2]{
\displaycomment{\Note[dkred]{#1}{#2}}%
\typeout{TODO #1: (\thepage) #2}}
\newcommand{\Remark}[2]{
\displaycomment{\Note{(#1}{#2)}}}
\newcommand{\ADG}[1]{\todonote{Andy}{#1}}
\newcommand{\adg}[1]{\Remark{Andy}{#1}}
\newcommand{\JB}[1]{\todonote{JB}{#1}}
\newcommand{\jb}[1]{\Remark{JB}{#1}}

\newcommand{\mmg}[1]{\Remark{MMG}{#1}}

\newcommand{\fun}{Fun\xspace}
\newcommand{\imp}{Imp\xspace}

\newcommand{\bfun}{Bernoulli \fun}
\newcommand{\fsharp}{F\#\xspace}
\newcommand{\csoft}{Csoft\xspace}

\def\lstCaml{\lstset{language=[Objective]Caml,
  morekeywords=[1]{type,val,fun,let,in,ref,of,try,if,then,else,match,with,do,open,module,member,rec,assume,assert,init,sample,random,marginal,observe,weight,guard,nil,yield,unit,bool,int,real,float,local,fail},
  morekeywords=[2]{public,interface,class},
  morekeywords=[3]{},%
  morekeywords=[4]{bn,fn,fv,dom,env,clauses},%
  morestring=[b]",
  sensitive=true,%
  columns=[l]fullflexible,
  texcl=true,
  mathescape=true,
  identifierstyle={\sffamily\small\maybecolor{dkgreen}},
  keywordstyle=[1]{\bfseries\maybecolor{dkblue}},
  keywordstyle=[2]{\bfseries\maybecolor{dkblue}},
  keywordstyle=[3]{\bfseries\maybecolor{black}},
  keywordstyle=[4]{\rmfamily\itshape},
  morecomment=[s]{(*}{*)},
  morecomment=[is]{(*---\ }{*)},
  morecomment=*[l][identifierstyle]{//},
  rangeprefix=(*---\ ,
  includerangemarker=false,
  stringstyle=\ttfamily,
  commentstyle=\rmfamily\itshape,
  showspaces=false,
  showstringspaces=false,
  literate={/\\}{$\wedge\,$}{2} {\\/}{$\vee\,$}{2}
  {>>>}{$\Athen[]\,$}{1}
  {|->}{$\mapsto$}{1}
  {->}{$\rightarrow\,$}{1} {=>}{$\Rightarrow\,$}{1}
  {<-}{$\leftarrow\,$}{1}
  {<=>}{$\Leftrightarrow\,$}{3} {'a}{{\small$\alpha\,$}}{1}
  {'b}{{\small$\beta\,$}}{1} {'c}{{\small$\gamma\,$}}{1}
  {'d}{{\small$\delta\,$}}{1} {'e}{{\small$\epsilon\,$}}{1}
  {'f}{{\small$\phi\,$}}{1} {=def}{$\deq\ $}{3}
  {~}{$\neg$}{1},
  tabsize=2,
  breaklines=true}}

\lstCaml
\let\ls\lstinline
\newcommand{\kw}[1]{\mbox{\normalfont\lstinline{#1}}}

\def\lstsnippet#1#2#3{\lstinputlisting[linerange=#2-#3]{#1}}
\def\lstfrag#1/#2.{\lstsnippet{#1}{#2Begin}{#2End}}

\definecolor{gray}{rgb}{0.8,0.8,0.8}
\lstnewenvironment{rcfcode}
  {\lstset{frame=single,frameround=tttt,rulecolor=\color{gray},xrightmargin=7pt,xleftmargin=7pt}}
  {}

\newcommand{\GAP}{1ex}

\makeatletter
  \newcommand{\addToLabel}[1]{%
    \protected@edef\@currentlabel{\@currentlabel#1}%
  }
\makeatother

\newcommand{\QED}{}
\newenvironment{restate}[1]%
{\begin{trivlist}\item[]{\normalsize\bf Restatement of #1}\hspace*{4mm}\it}%
  {\end{trivlist}}

\newcommand{\hbra}{
  \hbox to \columnwidth{\vrule width0.3mm height 1.8mm depth-0.3mm
    \leaders\hrule height1.8mm depth-1.5mm\hfill
    \vrule width0.3mm height 1.8mm depth-0.3mm}}
\newcommand{\hket}{
  \hbox to \columnwidth{\vrule width0.3mm height1.5mm
    \leaders\hrule height0.3mm\hfill
    \vrule width0.3mm height1.5mm}}

\newcommand{\ratio}{.35}

\newenvironment{display}[2][\ratio]{
  \begin{tabbing}
    \hspace{0.1em} \= \hspace{1.5em} \= \hspace{#1\linewidth-3.2em} \= \hspace{1.5em} \= \kill
    \textbf{#2}\\[-.8ex]
    \hbra\\[-.8ex]
  }
  {\\[-.8ex]\hket
  \end{tabbing}}

\newcommand{\entry}[2]{\>\>$#1$\>\>#2}
\newcommand{\clause}[3][]{\>$#2$\>#1\>#3}
\newcommand{\category}[2]{\clause{#1::=}{#2}}
\newcommand{\smallCategory}[3]{\clause{#1::={#2}}{#3}}
\newcommand{\extendCategory}[2]{\clause{#1::=\;\dots\mid}{#2}}

\newcommand{\ruleclause}{\>}

 \lstnewenvironment{displaylisting}[1]
   {~\\\noindent\textbf{#1}\\[-.8ex]\hbra\vspace{-2ex}}
   {\vspace{-2ex}\hket\vspace{1ex}}

\newcounter{rule}

\newcommand{\staterule}[4][]{%
  \refstepcounter{rule}%
  \addToLabel{(\textsc{#2})}\label{#2}%
  $\begin{array}[b]{@{}l@{}}%
    \mbox{(\textsc{#2})#1}\\%
    \begin{array}{@{}c@{}}
      #3\\
      \hline      \raisebox{0ex}[2.5ex]{\strut}#4%
    \end{array}
  \end{array}$}

\newcommand{\Sref}[1]{Section~\ref{#1}}

\newcommand{\Set}[1]{\{#1\}}                    
\newcommand{\deq}{\triangleq}

\newcommand{\Substfun} [2] {\tiny{{}^{#1}} \!\!/\! _{#2}}
\renewcommand{\Substfun} [2] {\raisebox{.5ex}{\ensuremath{\scriptstyle{#1}}}\!/\! _{#2}}
\newcommand{\Subst}    [2] {\left\{ \Substfun{#1}{#2} \right\}}
\newcommand{\SubstSeq} [3][n]{\left\{ \Substfun{#2_1}{#3_1} \cdots\Substfun{#2_{#1}}{#3_{#1}}\right\}}

\renewcommand{\implies}{\Rightarrow}

\newcommand{\ol}[1]{\overline{#1}}

\newcommand{\qq}[1]{[\hspace{-.25ex}[#1]\hspace{-.25ex}]}

\newcommand{\ty}{:}

\renewcommand{\emptyset}{\varnothing}
\newcommand{\0}{\emptyset}
\newcommand{\nin}{\notin}
\newcommand{\fv}{\operatorname{fv}}               
\newcommand{\dom}{\operatorname{dom}}             
\newcommand{\locs}{\operatorname{locs}}

\newcommand{\tty}{\mathord{:}}
\newcommand{\emptyEnv}{\varepsilon}

\newcommand{\translatorFont}[1]{\mathcal{#1}}
\newcommand{\Alg}[1]{\mathcal{#1}}

\newcommand{\funF}[1]{\kw{#1}}

\newcommand{\typeF}[1]{\funF{#1}}

\newcommand{\graphF}[1]{\mathsf{#1}}
\newcommand{\metaF}[1]{\mathtt{#1}}

\newcommand{\vqq}[2][\;]{\translatorFont{V}\qq{#2}#1}
\newcommand{\tqq}[1]{\translatorFont{T}\qq{#1}}
\newcommand{\pqq}[2][\;]{\translatorFont{J}\qq{#2}#1}
\newcommand{\Pqq}[2][\Sigma']{\translatorFont{J}\qq{#2}^{#1}_\Sigma}

\newcommand{\mAlg}{\Alg{M}}

\newcommand{\realT}{\typeF{real}}
\newcommand{\boolT}{\typeF{bool}}
\newcommand{\intT}{\typeF{int}}
\newcommand{\unitT}{\typeF{unit}}
\newcommand{\State}[1][\Ls]{\metaF{S}\langle #1 \rangle}
\newcommand{\Dist}[1][\Ls]{\metaF{M}\langle \State[#1] \rangle}
\newcommand{\DIST}[1][\;]{\metaF{M}#1}
\newcommand{\Sample}[1]{\kw{random}\ (#1)}
\newcommand{\Vals}[2][\!]{\mathbf{V}_{#1 #2}}

\newcommand{\Int}{\mathbb{Z}}

\newcommand{\Real}{\mathbb{R}}
\newcommand{\Bool}{\mathbb{B}}
\newcommand{\Powerset}[2][]{\mathcal{P}_{#1}(#2)}

\newcommand{\abs}[1]{\lvert #1 \rvert}

\newcommand{\given}[5][]{\mathcal{D}{#2}[#1 #4 #1\lvert\rvert#1 #3=#5 #1]}
\newcommand{\sdot}[1][\,]{{#1 \cdot #1}}
\newcommand{\PROB}{\mathbf{\operatorname{P}}}
\newcommand{\Prob}[2][]{\PROB_{#1}\left[#2\right]}

\newcommand{\PArrP}[3][]{(#2 #1\leadsto #1 #3)}
\newcommand{\PArr}[2][]{#2 #1\leadsto #1}

\newcommand{\typeOf}[1]{\operatorname{ty}(#1)}

\newcommand{\Lif}[3]{\kw{if}\ {#1}\ \kw{then}\ {#2}\ \kw{else}\ {#3}}
\newcommand{\Lfor}[2]{\kw{for}\ {#1}\ \kw{do}\ {#2}}
\newcommand{\Lforeach}[4]{\kw{for}\ {#1}\ \kw{in}_{#2}\ {#3} \ensuremath{\rightarrow} {#4}}
\newcommand{\aqq}[1]{\translatorFont{A}\qq{#1}}

\newcommand{\Lunit}{\kw{()}}
\newcommand{\Linit}{\kw{init}}

\newcommand{\sample}{\xleftarrow{{}_s}}
\newcommand{\assign}{\leftarrow}
\newcommand{\impdt}[1]{\translatorFont{I}\qq{#1}}
\newcommand{\impfg}[1]{\translatorFont{G}\qq{#1}}
\newcommand{\impeg}[1]{\translatorFont{E}\qq{#1}}

\newcommand{\Runs}{\Omega}
\newcommand{\run}{\omega}
\newcommand{\RV}{\kw{value}}
\newcommand{\Obs}{\kw{valid}}

\newcommand{\Range}{\mathcal{R}}

\newcommand{\Padd}[1][\;]{\metaF{add}{#1}}
\newcommand{\Pdrop}[1][\;]{\metaF{drop}{#1}}

\newcommand{\Plookup}[3][\;]{\metaF{lookup}{#1}{#2}{#1}{#3}}

\newcommand{\Punit}{\mbox{\textnormal{\texttt{()}}}}

\newcommand{\Athen}[1][\;]{\ensuremath{{#1}{\metaF{>\!\!\!>\!\!\!>}}{#1}}}
\newcommand{\Atuple}[1][\;]{\ensuremath{{#1}{\metaF{*\!\!*\!\!*}}{#1}}}

\newcommand{\Aleft}[1][\;]{\metaF{left}{#1}}
\newcommand{\Aarr}[1][\;]{\metaF{pure}{#1}}
\newcommand{\Aextend}[1][\;]{\metaF{extend}{#1}}
\newcommand{\Achoose}[1][\;]{\metaF{choose}{#1}}
\newcommand{\Aconstrain}[1][\;]{\metaF{observe}{#1}}
\newcommand{\Aobserve}[1][\;]{\Aconstrain[#1]}

\newcommand{\Femp}{\0}

\newcommand{\Fsample}[3]{\graphF{Sample}_{#2}(#1,#3)}
\newcommand{\FSELECT}{\graphF{Select}}
\newcommand{\Fselect}[4][n]{\FSELECT_{#1}(#2,#3,{#4}_1,\dots,{#4}_{#1})}
\newcommand{\Fequal}[2]{\graphF{Equal}(#1,#2)}
\newcommand{\Fconstant}[2]{\graphF{Constant}_{#2}(#1)}
\newcommand{\Fop}[4][\otimes]{\graphF{Binop}_{#1}(#2,#3,#4)}
\newcommand{\Fconstrain}[2]{\Fconstant{#1}{0_{#2}}}

\newcommand{\Fgate}[3]{\graphF{Gate}(#1,#2,#3)}
\newcommand{\New}[1]{\graphF{new}\ #1\ \graphF{in}\ }

\begin{document}

\title[Measure Transformer Semantics for Bayesian Machine Learning]{Measure Transformer Semantics for\\Bayesian Machine Learning\rsuper*}
\author[J.~Borgstr{\"o}m]{Johannes Borgstr{\"o}m\rsuper a}
\address{{\lsuper a}Dept. of Information Technology\\ Uppsala University\\ Uppsala, Sweden}
\email{borgstrom@acm.org}

\author[A.~D.~Gordon]{Andrew D.~Gordon\rsuper b}
\address{{\lsuper{b,d}}Microsoft Research\\ Cambridge, UK}{}
\email{adg@microsoft.com, jfdm1@roundwood.org}

\author[M.~Greenberg]{Michael Greenberg\rsuper c}
\address{{\lsuper c}University of Pennsylvania\\ Philadelphia, PA, USA}{}
\email{mgree@seas.upenn.edu}

\author[J.~Margetson]{James Margetson\rsuper d}
\address{\vskip-6 pt}{}

\author[J.~Van Gael]{Jurgen Van Gael\rsuper e}
\address{{\lsuper e}Microsoft FUSE Labs\\ Cambridge, UK}{}
\email{jurgen.vangael@gmail.com}



\ACMCCS{[{\bf Theory of computation}]: Semantics and
  reasoning---Program constructs; [{\bf Computing methodologies}]:
  Machine learning---Machine learning approaches}

\titlecomment{{\lsuper*}An abridged version of this paper appears in the
  proceedings of the 20th European Symposium on Programming (ESOP'11),
  part of ETAPS 2011, held in Saarbr{\"u}cken, Germany, March
  26--April 3, 2011. }

\begin{abstract}
The Bayesian approach to machine learning amounts to 
computing posterior distributions of random variables 
from a probabilistic model of how the variables are related 
(that is, a prior distribution) and a set of observations of variables.
There is a trend in machine learning towards expressing Bayesian models as probabilistic programs.
As a foundation for this kind of programming, 
we propose a core functional calculus 
with primitives for sampling prior distributions and observing variables.
We define measure-transformer combinators inspired by theorems in measure theory, 
and use these to give a rigorous semantics to our core calculus.
The original features of our semantics include 
its support for discrete, continuous, and hybrid measures,
and, in particular, for observations of zero-probability events.
We compile our core language to a small imperative language that
is processed by an existing inference engine for factor graphs,
which are data structures that enable many efficient inference algorithms.
This allows efficient approximate inference of posterior marginal distributions,
treating thousands of observations per second for large instances of realistic models.
\end{abstract}
\maketitle


\section{Introduction}


\noindent
In the past 15 years, statistical machine learning has unified many
seemingly unrelated methods through the Bayesian paradigm.
With a solid understanding of the theoretical foundations, advances in algorithms for
inference, and numerous applications, the Bayesian paradigm is now the
state of the art for learning from data.
The theme of this paper is the idea of expressing Bayesian models as probabilistic programs,
which was pioneered by BUGS \cite{GTS94:Bugs}
and is recently gaining in popularity,
witness the following list of probabilistic programming languages:
AutoBayes \cite{Schumann08:AutoBayes},
Alchemy \cite{Domingos:2008:ML:1793956.1793962},
Blaise \cite{blaise},
BLOG \cite{DBLP:conf/ijcai/MilchMRSOK05},
Church~\cite{DBLP:conf/uai/GoodmanMRBT08},
Csoft~\cite{Csoft:WM09},
FACTORIE~\cite{FACTORIE},
Figaro~\cite{DBLP:conf/ilp/Pfeffer10},
HANSEI~\cite{monolingual2009},
HBC~\cite{HBC},
IBAL~\cite{DBLP:conf/ijcai/Pfeffer01},
$\lambda_\circ$~\cite{DBLP:conf/popl/ParkPT05},
Probabilistic cc~\cite{DBLP:conf/popl/GuptaJP99},
PFP~\cite{DBLP:journals/jfp/ErwigK06},
and Probabilistic Scheme~\cite{Radul:2007:RPL:1297081.1297085}.

In particular, we draw inspiration from \csoft \cite{Csoft:WM09}, 
an imperative language where programs denote factor graphs \cite{DBLP:journals/tit/KschischangFL01},
data structures that support efficient inference algorithms \cite{koller09:PGM}.
\csoft is the native language of Infer.NET~\cite{infer.net}, a software library for Bayesian reasoning.
This paper gives an alternative probabilistic semantics 
to languages with features similar to those of \csoft.
%

\subsubsection*{Bayesian Models as Probabilistic Expressions}

Consider a simplified form of TrueSkill \cite{DBLP:conf/nips/HerbrichMG06},
a large-scale online system for ranking computer gamers.
There is a population of players, each assumed to have a skill, which is a
real number that cannot be directly observed.
We observe skills only indirectly via a series of matches.
The problem is to infer the skills of players given the outcomes of the matches.
\begin{FULL}%
  Here is a concrete example: \emph{Alice, Bob, and Cyd are new
    players. In a tournament of three games, Alice beats Bob, Bob
    beats Cyd, and Alice beats Cyd.  What are their skills?}
\end{FULL}%
In a Bayesian setting, we represent our uncertain knowledge of the
skills as continuous probability distributions.  The following
probabilistic expression models the situation by generating probability
distributions for the players' skills, given three played games (observations).

\begin{lstlisting}
// prior distributions, the hypothesis
let skill() = random (Gaussian(10.0,20.0))
let Alice,Bob,Cyd = skill(),skill(),skill()
// observe the evidence
let performance player = random (Gaussian(player,1.0))
observe (performance Alice > performance Bob) //Alice beats Bob
observe (performance Bob   > performance Cyd) //Bob beats Cyd
observe (performance Alice > performance Cyd) //Alice beats Cyd
// return the skills
Alice,Bob,Cyd
\end{lstlisting}
A run of this expression goes as follows.
We sample the skills of the three players from the \emph{prior
  distribution} $\kw{Gaussian}(10.0,20.0)$.
Such a distribution can be pictured as a bell curve centred on the \emph{mean} 10.0,
and gradually tailing off at a rate given by the \emph{variance}, here 20.0.
Sampling from such a distribution is a randomized operation that returns a
real number, most likely close to the mean.
For each match, the run continues by sampling an individual performance for
each of the two players.  Each performance is centred on the skill of a
player, with low variance, making the performance closely correlated with
but not identical to the skill.  We then observe that the winner's
performance is greater than the loser's.
An \emph{observation} $\kw{observe}\ M$ always returns $\kw{()}$,
but represents a constraint that $M$ must be true.
A whole run is valid if all encountered observations are true.
The run terminates by returning the three skills.

A classic computational method to compute an approximate posterior distribution of each of the skills is Monte Carlo sampling~\cite{Mackay03}.
We run the expression many
times, but keep just the valid runs---the ones where the sampled skills and performances are consistent with the observed outcomes.
We then compute the means of the
resulting skills by applying standard statistical formulas.
In the example above, the \emph{posterior distribution} of the
returned skills moves so that the mean of Alice's skill is greater than
 Bob's, which is greater than Cyd's.
%
To the best of our knowledge, all prior inference techniques for probabilistic languages with continuous distributions,
apart from \csoft and recent versions of IBAL~\cite{pfeffer07:ibal}, are
based on nondeterministic inference using some form of Monte Carlo sampling.\adg{added "continuous distributions" because of  work with explicit enumerations eg \cite{DBLP:journals/jfp/ErwigK06}}

Inference algorithms based on factor graphs
\cite{DBLP:journals/tit/KschischangFL01,koller09:PGM} are an efficient
alternative to Monte Carlo sampling.
Factor graphs, used in \csoft, allow deterministic but approximate inference
algorithms, which are known to be significantly more efficient than
sampling methods, where applicable.

Observations with zero probability arise naturally in Bayesian models.
For example, in the model above, a drawn game would be modelled as the performance of two players being observed to be equal.
Since the performances are randomly drawn from a continuous distribution, 
the probability of them actually being equal is zero,
so we would not expect to see \emph{any} valid runs in a Monte Carlo simulation.
(To use Monte Carlo methods, one must instead write that the absolute difference between two drawn performances is
less than some small $\epsilon$.)
However, our semantics based on measure theory makes sense of such observations.
Our semantics is the first for languages 
with continuous or hybrid distributions, such as \fun or \imp, 
that are implemented by deterministic inference via factor graphs.

\subsubsection*{Plan of the Paper}

We propose \fun:
\begin{iteMize}{$\bullet$}
\item \fun is a functional language for Bayesian models with
  primitives for probabilistic sampling and observations (\Sref{sec:prob-fun}).
\item
\fun programs have a rigorous probabilistic semantics as measure transformers (\Sref{sec:coarrows}).
\item
\fun has an efficient implementation: our system compiles \fun to
\imp (\Sref{sec:semantics-fg}), a subset of \csoft, and then relies on Infer.NET (\Sref{sec:experience}).
\item \fun supports array types and array comprehensions in order to
express Bayesian models over large datasets (\Sref{sec:arrays}). 
\end{iteMize}
\noindent
Our main contribution is a framework for finite measure transformer
semantics, which supports discrete measures,
continuous measures, and mixtures of the two,
and also supports observations of zero probability events.
\adg{decided we'd talk about "arrays" instead of lists of collections}

As a substantial application, we supply measure transformer
semantics for \fun and \imp, 
and use the semantics to verify the translations in our compiler.
Theorem~\ref{thm:correspondence-discrete} establishes agreement with
  the discrete semantics of \Sref{sec:prob-fun} for {\bfun}.
Theorem~\ref{thm:fun-to-imp-correct} establishes the correctness of the
compilation from \fun to \imp. 

We designed {\fun} to be a subset of the \fsharp dialect of ML
\cite{SGC07:ExpertFSharp}, for implementation convenience:
\fsharp reflection allows easy access to the abstract syntax of a program.
All the examples in the paper have been executed with our system, described in \Sref{sec:experience}.
%
%
  We end the paper with a description of related work (\Sref{sec:related})
  and some concluding remarks (\Sref{sec:conc}).
\begin{SHORT}
  A companion technical report \cite{BGGMVG11:MeasureTransformerSemantics-tr} includes: detailed proofs;
  extensions of \fun, \imp, and our factor graph notations with array types suitable for inference on large datasets;
  listings of examples including versions of large-scale algorithms;
  and a description, including performance numbers, of our practical implementation of a compiler from \fun to \imp, and a backend based on Infer.NET.
\end{SHORT}

\begin{FULL}%
  Appendix~\ref{app:detailed-proofs} contains proofs omitted from the main body of the paper.
  The technical report version of our paper \cite{fun-esop11} includes additional details, including the code of an {\fsharp} implementation of measure transformers in the discrete case.
\end{FULL}%

\section{Bayesian Models as Probabilistic Expressions}
\label{sec:prob-fun}

\begin{SHORT}
We present a core calculus, \fun, for Bayesian reasoning via
probabilistic functional programming with observations.
\end{SHORT}
\begin{FULL}%
  We introduce the idea of expressing a
  probabilistic model as code in a functional language, \fun, with
  primitives for generating and observing random variables.
  As an illustration, we first consider a subset, \bfun, limited to weighted Boolean choices. 
  We describe in elementary terms an operational semantics for \bfun that allows us
  to compute the conditional probability that the expression yields a
  given value given that the run was valid.
\end{FULL}%

\subsection{Syntax, Informal Semantics, and Bayesian Reading}\label{sec:fun-syntax} 

\adg{we'll use int rather than nat everywhere}
Expressions are strongly typed, with types $t,u$ built up from base scalar types $b$ and pair types.
We let $c$ range over constant data of scalar type, $n$ over integers, and $r$ over real numbers. 
We write $\typeOf{c}=t$ to mean that constant $c$ has type $t$.
For each base type $b$, we define a \emph{zero element} $0_b$ with $0_{\boolT}=\kw{true}$, and let $0_{t*u}=(0_t,0_u)$.
We have arithmetic and Boolean operations $\oplus$ on base types.
\begin{display}[.5]{Types, Constant Data, and Zero Elements:}
\smallCategory{b}{\kw{bool}\mid\kw{int}\mid\kw{real}}{base type}\\
\smallCategory{t,u}{\kw{unit}\mid b\mid (t * u)}{compound type}\\
\>$\typeOf{\Lunit} = \unitT$ 
\quad $\typeOf{\kw{true}} = \typeOf{\kw{false}} = \boolT$ \quad
$\typeOf{n} = \intT$ \quad $\typeOf{r} = \realT$ \\
\>$0_{\kw{bool}} = \kw{true}$ \qquad $0_{\kw{int}}=0$ \qquad $0_{\kw{real}} = 0.0$
\end{display}
\begin{display}{Signatures of Arithmetic and Logical Operators: ${\otimes} : b_1,b_2 \to b_3$}
\>$\&\&, ||, {=} : \kw{bool}, \kw{bool} \to \kw{bool}$\qquad
${>}, {=} : \kw{int}, \kw{int} \to \kw{bool}$ \\
\>$+, -, *, \% : \kw{int}, \kw{int} \to \kw{int}$\qquad
${>} : \kw{real}, \kw{real} \to \kw{bool}$ \qquad
$+, -, * : \kw{real}, \kw{real} \to \kw{real}$
\end{display}
\adg{Also removed integer division and modulus, as they don't actually work in Infer.NET}
We have several standard probability distributions as primitive: $D: t \to
u$ takes parameters in $t$ and yields a random value in $u$.
\begin{FULL}%
The names $x_i$
below only document the meaning of the parameters.
\end{FULL}%
\begin{display}{Signatures of Distributions: 
    $D: (x_1 \ty b_1*\dots*x_n \ty b_n) \to b$}
\>$\kw{Bernoulli} : (\kw{success}:\kw{real}) \to \kw{bool}$\\
\>$\kw{Binomial} : (\kw{trials}:\kw{int}*\kw{success}:\kw{real}) \to \kw{int}$\\
\>$\kw{Poisson} : (\kw{rate}:\kw{real}) \to \kw{int}$\\
\>$\kw{DiscreteUniform} : (\kw{max}:\kw{int}) \to \kw{int}$\\
\>$\kw{Gaussian}: (\kw{mean}:\kw{real}*\kw{variance}:\kw{real}) \to \kw{real}$\\
\>$\kw{Beta}: (\kw{a}:\kw{real}*\kw{b}:\kw{real}) \to \kw{real}$\\
\>$\kw{Gamma}: (\kw{shape}:\kw{real}*\kw{scale}:\kw{real}) \to \kw{real}$
\end{display}
\adg{omitted division on reals, because it's not continuous, and because it's not in Infer.NET;
  equality only on bools (drivable in theory but not in infer.net) and integer; no equality on reals, although derivable given $>$ and negation}
\adg{renamed from \kw{sample} to \kw{random}, as the former seemed like a reference to Monte Carlo methods by ML folks.
We might even consider eliminating the keyword entirely.}
The expressions and values of {\fun} are below.
Expressions are in a limited syntax akin to A-normal form,
with let-expressions for sequential composition.
\begin{display}[.3]{\fun: Values and Expressions}
\smallCategory{V}{{x}\mid{c}\mid(V,V)}{value}\\
\category{M,N}{expression}\\
\entry{V}{value}\\
\entry{V_1 \otimes V_2}{arithmetic or logical operator}\\
\entry{V.1}{left projection from pair}\\
\entry{V.2}{right projection from pair}\\
\entry{\Lif{V}{M_1}{M_2}}{conditional}\\
\entry{\kw{let}\ x=M\ \kw{in}\ N}{let (scope of $x$ is $N$)}\\\
\entry{\Sample{D(V)}}{primitive distribution}\\
\entry{\kw{observe}\ V }{observation}
\end{display}
In the discrete case, \fun has a standard \emph{sampling semantics}~(cf.~\cite{DBLP:conf/popl/ParkPT05}); 
the formal semantics for the general case comes later.
A run of a closed expression $M$ is the process of evaluating $M$ to a value.
The evaluation of most expressions is standard, apart from sampling and observation.

To run $\Sample{D(V)}$, where $V=(c_1,\dots,c_n)$,
choose a value $c$ at random from the distribution $D(c_1,\dots,c_n)$
(independently from earlier random choices)
and return $c$.

To run $\kw{observe}\ V$, always return $\kw{()}$.
We say the observation is \emph{valid} if and only if the value $V$ is some zero element $0_b$.

Due to the presence of sampling, different runs of the same expression may yield more than one value, with differing probabilities.
Let a run be \emph{valid} so long as every encountered observation is valid.
The sampling semantics of an expression is the conditional probability of returning a particular value, given a valid run.
Intuitively, Boolean observations are akin to assume statements in assertion-based program specifications,
where runs of a program are ignored if an assumed formula is false.

\begin{display}{Example: Two Coins, Not Both Tails}
\>\begin{lstlisting}
let heads1 = random (Bernoulli(0.5)) in
let heads2 = random (Bernoulli(0.5)) in
let u = observe (heads1 || heads2) in
(heads1,heads2)
\end{lstlisting}
\end{display}
The subexpression \ls$random (Bernoulli(0.5))$ generates \kw{true} or \kw{false} with equal likelihood.
The whole expression has four distinct runs, each with probability $1/4$, corresponding to the possible combinations of Booleans $\kw{heads1}$ and $\kw{heads2}$.
All these runs are valid, apart from the one where $\kw{heads1}=\kw{false}$ and $\kw{heads2}=\kw{false}$ (representing two tails),
since $\kw{observe} (\kw{false} || \kw{false})$ is not a valid observation.
The sampling semantics of this expression is a probability distribution assigning
probability $1/3$ to the values $(\kw{true},\kw{false})$, $(\kw{false},\kw{true})$, and $(\kw{true},\kw{true})$,
but probability $0$ to the value $(\kw{false},\kw{false})$.

The sampling semantics allows us to interpret an expression as a Bayesian model.
We interpret the distribution of possible return values as the \emph{prior probability} of the model.
The constraints on valid runs induced by observations represent new evidence or training data.
The conditional probability of a value given a valid run is the \emph{posterior probability}:
an adjustment of the prior probability given the evidence or training data.

Thus, the expression above can be read as a Bayesian model of the problem:
\emph{I toss two coins.  I observe that not both are tails. What is the probability of each outcome?}
\begin{FULL} 
The uniform distribution of two Booleans represents our prior knowledge about two coins,
the \kw{observe} expression represents the evidence that not both are tails,
and the overall sampling semantics is the posterior probability of two coins given this evidence.

Next, we define 
syntactic conventions and a type system for {\fun},
define a formal semantics for the discrete subset of {\fun},
and describe further examples.
Our discrete semantics is a warm up before Section~\ref{sec:coarrows}.
There we deploy measure theory to give a semantics to our full language,
which allows both discrete and continuous prior distributions.
\end{FULL}%

\subsection{Syntactic Conventions and Monomorphic Typing Rules}

\begin{FULL}%
We recite our standard syntactic conventions and typing rules.
\end{FULL}%

We identify phrases of syntax $\phi$ (such as values and expressions) up to consistent renaming of bound variables (such as $x$ in a let-expression).
Let $\fv(\phi)$ be the set of variables occurring free in phrase $\phi$.
Let $\phi\Subst{\psi}{x}$ be the outcome of substituting phrase $\psi$ for each free occurrence of variable $x$ in phrase $\phi$.
To keep our core calculus small, we
treat function definitions as macros with call-by-value semantics.
In particular, in examples, we write first-order non-recursive function definitions
in the form $\kw{let}\ f\ x_1\ \dots\ x_n=M$,
and we allow function applications $f\ M_1\ \dots\ M_n$ as expressions.
We consider such a function application as being a shorthand
for the expression $\kw{let}\ x_1=M_1\ \kw{in}\ \dots \kw{let}\ x_n=M_n\ \kw{in}\ M$,
where the bound variables $x_1$, \dots, $x_n$ do not occur free in $M_1$, \dots, $M_n$.
We allow expressions to be used in place of values, via insertion of suitable let-expressions.
For example, $(M_1,M_2)$ stands for $\kw{let}\ x_1=M_1\ \kw{in}\ \kw{let}\ x_2=M_2\ \kw{in}\ (x_1,x_2)$,
and $M_1 \otimes M_2$ stands for $\kw{let}\ x_1=M_1\ \kw{in}\ \kw{let}\ x_2=M_2\ \kw{in}\ x_1 \otimes x_2$,
when either $M_1$ or $M_2$ or both is not a value.
Let $M_1;M_2$ stand for $\kw{let}\ x=M_1\ \kw{in}\ M_2$ where $x \notin \fv(M_2)$.
The notation $t = t_1 * \dots * t_n$ for tuple types means the following:
when $n=0$, $t=\kw{unit}$; when $n=1$, $t=t_1$; 
and when $n>1$, $t=t_1 * (t_2 * \dots * t_n)$.
In listings, we rely on syntactic abbreviations available in \fsharp,
such as layout conventions (to suppress $\kw{in}$ keywords)
and writing tuples as $M_1,\dots,M_n$ without enclosing parentheses.

Let a \emph{typing environment}, $\Gamma$,
be a list of the form $\emptyEnv, x_1 \ty t_1, \dots, x_n \ty t_n$;
we say $\Gamma$ is \emph{well-formed} and write $\Gamma \vdash \diamond$ to mean that the variables $x_i$ are pairwise distinct.
Let $\dom{(\Gamma)}=\{x_1,\dots,x_n\}$ if $\Gamma=\emptyEnv, x_1 \ty t_1, \dots, x_n \ty t_n$.
We sometimes use the notation $\ol{x : t}$ for $\Gamma = \emptyEnv, x_1 \ty t_1, \dots, x_n \ty t_n$
where $\ol{x}=x_1,\dots,x_n$ and $\ol{t}=t_1,\dots,t_n$.

\begin{display}[0.35]{ Typing Rules for {\fun} Expressions: $\Gamma \vdash M : t$}
\begin{FULL}%
  \staterule{Fun Var}{ \Gamma \vdash \diamond \quad (x \ty t) \in
    \Gamma }{ \Gamma \vdash x : t } \quad

  \staterule{Fun Const}{ \Gamma \vdash \diamond }{ \Gamma \vdash c :
    \typeOf{c} } \quad

  \staterule{Fun Pair}{
    \begin{prog}
      \Gamma \vdash V_1 : t_1 \\
      \Gamma \vdash V_2 : t_2
    \end{prog}
  }{ \Gamma \vdash (V_1, V_2) : t_1 * t_2 } \quad
\end{FULL}%
\staterule{Fun Operator}{
  \begin{prog}
  \otimes : b_1,b_2 \to b_3 \\
  \Gamma \vdash V_1 : b_1 \quad
  \Gamma \vdash V_2 : b_2
  \end{prog}
}{
\Gamma \vdash V_1 \otimes V_2 : b_3
}
  \\[\GAP]

  \staterule{Fun Proj1}{ \Gamma \vdash V : t_1 * t_2 }{ \Gamma \vdash V.1 : t_1 } \quad

  \staterule{Fun Proj2}{ \Gamma \vdash V : t_1 * t_2 }{ \Gamma \vdash V.2 : t_2 } \quad

  \staterule{Fun If}{
    \begin{prog}
      \Gamma \vdash V : \boolT \quad  \Gamma \vdash M_1 : t \quad \Gamma \vdash M_2 : t
    \end{prog}
  }{ \Gamma \vdash \Lif{V}{M_1}{M_2} : t
  } \\[\GAP]

  \staterule{Fun Let}{
    \begin{prog}
      \Gamma \vdash M_1 : t_1 \\
      \Gamma,x:t_1 \vdash M_2 : t_2
    \end{prog}
  }{ \Gamma \vdash \kw{let}\ x = M_1\ \kw{in}\ M_2 : t_2 }
\quad

\staterule{Fun Random}{
D : (x_1 \ty b_1 * \dots * x_n \ty b_n) \to b\\
\Gamma \vdash V \ty (b_1 * \dots * b_n)
}{
\Gamma \vdash \Sample{D(V)} : b
}  \quad

\staterule{Fun Observe}{
\Gamma \vdash V : b
}{
\Gamma \vdash \kw{observe}\ V : \unitT
}
\end{display}

\begin{FULL}%
\begin{lem}\label{lemma:subst}
If $\Gamma, x \ty t, \Gamma' \vdash M : t'$ and $\Gamma \vdash V : t$
then $\Gamma, \Gamma' \vdash M\Subst{V}{x} : t'$.
\end{lem}
\begin{proof}
By induction on the derivation of $\Gamma, x \ty t, \Gamma' \vdash M : t'$.\QED
\end{proof}

  \begin{lem}
    If $\Gamma \vdash M : t$ then $\Gamma \vdash \diamond$.
  \end{lem}
  \begin{proof}
    By induction on the derivation of $\Gamma \vdash M : T$.\QED
  \end{proof}

  \begin{lem}[Unique Types]
    If $\Gamma \vdash M : t$ and $\Gamma \vdash M : t'$ then $t=t'$.
  \end{lem}
  \begin{proof}
    By induction on the structure of $M$.  The proof needs that the result
    types of the signatures of overloaded binary operators and of distributions are
    determined by the argument types.\QED
  \end{proof}
\end{FULL}%

\begin{FULL}%
\subsection{Formal Semantics for {\bfun}}

Let {\bfun} be the fragment of our calculus where every \kw{random} expression takes the form $\Sample{\kw{Bernoulli}(c)}$ for some real $c \in (0,1)$,
that is, a weighted Boolean choice returning $\kw{true}$ with probability $c$, and $\kw{false}$ with probability $1-c$.
We show that a closed well-typed expression $M$ induces conditional probabilities $\Prob[M]{\RV=V \mid \Obs}$,
the probability that the value of a valid run of $M$ is $V$.

\pagebreak
For this calculus, we inductively define an operational semantics, $M \to^p M'$, meaning that expression $M$ takes a step to $M'$ with probability $p$.
\begin{display}[.5]{Reduction Relation: $M \to^p M'$ where $p \in (0,1]$}
\clause{V_1 \otimes V_2 \to^1 {\otimes}(c_1,c_2)} \\
\clause{(V_1,V_2).1 \to^1 V_1} \\
\clause{(V_1,V_2).2 \to^1 V_2} \\
\clause{\kw{if}\ \kw{true}\ \kw{then}\ M_1\ \kw{else}\ M_2 \to^1 M_1} \\
\clause{\kw{if}\ \kw{false}\ \kw{then}\ M_1\ \kw{else}\ M_2 \to^1 M_2} \\
\clause{\kw{let}\ x=V\ \kw{in}\ M \to^1 M\Subst{V}{x}} \\[\GAP]
\clause{{\cal R}[M] \to^p {\cal R}[M']$ if $M \to^p M'  \text{~for reduction context~} {\cal R} \text{~given by}}\\
\clause{{\cal R} ::= [] \mid \kw{let}\ x={\cal R}\ \kw{in}\ M}\\[\GAP]
\clause{\Sample{\kw{Bernoulli}(c)} \to^c \kw{true}}{}\\
\clause{\Sample{\kw{Bernoulli}(c)} \to^{1-c} \kw{false}}{}\\
\clause{\kw{observe}\ V \to^1 \kw{()}} 
\end{display}
Since there is no recursion or unbounded iteration in \bfun,
there are no non-terminating reduction sequences $M_1 \to^{p_1} \dots M_n \to^{p_n} \dots$.

%

Moreover, we can prove standard preservation and progress lemmas.
\begin{lem}[Preservation]\label{lemma:preservation}
If $\Gamma \vdash M : t$ and $M \to^p M'$ then $\Gamma \vdash M' : t$.
\end{lem}
\begin{proof}
By induction on the derivation of $M \to^p M'$. \QED
\end{proof}

\begin{lem}[Progress]\label{lemma:progress}
If $\emptyEnv \vdash M : t$ and $M$ is not a value then there are $p$ and $M'$ such that $M \to^p M'$.
\end{lem}
\begin{proof}
By induction on the structure of $M$.\QED
\end{proof}

\begin{lem}[Determinism]\label{lemma:det}
If $M \to^{p} M'$ and $M \to^{p'} M'$ then $p=p'$.
\end{lem}
\begin{proof}
By induction on the structure of $M$.\QED
\end{proof}

\begin{lem}[Probability]\label{lemma:probability}
If $\emptyEnv \vdash M : t$ then $\Sigma_{\Set{(p,N)\mid M\to^pN}}p=1$.
\end{lem}
\begin{proof}
By induction on the structure of $M$.\QED
\end{proof}
We consider a fixed expression $M$ such that $\emptyEnv \vdash M : t$.

Let the space $\Runs$ be the set of all runs of $M$,
where a \emph{run} is a sequence $\run = (M_1,\dots,M_{n+1})$ for $n \geq 0$
and $p_1$, \dots, $p_n$ such that $M = M_1 \to^{p_1} \dots \to^{p_n} M_{n+1} = V$;
we define the functions $\kw{value}(\run)=V$ and $\kw{prob}(\run)=1 p_1 \dots p_n$,
and we define the predicate $\kw{valid}(\run)$ to hold if and only if
whenever $M_j={\cal R}[\kw{observe}\ V]$ then $V=0_b$ for some zero element $0_b$.
Since $M$ is well-typed, is normalizing, and samples only from Bernoulli distributions, $\Runs$ is finite.

Let $\alpha,\beta\subseteq\Runs$ range over \emph{events}, 
and let probability $\Prob[M]{\alpha}=\sum_{\omega \in \alpha} \kw{prob}(\omega)$.  
Below, we write $\Prob{\cdot}$ for $\Prob[M]{\cdot}$ 
when $M$ is clear from the context.

\begin{prop}
The function $\Prob{\alpha}$ forms a \emph{probability distribution}, that is,
(1) we have $\Prob{\alpha} \geq 0$ for all $\alpha$,
(2) $\Prob{\Runs} = 1$,
and (3) $\Prob{\alpha \cup \beta} = \Prob{\alpha} + \Prob{\beta}$ if $\alpha \cap \beta = \emptyset$.
\end{prop}
\begin{proof}
Item (1) follows from the fact that $p\ge 0$ whenever $M\to_p N$.
Item (2) follows from Lemma~\ref{lemma:probability}, Lemma~\ref{lemma:preservation}, and termination.
Item (3) is immediate from the definition.\QED
\end{proof}
To give the semantics of our expression $M$ we first define the following
probabilities and events.
Given a value $V$, $\RV=V$ is the event $\RV^{-1}(V)=\Set{\omega\mid
  \RV(\omega)=V}$.  Hence, $\Prob{\RV = V}$ is the 
\emph{prior probability} that a run of $M$ terminates with $V$.
We let the event $\Obs = \{\run \in \Runs \mid \kw{valid}(\omega)\}$;
hence, $\Prob{\Obs}$ is the probability that a run is valid.

If $\Prob{\beta}\neq0$, the \emph{conditional probability of $\alpha$ given $\beta$} is 
$$
\Prob{\alpha \mid \beta} \deq \frac{ \Prob{\alpha \cap \beta} }{ \Prob{\beta} }
$$ 
The semantics of a program $M$ is given by the conditional probability
distribution
\[
\Prob[M]{\RV=V \mid \Obs}= \frac{\Prob[M]{(\RV^{-1}(V)) \cap \Obs}}{\Prob[M]{\Obs}},
\]
the
conditional probability that a run of $M$ returns $V$ given a valid run, also
known as the \emph{posterior probability}.

The conditional probability $\Prob[M]{\RV=V \mid \Obs}$ is only defined when 
$\Prob[M]{\Obs}$ is not zero.
For pathological choices of $M$ such as $\kw{observe}\ \kw{false}$
or $\kw{let}\ x = 3\ \kw{in}\ \kw{observe}\ x$
there are no valid runs, so $\Prob{\Obs}=0$, and $\Prob{\RV=V \mid \Obs}$ is undefined.
(This is an occasional problem in practice; Bayesian inference engines such as Infer.NET fail in this situation with a zero-probability exception.)


\subsection{An Example in {\bfun}}\label{sec:Epidemiology}
The expression below encodes the question:
\emph{1\% of a population have a disease. 80\% of subjects with the disease test positive,
and 9.6\% without the disease also test positive.
If a subject is positive, what are the odds they have the disease?} \cite{Yudkowsky03}

\begin{display}{Epidemiology: Odds of Disease Given Positive Test}
\>\begin{lstlisting}
let has_disease = random (Bernoulli(0.01))
let positive_result = if has_disease 
                 then random (Bernoulli(0.8))
                 else random (Bernoulli(0.096))
observe positive_result
has_disease
\end{lstlisting}
\end{display}
For this expression, we have $\Omega=\{\omega_{tt},\omega_{tf},\omega_{ft},\omega_{ff}\}$
where each run $\omega_{c_1 c_2}$ corresponds to the choice $\kw{has_disease}=c_1$ and $\kw{positive_result}=c_2$.
The probability of each run is:
\begin{iteMize}{$\bullet$}
\item $\kw{prob}(\omega_{tt}) = 0.01 \times 0.8 = 0.008$ (true positive)
\item $\kw{prob}(\omega_{tf}) = 0.01 \times 0.2 = 0.002$ (false negative)
\item $\kw{prob}(\omega_{ft}) = 0.99 \times 0.096 = 0.09504$ (false positive)
\item $\kw{prob}(\omega_{ff}) = 0.99 \times 0.904 = 0.89496$ (true negative)
\end{iteMize}
The semantics $\Prob{\RV=\kw{true} \mid \Obs}$ here is the conditional
probability of having the disease, given that the test is positive.

Here $\Prob{\Obs}=\kw{prob}(\omega_{ft})+\kw{prob}(\omega_{tt})$ and
$\Prob{\RV = \kw{true} \cap \Obs}=\kw{prob}(\omega_{tt})$, so we have
$\Prob{\RV = \kw{true} \mid \Obs} = 0.008 / (0.008+0.09504) = 0.07764$.
So the likelihood of disease given a positive test is just 7.8\%, less
than one might think.

This example illustrates inference on an explicit enumeration of the runs in $\Omega$.
The size of $\Omega$ is exponential in the number of $\kw{random}$ expressions, so although illustrative,
this style of inference does not scale up.
As we explain in Section~\ref{sec:semantics-fg}, our implementation strategy is to translate {\fun} expressions to 
the input language of an existing inference engine based on factor graphs,
permitting efficient approximate inference.
\end{FULL}%

\section{Semantics  as Measure Transformers}
\label{sec:coarrows}

\adg{Formal bases for probability include the following:
probability mass functions are specific to discrete distributions,
probability density functions are specific to continuous distributions
cumulative distribution functions assume an ordering on each probability domain.
Instead, a measure assigns a weight to each set of samples, and not just to individual samples.}

\noindent
We cannot generalize the operational semantics of the previous section to
continuous distributions, such as $\Sample{\kw{Gaussian}(1,1)}$, since the probability of any particular sample is zero.
%
A further difficulty is the need to observe events with probability zero, a common situation in machine learning.
For example, consider the naive Bayesian classifier, a common, simple probabilistic model.
In the training phase, it is given objects together with their classes and
the values of their pertinent features.
Below, we show the training for a single feature: the weight of the object.
The zero probability events are weight measurements,
assumed to be normally distributed around the class mean.
The outcome of the training is the posterior weight distributions for the different classes.

\adg{throughout use $x-y$ explicitly in examples; this paper is not a tutorial on our tool}
\begin{displaylisting}{Naive Bayesian Classifier, Single Feature Training:}
 let wPrior() = random (Gaussian(0.5,1.0))
 let Glass,Watch,Plate = wPrior(),wPrior(),wPrior()
 let weight objClass objWeight = observe (objWeight-(random (Gaussian(objClass,1.0))))
 weight Glass .18;  weight Glass .21
 weight Watch .11;  weight Watch .073
 weight Plate .23;  weight Plate .45
 Watch,Glass,Plate
\end{displaylisting}
Above, the call to \ls$weight Glass .18$ modifies the distribution of
the variable \ls$Glass$.\label{NaiveBayes}
The example uses \ls$observe (x-y)$ to denote that the
difference between the weights \kw{x} and \kw{y} is 0.
The reason for not
instead writing \ls$x=y$ is that conditioning on events of zero probability
without specifying the random variable they are drawn from is not in
general well-defined, cf. Borel's paradox~\cite{jaynes03:borelParadox}.
To avoid this issue, we instead observe the random variable \ls$x-y$ of type $\kw{real}$, at the
value $0$.  
\begin{FULL}%
  (Our compiler does permit the expression \ls$observe (x=y)$,
  as sugar for \ls$observe (x-y)$).
\end{FULL}%

To give a formal semantics to such observations, as well as to
mixtures of continuous and discrete distributions, we turn to measure
theory, following standard sources 
\citep{billingsley95,Rosenthal06}.  
Two basic concepts are measurable spaces and measures.
A measurable space is a set of values equipped with a collection of \emph{measurable} subsets;
these measurable sets generalize the events of discrete probability.
A \emph{measure} is a function that assigns a positive size to each measurable set;
\emph{finite measures}, which assign a finite size to each measurable set, 
generalize probability distributions.

\begin{FULL}%
  We work in the usual mathematical metalanguage of sets and total
  functions.  To machine-check our theory, one might build on a recent
  formalization of measure theory and Lebesgue integration in
  higher-order logic \cite{MHT10:LebesgueIntegration}.
\end{FULL}%

\subsection{Types as Measurable Spaces}

In the remainder of the paper, we
let $\Omega$ range over sets of possible outcomes;
 in our semantics $\Omega$ will range over $\Bool = \{\kw{true},\kw{false}\}$, $\Int$, $\Real$,
and finite Cartesian products of these sets.
A \emph{$\sigma$-algebra} over $\Omega$ is a set
$\mAlg\subseteq\Powerset{\Omega}$ which (1) contains $\0$ and $\Omega$, and
(2) is closed under complement and countable union and intersection.
A \emph{measurable space} is a pair $(\Omega,\mAlg)$ where $\mAlg$ is a
$\sigma$-algebra over $\Omega$; the elements of $\mAlg$ are called 
\emph{measurable sets}.
We use the notation $\sigma_\Omega(S)$, when $S\subseteq\Powerset{\Omega}$,
for the smallest $\sigma$-algebra over $\Omega$ that is a superset of
$S$; we may omit $\Omega$ when it is clear from context.
%
%
\begin{FULL}%
  Given two measurable spaces $(\Omega_1,\mAlg_1)$ and
  $(\Omega_2,\mAlg_2)$, we can compute their product as
  $(\Omega_1,\mAlg_1)\times(\Omega_2,\mAlg_2)\deq
  (\Omega_1\times\Omega_2,\sigma_{\Omega_1\times\Omega_2}\Set{A\times B\mid
    A\in\mAlg_1, B\in\mAlg_2})$
\end{FULL}%
If $(\Omega, \mAlg)$ and $(\Omega', \mAlg')$ are measurable spaces, then the function
$f:\Omega\to\Omega'$ is \emph{measurable} if and only if for all $A\in\mAlg'$,
$f^{-1}(A)\in\mAlg$,
where the \emph{inverse image} $f^{-1} : \Powerset{\Omega'} \to \Powerset{\Omega}$
is given by $f^{-1}(A) \deq \{\omega \in \Omega \mid f(\omega) \in A\}$. 
We write $f^{-1}(x)$ for $f^{-1}(\Set x)$ when $x\in \Omega'$.

We give each first-order type $t$ an interpretation as a
measurable space $\tqq{t}\deq(\Vals t, \Alg{M}_t)$ below. 
We identify closed values of type $t$ with elements of $\Vals{t}$,
and write $\Punit$ for $\0$, the unit value.

\begin{display}{Semantics of Types as Measurable Spaces:}
\clause{\tqq{\unitT}=(\Set{\Punit},\Set{\Set{\Punit},\0})} 
{$\tqq{\boolT}=(\Bool,\Powerset\Bool)$} \\
\clause{\tqq{\intT}= (\Int,\Powerset\Int)} 
{$\tqq{\realT}=(\Real,\sigma_{\Real}(\Set{[a,b]\mid a,b\in \Real}))$}\\
\clause{\tqq{t*u}=\tqq{t}\times\tqq{u}}
\end{display}
The set $\sigma_{\Real}(\Set{[a,b]\mid a,b\in \Real})$ in the definition of
$\tqq{\realT}$ is the Borel $\sigma$-algebra on the real line, which is the
smallest $\sigma$-algebra containing all closed (and open) intervals.
Below, we write $f:t\to u$ to denote that $f:\Vals t\to \Vals u$ is measurable,
that is, that $f^{-1}(B)\in \mAlg_t$ for all $B\in\mAlg_u$.

\subsection{Finite Measures}
\newcommand{\Lebesgue}[3]{\int_#1 #2\,d#3}

A \emph{measure} $\mu$ on a measurable space $(\Omega, \mAlg)$ is a
function $\mAlg\to\Real^{+}\cup\{\infty\}$ 
that is countably additive, that is, $\mu(\0)=0$ and if the
sets $A_0,A_1,\ldots\in\mAlg$ are pairwise disjoint, then
$\mu(\cup_i A_i)=\sum_i\mu(A_i)$.  We write $\abs\mu\deq\mu(\Omega)$.
A \emph{finite measure} $\mu$ is a measure $\mu$ satisfying $\abs\mu\neq\infty$;
a \emph{$\sigma$-finite measure} $\mu$ is a measure such that $\Omega=A_0\cup A_1\cup\dots$ with $\mu(A_i)\neq\infty$.
All the measures we consider in this paper are $\sigma$-finite.

Let $\DIST{t}$ be the set of finite measures on the measurable space $\tqq{t}$.
\begin{FULL}%
  Additionally, a finite measure $\mu$ on $(\Omega, \mAlg)$ is
  a \emph{probability measure} when $\abs\mu=1$.  
  We do not restrict $\DIST{t}$ to just probability measures, although one
  can obtain a probability measure from a non-zero finite measure by
  normalizing with $1/\abs\mu$.
\end{FULL}%
We make use of the following constructions on measures.
\begin{iteMize}{$\bullet$}
\item Given a function $f:t\to u$ and a measure
  $\mu\in\DIST{t}$, there is a measure $\mu f^{-1}\in\DIST{u}$ given by $(\mu
  f^{-1})(B)\deq\mu(f^{-1}(B))$.
\item Given a finite measure $\mu$ and a measurable set $B$, we let
  $\mu\rvert_B(A)\deq\mu(A\cap B)$ be the restriction of $\mu$ to $B$.
\item We can add two measures on the same set as $(\mu_1+\mu_2)(A) \deq
  \mu_1(A)+\mu_2(A)$.
  \begin{FULL}%
  \item We can multiply a measure by a positive constant as $(r\cdot\mu)(A)
    \deq r\cdot\mu(A)$.
  \end{FULL}%
\item The (independent) product ($\mu_1\times\mu_2$) of two ($\sigma$-finite) measures is
  also definable~\cite[Sec. 18]{billingsley95}, and satisfies $(\mu_1\times\mu_2)(A\times B) =
  \mu_1(A)\cdot\mu_2(B)$. 
\item If $\mu_i$ is a measure on $t_i$ for $i\in\Set{1,2}$, we let the disjoint sum 
  $\mu_1\oplus\mu_2$ be the measure on $t_1+t_2$ defined as $A_1\uplus A_2\mapsto\mu_1(A_1)+\mu_2(A_2)$.
\item
Given a measure $\mu$ on the measurable space $\tqq{t}$,
a measurable set $A \in \mAlg_t$ and a function $f : t \to \realT$,
we write $\Lebesgue{A}{f}{\mu}$
or equivalently $\Lebesgue{A}{f(x)}{\mu(x)}$
for standard (Lebesgue) integration.
This integration is always well-defined if $\mu$ is finite and $f$ is
non-negative and bounded from above.
\item Given $t$, we let $\lambda_t$ be the ``standard'' measure on
  $\tqq{t}$ built from independent products and disjoint sums of the
  Lebesgue measure on $\realT$ and the counting measure on discrete
  $b$. We often omit $t$ when it is clear from the context.  (We also
  use $\lambda$-notation for functions, but we trust any ambiguity is
  easily resolved.)
\item 
Given a measure $\mu$ on a measurable space $\tqq{t}$
we call a function $\dot\mu : t \to \realT$ a \emph{density} for $\mu$
iff $\mu(A)=\int_A\dot\mu\,d\lambda$ for all $A \in \mAlg$.
\end{iteMize}

\subsubsection*{Standard Distributions}
Given a closed well-typed \fun expression $\Sample{D(V)}$ of base type $b$,
we define a corresponding finite measure $\mu_{D(V)}$ on measurable space~$\tqq{b}$,
via its density $D(V)=\dot\mu_{D(V)}$.
In the discrete case, we first define the probability 
mass function, written $D(V)\ c$, 
and then define the measure $\mu_{D(V)}$ as a summation.
\begin{display}[.4]{Masses $D(V)\ c$ and Measures $\mu_{D(V)}$ for Discrete Probability Distributions:}
 \clause{\kw{Bernoulli}(p)\ \kw{true} \deq p}
  {if $0\le p\le 1$, 0 otherwise}\\
  \clause{\kw{Bernoulli}(p)\ \kw{false} \deq 1-p}
  {if $0\le p\le 1$, 0 otherwise}\\
 \clause{\kw{Binomial}(n,p)\ i \deq {i \choose n} p^i/n!}
  {if $0\le p\le 1$, 0 otherwise}\\
  \clause{\kw{DiscreteUniform}(m)\ i\deq 1/m\qquad}{if $0\le i< m$, $0$ otherwise}\\
  \clause{\kw{Poisson}(l)\ n\deq e^{-l}l^n/n!}{if $l,n\ge0$, $0$ otherwise}\\[\GAP]
  \clause{ \mu_{D(V)}(A) \deq \sum_i D(V)\ c_i}{if $A=\bigcup_i \{c_i\}$ for pairwise disjoint $c_i$}
\end{display}
In the continuous case, we first define the probability density function $D(V)\ r$ 
and then define the measure $\mu_{D(V)}$ as an integral.
Below, we write $\mathbf{G}$ for the standard Gamma function, 
which on naturals $n$ satisfies $\mathbf{G}(n)=(n-1)!$.
\begin{display}[.43]{Densities $D(V)\ r$ and Measures $\mu_{D(V)}$  for Continuous Probability Distributions:}
  \clause{\kw{Gaussian}(m,v)\ r \deq
    \smash{e^{{-(r-m)^2/2v}}}/\sqrt{2\pi v}}
  {if $v>0$, 0 otherwise}\\
  \clause{\kw{Gamma}(s,p)\ r \deq r^{s-1}e^{-pr}p^s/\mathbf{G}(s)}{if $r,s,p>0$, 0 otherwise}\\
  \clause{\kw{Beta}(a,b)\ r \deq r^{a-1}(1-r)^{b-1} \textbf{G}(a+b)/(\textbf{G}(a)\textbf{G}(b))}{}
   {if $a,b>0$ and $0\le r\le 1$, 0 otherwise}\\[\GAP]
   \clause{\mu_{D(V)}(A) \deq \int_A D(V) \,d\lambda }{where $\lambda$ is the Lebesgue measure on $\Real$}
\end{display}
The Dirac $\delta$ measure is defined on the measurable space $\tqq{b}$ for each base type $b$,
and is given by $\delta_c(A)\deq 1$ if $c\in A$, $0$ otherwise.
%

\subsubsection*{Conditional density}
The notion of density can be generalized as follows,
yielding an unnormalized counterpart to conditional probability.
Given a measurable function $p:t\to u$,
we consider two families of events on $t$.
Firstly, 
events
 $E_c\deq\{x\in\Vals{t}\mid p(x)=c\}$ where $c$ ranges over~$\Vals{u}$.
Secondly, rectangles 
$R_d\deq\{x\in\Vals{t}\mid x \le d\}$ 
where $d$ ranges over~$\Vals{t}$
and $\le$ is the coordinate-wise partial order
(that on pair types satisfies $(a,b)\le(c,d)$ iff $a\le c$ and $b\le d$,
that on $\intT$ and $\realT$ is the standard ordering, and that only relates equal booleans).

Given a finite measure $\mu$ on $\tqq{t}$ and $c\in\Vals{u}$,
we let $F_c:t\to\Real$ be defined by the limit below
(following \cite{fraser95:probDens})
\begin{equation}
{F_c(d)} \deq \lim_{i\to\infty}\mu(R_d\cap p^{-1}(B_i))/\lambda_u(B_i)
\label{eq:2}
\end{equation}
if the limit exists and is the same for all sequences $\{B_i\}$ of closed sets converging regularly to~$c$. 
On points $d$ where no unique limit exists, we let 
\[F_c(d)\deq\inf\,\{F_c(d')\mid d\le d'\land d\neq d'\land F_c(d')\text{~defined}\}\]
where we let $\inf\,\emptyset\deq\infty$. If $F_c$ is bounded, we define $\given \mu p \cdot c \in \Real$ (the $\mu$-density at $E_c$)
as the finite measure on $\tqq{t}$ with (unnormalized) cumulative distribution function $F_c$, that is,
$\given \mu p {R_d} c=F_c(d)$.
(If $F_c$ is not bounded, it is not the distribution function of a finite measure.)

As examples of this definition, when $u$ is discrete we have that
$\given \mu p A c = \mu(A\cap\Set{x\mid p(x)=c})$,
so discrete density amounts to filtering.
In the continuous case, if $\Vals t=\Real\times\Real^k$,
$p=\lambda(x,{y}).(x-c)$ and $\mu$ has a continuous density~$\dot\mu$ 
then 
\begin{eqnarray*}
  F_c(a,b)&=&
  \lim_{i\to\infty}\frac{\mu(R_{(a,b)}\cap p^{-1}(B_i))}{\lambda_{\Real}(B_i)}\\
&=&
  \lim_{i\to\infty}\frac{\int_{(R_{(a,b)}\cap p^{-1}(B_i))}\!\dot\mu(x,{y})\, d\lambda_t (x,y)}{\lambda_{\Real}(B_i)}\\
  &=&
  \int_{\Set{{y}\mid (c, y)\in R_{(a,b)}}} \!\!\!\!\!\!\!\!
  \dot\mu(c,{y})\, d\lambda_{\Real^k}(y)
  \quad\text{when $a\neq c$ by continuity.}
\end{eqnarray*}
When $a=c$ the limit may not be unique, in which case we have 
\begin{eqnarray*}
F_c(c,b)&=&\inf\,\{F_c(d')\mid (c,b)\le d'\}\\
&=&
  \int_{\Set{{y}\mid (a, y)\in R_{(a,b)}}} \!\!\!\!\!\!\!\!
  \dot\mu(a,{y})\, d\lambda_{\Real^k}(y)\quad\text{ by monotonicity of $F_c$ and continuity.}
\end{eqnarray*}
We then get
\begin{eqnarray}
\given\mu p A c
&=&
  \int_{\Set{{y}\mid (c, y)\in A}} \!\!\!\!\!\!\!\!
  \dot\mu(c,{y})\, d\lambda_{\Real^k}(y).   \label{eq:3}
\end{eqnarray}
One case when conditional density may not be defined is when the conditioning event is at a discontinuity of the density function: let $t=\realT*\realT$, $p(x,y)=x$, 
and $\dot\mu(x,y)=1$ if $0\le x,y\le 1$, otherwise 0.
Then $F_1(x,y)=0$ if $x<1$ or $y\le 0$, and otherwise the limit (\ref{eq:2}) is not unique.
Thus $F_1(1,0)=\infty$, so $F_1$ is not bounded and $\given \mu p \cdot 1$ is undefined.
For more examples, see Section~\ref{sec:discussion}.

There exists a more declarative approach to $\mathcal{D}\mu$.
For $A\in\mAlg_t$, we let $\nu_A(B) = \mu(A\cap p^{-1}(B))$; 
this measure is said to be \emph{absolutely continuous} (wrt. $\lambda_u$)
if $\nu_A(B)=0$ whenever $\lambda_u(B)=0$.  
If $\mu$ is \emph{outer regular}, 
i.e. $\mu(A)=\inf\{\mu(G)\mid A\subset G, G\text{~open}\}$ for all $A$, 
and $\nu_A$ is absolutely continuous,
the defining limit (\ref{eq:2}) 
exists \emph{almost everywhere}~\cite{fraser95:probDens},
that is, there is a set $C$ with $\mu(C)=0$ such that $c\in C$ 
if $F_c(d)$ is undefined.
Then, $\given \mu p A c$ is a version of the Radon-Nikodym derivative 
of $\nu_A(B)$ (wrt. $\lambda_u$). For all $B\in\mAlg_u$, 
conditional density thus satisfies the equation
\begin{equation}
\mu(A\cap p^{-1}(B))=\int_{B}\!\given\mu p A x \;d\lambda_u(x).
\label{eq:1}
\end{equation}
%
%
%
The existence of a family of finite
measures $\given\mu p\sdot c$ on $\tqq{t}$ satisfying equation
(\ref{eq:1}) above is guaranteed in certain situations, e.g., 
when $\mu p^{-1}$ has density $d$ at $c$ we may take
$\mathcal{D}\mu$ as a version of the regular conditional probability 
$\mu[\cdot\mid p=c]$ (see for instance \cite[Theorem 33.3]{billingsley95}) 
scaled by $d$.
However, if $\mu(p^{-1}(c))=0$ the value of $\given \mu p A c$ may not be uniquely defined, 
since two versions of $\given\mu{p}\sdot\sdot$ may differ on a null set.
    %
%
In order to avoid this ambiguity we have given an explicit
construction that works for many useful cases.


\subsection{Measure Transformers}
We will now recast some standard theorems of measure theory as a library of
combinators, that we will later use to give semantics to probabilistic languages.
A \emph{measure transformer} is a partial function from finite measures to finite measures.
We let $\PArr t u$ be the set of partial functions $\DIST{t}\to\DIST{u}$.  
We use the combinators on measure transformers listed below
in the formal semantics of our languages.
The definitions of these combinators occupy the remainder of this
section.  We recall that $\mu$ denotes a measure and $A$ a
measurable set, of appropriate types.

\pagebreak
\begin{display}{Measure Transformer Combinators:}
\clause{\Aarr\in(t\to u)\to\PArrP{t}{u}}\\
\clause{\Athen[]\in\PArrP{t_1}{t_2}\to\PArrP{t_2}{t_3}\to\PArrP{t_1}{t_3}}\\
\clause{\Achoose[]{}\in (t\to\boolT)\to\PArrP{t}{u}\to\PArrP{t}{u}\to\PArrP{t}{u} }\\
\clause{\Aextend[]\in(t\to\DIST{u})\to\PArrP{t}{(t*u)}}\\
\clause{\Aconstrain[]\in(t\to b)\to\PArrP tt}
\end{display}

\MySubSubsection{Lifting a Function to a Measure Transformer}
To lift a pure measurable function to a measure
transformer, we use the combinator 
$\Aarr\in(t\to u)\to\PArrP{t}{u}$.
Given $f: t\to u$,
we let $\Aarr f\ \mu\ A\deq \mu f^{-1}(A)$,
where $\mu$ is a measure on $\tqq{t}$
and $A$ is a measurable set from $\tqq{u}$ (cf. \cite[Eqn 13.7]{billingsley95}).

\MySubSubsection{Sequential Composition of Measure Transformers}
To sequentially compose two measure transformers
we use standard function composition, defining
$\Athen[]\in\PArrP{t_1}{t_2}\to\PArrP{t_2}{t_3}\to\PArrP{t_1}{t_3}$ as
$T\Athen U \deq U\circ T$.

\begin{DRAFT}
\adg{these are not needed for the features of \fun and \imp in the TR}
\MySubSubsection{Independent Products of Measures}
We construct (independent) products of measure transformers using
$\Atuple[]\in \PArrP{t}{u_1}\to\PArrP{t}{u_2}\to \PArrP{t}{(u_1*u_2)}$.
\\ %
We define $(T \Atuple U)(\mu)\deq (1/\abs{\mu})(T(\mu)\times U(\mu))$.

This definition of $\Atuple[]$ violates the arrow
axioms~\cite{DBLP:conf/afp/Hughes04}; the reason for this is that not all
measures are product measures (just as not all random variables are
independent).

\MySubSubsection{Coproducts of Measure Transformers}
  We also define $\Aleft[]\in\PArrP{t_1}{t_2}\to\PArrP{(t_1+u)}{(t_2+u)}$
  as\linebreak
  $\Aleft T\ \mu\ (A\uplus B) \deq \mu(\0\uplus B) +
  T(\lambda s.\mu(s\uplus\0))(A)$.  This definition satisfies the arrowPlus
  axioms~\cite{DBLP:conf/afp/Hughes04}, making $\PArrP[]{}{}$ into what we
  may call a \emph{coarrow}.
\end{DRAFT}

\MySubSubsection{Conditional Choice between two Measure Transformers}

The combinator $\Achoose[]{}\in(t\to\boolT)\to\PArrP{t}{u}\to\PArrP{t}{u}\to\PArrP{t}{u}$
makes a choice between two measure transformers, parametric on a predicate $p$.  
Intuitively, $\Achoose p\ T_{\mathtt{T}}\ T_{\mathtt{F}}\ \mu$ first splits 
$\Vals t$ into two sets depending on whether or not $p$ is true. 
For each equivalence class, we then run the corresponding 
measure transformer on $\mu$ restricted to the class.
Finally, the resulting finite measures are added together, yielding a finite measure.
If $p^{-1}(\mathtt{true})=B$ we let $\Achoose p\ T_{\mathtt{T}}\ T_{\mathtt{F}}\ \mu\ A
= T_{\mathtt{T}}(\mu\rvert_{B})(A) + T_{\mathtt{F}}(\mu\rvert_{\Vals{t}\setminus B})(A)$.

\MySubSubsection{Extending Domain of a Measure}
The combinator $\Aextend[]\in(t\to\DIST{u})\to\PArrP{t}{(t*u)}$
extends the domain of a measure using a function yielding
measures.  It is reminiscent of creating a dependent pair,
since the distribution of the second component depends on the value of the first.
For $\Aextend m$ to be defined, we require that for every $A\in\mAlg_u$, 
the function $f_A \deq \lambda x.m(x)(A)$ is measurable, non-negative and bounded from above.  
In particular, this holds for all $A$ if $m$ is measurable and 
$m(x)$ always is a (sub-)probability distribution,
which is always the case in our semantics for \fun.
We let $\Aextend m\ \mu\ AB\deq\int_{\Vals t}\!m(x)(\Set{y\mid (x,y)\in AB})d\mu(x)$, 
where we integrate over the first component (call it $x$) with respect to the measure $\mu$, 
and the integrand is the measure under~$m(x)$ of the set $\{ y \mid (x,y) \in AB \}$ for each $x$
(cf. \cite[Ex. 18.20]{billingsley95}).

\MySubSubsection{Observation as a Measure Transformer}
The combinator $\Aconstrain[]\in(t\to b)\to\PArrP tt$ conditions
a measure over $\tqq{t}$ on the event that an indicator function
of type $t \to b$ is zero.
Here observation is \emph{unnormalized} conditioning of a measure on an event.
If defined, we let 
$
\Aconstrain p\ \mu\ A\deq { \given\mu p A {0_b} }
$.
As an example, if $p:t\to\boolT$ is a (measurable) predicate on values of
type~$t$, we have
$\Aconstrain p\ \mu\ A= \mu(A\cap\Set{x\mid p(x)=\kw{true}})$.
Notice that $\Aconstrain p\ \mu\ A$ can be greater than $\mu(A)$
when $p:t\to\realT$ (cf.~the naive Bayesian classifier on page~\pageref{NaiveBayes}), 
for which reason we cannot restrict ourselves to
(sub-)probability measures. 
For examples, see Equation~(\ref{eq:3}) and Section~\ref{sec:discussion}.
\subsection{Measure Transformer Semantics of {\fun}}
In order to give a compositional denotational semantics of \fun programs,
we give a semantics to open programs, later to be placed in some closing
context.  
Since observations change the distributions of program variables, we 
may draw a parallel to while programs.  
There, a program can be given a denotation
as a function from variable valuations to a return value and a
variable valuation.
Similarly, we give semantics to an open \fun term by mapping a measure
over assignments to the term's free variables to a joint measure of
the term's return value and assignments to its free variables.
This choice is a generalization of the (discrete) semantics of
p\textsc{While}~\cite{DBLP:conf/popl/BartheGB09}.
\begin{FULL}%
  This contrasts with \citeT{Ramsey and Pfeffer}{DBLP:conf/popl/RamseyP02}, where the
  semantics of an open program takes a variable valuation and returns a
  (monadic computation yielding a) distribution of return values.
\end{FULL}%

First, we define a data structure for
an evaluation environment assigning values to variable names, and
corresponding operations.
Given an environment $\Gamma=x_1\tty t_1,\dots,x_n\tty t_n$, we let
$\State$ be the set of states, or finite maps $s=\Set{x_1\mapsto
  c_1,\dots,x_n\mapsto c_n}$ such that for all $i=1,\dots,n$,
$\typeOf{c_i}=t_i$.  We let
$\tqq{\State}\deq\tqq{\unitT*t_1*\dots*t_n}$ be the measurable space of states in $\State$.
We define $\dom(s)\deq\Set{x_1,\dots,x_n}$.
We define the following operators.
\begin{display}{Auxiliary Operations on States and Pairs:}
  \clause{\Padd x\ (s,c)\deq s\cup\Set{x\mapsto c}}{if $\typeOf{c}=t$ and $x\nin\dom(s)$, $s$ otherwise.}\\
  \clause{\Plookup x\ s\deq s(x)}{if $x\in\dom(s)$, $\Punit$ otherwise.}\\
  \clause{\Pdrop X\ s\deq \Set{(x\mapsto c)\in s \mid x\nin X}
               \qquad \qquad\mathtt{fst}((x,y)) \deq x \qquad \mathtt{snd}((x,y)) \deq y }
\end{display}
We write $s\rvert_X$ for $\Pdrop (\dom(s)\setminus X)\ s$.
We apply these combinators to give a semantics to \fun programs as
measure transformers.  We assume that all bound variables in a program
are different from the free variables and each other.  Below, $\vqq{V}s$
gives the valuation of $V$ in state $s$, and $\aqq{M}$ gives the 
measure transformer denoted by $M$.

\begin{display}[.35]{Measure Transformer Semantics of \fun:}
\clause{\vqq{x}s \deq \Plookup x s} \\
\clause{\vqq{c}s \deq c}\\
\clause{\vqq{(V_1,V_2)}s \deq (\vqq{V_1} s,\vqq{V_2} s)}\\[\GAP]
\clause{\aqq{V} \deq \Aarr \lambda s.(s,\vqq{V}s)}\\
\clause{\aqq{V_1\otimes V_2} \deq \Aarr \lambda s.(s,{\otimes}(\vqq{V_1}s,\vqq{V_2}s))}\\
\clause{\aqq{V.1} \deq \Aarr \lambda s.(s,\mathtt{fst}(\vqq{V}s))}\\
\clause{\aqq{V.2} \deq \Aarr \lambda s.(s,\mathtt{snd}(\vqq{V}s))}\\[\GAP]
\clause{\aqq{\kw{if}\ V\ \kw{then}\ M\ \kw{else}\ N }\deq
 \Achoose (\lambda s.\vqq{V}s)\ \aqq{M}\ \aqq{N}}\\
\clause{\aqq{\Sample{D(V)}} \deq \Aextend \lambda s. \mu_{D(\vqq{V}s)}}\\
\clause{\aqq{\kw{observe}\ V } \deq (\Aconstrain \lambda s.\vqq{V}s) \Athen \Aarr \lambda s.(s,\Punit)}
\\
\clause{
  \begin{prog}
  \aqq{\kw{let}\ x=M\ \kw{in}\ N} \deq {} 
  \aqq{M} \Athen  \Aarr (\Padd x) \Athen
  \aqq{N}\Athen \Aarr \lambda(s,y).((\Pdrop{\Set x}\ s),y)
  \end{prog}
}
\end{display}
A value expression $V$ returns the valuation of $V$ in the current state,
which is left unchanged.  Similarly, binary operations and projections have
a deterministic meaning given the current state.  
An \kw{if} $V$ expression runs the measure transformer given by the
\kw{then} branch on the states where $V$ evaluates true, and the
transformer given by the \kw{else} branch on all other states, using
the combinator $\Achoose[]$.
A primitive distribution $\Sample{D(V)}$ extends the state measure with a
value drawn from the distribution $D$, with parameters $V$ depending on
the current state.
An observation $\kw{observe}\ V$ modifies the current measure by
restricting it to states where $V$ is zero.
It is implemented with the $\Aconstrain[]$ combinator, and it always returns the unit value.
The expression $\kw{let}\ x=M\ \kw{in}\ N$ intuitively first runs $M$
and binds its return value to $x$ using $\Padd[]$.  After running $N$, the
binding is discarded using $\Pdrop[]$.

\begin{lem}
If $s:\State$ and $ \Gamma \vdash V:t$ then $\vqq{V}s \in \Vals {t}$.
\end{lem}

\begin{lem}
If  $\Gamma \vdash M:t$ then $\aqq{M} \in \PArr {\State} {(\State*t)} $.
\end{lem}
%
%
The measure transformer semantics of \fun is hard to use directly,
except in the case of \bfun
where they can be directly implemented: a naive
implementation of $\Dist$ is as a map assigning a probability to each
possible variable valuation.
If there are $N$ variables, each sampled from a Bernoulli distribution,
in the worst case there are $2^N$ paths to be explored in the computation, each of
which corresponds to a variable valuation. 
\begin{FULL}%
Our direct implementation of the
measure transformer semantics, described in the technical report version of our paper~\cite{fun-esop11}, 
explicitly constructs the valuation.
It works fine for small examples but would blow up on large datasets.
\end{FULL}%
In this simple case, the measure
transformer semantics of closed programs also coincides with the sampling
semantics.
\begin{thm}\label{thm:correspondence-discrete}
  Suppose $\emptyEnv \vdash M : t$ for some $M$ 
  in {\bfun}.
  If $\mu = \aqq{M}\ \delta_{\Punit}$ and $\emptyEnv \vdash V : t$ then
  $\Prob[M]{\RV=V \mid \Obs} = \mu(\Set{(\Punit,V)})/\abs\mu$.
\end{thm}
\begin{FULL}%
  \begin{proof}
    We add a construct to give a semantics to open \bfun expressions.
    Let $\Linit(M,\mu)$ stand for $M$ starting in an initial
    probability measure $\mu$ on $\State$.
    Let $\Linit(M,\mu)\to^{p_s}M \SubstSeq{V}x$
    when $s=\Set{x_i\mapsto V_i\mid i = 1..n}\in\State$ and
    $p_s=\mu(\Set{s'\mid s'\rvert_{\fv(M)}=s\rvert_{\fv(M)}})$.
    In particular, if $M$ is closed, then $\Linit(M,\delta_{\Punit})\to^1M$,
    so $\Linit(M,\delta_{\Punit})$ has the same traces as $M$ but for an additional (valid) initial step.

    By induction on the derivation of $\Gamma \vdash M : t$, we prove that
    if $\Gamma \vdash M : t$ and $\emptyEnv \vdash V : t$ 
    and $\mu\in\Dist$, 
    then $\nu(\State\times\Set{V}) =  \Prob[N]{\Obs \cap \RV=V}$ and 
    $\nu(\State\times\Vals t) = \Prob[N]{\Obs}$, where
    $\nu=\aqq{M}\ \mu$ and $N=\Linit(M,\mu)$.  

    Then, for closed $M$ we get $\Prob[M]{\RV=V \mid \Obs} =
    \Prob[M]{\Obs \cap \RV=V}/\Prob[M]{\Obs} =$\\ 
    $\nu(\Set{(\Punit,V)})/\nu(\Set{\Punit}\times\Vals t)$.\QED
  \end{proof}
\end{FULL}%

\subsection{Discussion of the Semantics}\label{sec:discussion}
In this section we discuss some small examples that are 
illustrative of the semantics of the \kw{observe} primitive.
The first example highlights the difference between discrete observations
and observations on continuous types.

The subsequent examples contrast our definition of \kw{observe} 
with some alternative definitions.
The second example deals with the definition of discrete observations,
that is shown to coincide with the filtering semantics of \bfun,
unlike two alternative semantics.  
In the third example, we treat continuous observations, 
showing that distributing an observation into both branches 
of an if statement yields the same result, 
in contrast to an alternative semantics of observations 
as computing (normalized) conditional probability distributions.

In the fourth example, we show an example of model comparison that
depends on the unnormalized nature of observations.
In the fifth example, we show a well-typed \fun program 
with an observation (of a derived random variable) 
that failed to be well-defined in the original semantics of observation.

\subsubsection*{Discrete versus continuous observations. }
As an example to highlight the difference between continuous and discrete observations, 
we first consider the following program, which observes that a 
normally distributed random variable is zero.
The resulting distribution of the return value $\kw{x}$ is a point mass at 0.0,
as expected.  The measure of $\{0.0\}$ in this distribution is 
$\kw{Gaussian}(0.0,1.0)\;0.0 \approx 0.4$. 
\pagebreak
\begin{displaylisting}{Continuous Observation:}
 let x = random (Gaussian(0.0, 1.0)) in let _ = observe x in x
\end{displaylisting} 
The second program instead observes that a Boolean variable is true.
This has zero probability of occurring, 
and since the Boolean type is discrete,
the resulting measure is the zero measure.
\begin{displaylisting}{Discrete Observation:}
 let x = random (Gaussian(0.0, 1.0)) in let b = (x==0.0) in let _ = observe b in x
\end{displaylisting}   
These examples show the need for observations at $\realT$ type, 
as well as at type $\boolT$.
(This also clearly distinguishes \ls$observe$ from assume in assertional programming.)

\subsubsection*{Discrete Observations amount to filtering.}
A consequence of Theorem~\ref{thm:correspondence-discrete}
is that our measure transformer semantics
is a generalization of the sampling semantics for discrete probabilities.
For this theorem to hold, it is critical that \kw{observe} denotes
unnormalized conditioning (filtering).
Otherwise programs that perform observations inside the branches 
of conditional expressions would have undesired semantics.
As the following example shows, the two program fragments
\ls$observe (x=y)$ %
and %
\ls$if x then observe (y=true)$ \ls$else observe (y=false)$ %
would have different measure transformer semantics
although they have the same sampling semantics.

\begin{displaylisting}{Simple Conditional Expression: $M_{\mathrm{if}}$}
 let x = random (Bernoulli(0.5))
 let y = random (Bernoulli(0.1))
 if x then observe (y=true) else observe (y=false)
 y
\end{displaylisting}
\jb{Expand the MT semantics?}
In the sampling semantics,
the two valid runs are when $\kw{x}$ and $\kw{y}$ are both \kw{true} (with probability 0.05),
and both \kw{false} (with probability 0.45),
so we have $\Prob{\kw{true}\mid\Obs}=0.1$ and $\Prob{\kw{false}\mid\Obs}=0.9$.

If, instead of the unnormalized definition $\Aconstrain p\ \mu\ A= \mu(A\cap\Set{x\mid p(x)})$, we had either of the normalizing definitions
$$\Aobserve p\ \mu\ A
=\frac{\mu(A\cap\Set{x\mid p(x)})}{\mu(\Set{x\mid p(x)})}
\qquad\text{or}\qquad
\abs\mu\frac{\mu(A\cap\Set{x\mid p(x)})}{\mu(\Set{x\mid p(x)})}
$$ 
then 
$\aqq{M_{\mathrm{if}}}\ \delta_{\Punit}\ \Set{\kw{true}}=
\aqq{M_{\mathrm{if}}}\ \delta_{\Punit}\ \Set{\kw{false}}$,
which would invalidate the theorem.

Let $M' = M_{\mathrm{if}}$ with $\kw{observe}\ (x=y)$ substituted for the
conditional expression.  With the actual or either of the flawed definitions of
$\Aobserve[]$ we have 
$\aqq{M'}\ \delta_{\Punit}\ \Set{\kw{true}}=
(\aqq{M'}\ \delta_{\Punit}\ \Set{\kw{false}})/9$.

\subsubsection*{Continuous Observations are not normalizing.}
As in the discrete case, continuous observations do not 
renormalize the resulting measure. %
In the program below, the variables \kw{x} and \kw{y} are independent: 
observing \kw{x} at a given value amounts to scaling the measure of \kw{y}
by some fixed amount.
\begin{displaylisting}{Simple Continuous Observation: $M_{\mathrm{obs}}$}
 let x = random (Gaussian(0.0, 1.0))
 let y = random (Gaussian(0.0, 1.0))
 observe (x-1.0)
 y
\end{displaylisting}
\pagebreak

The resulting distribution $\mu_y$ of \kw{y} is the normal distribution, 
scaled by a factor $\kw{Gaussian(0.0,1.0)}\ 1.0\approx0.24$.
In particular, $\mu_y(\{\kw{y}\in\Real:\kw{y}>-1\})/\abs{\mu_y}\approx0.16$.
Below, we let $\nu$ be the joint distribution of \kw{x} and \kw{y} before the observation.

If we replace the observation by an if statement that performs 
the same observation in each branch, the resulting distribution is unchanged.
Let $M' = M_{\mathrm{obs}}$ with the conditional expression 
$N:=$\kw{if x+y>0 then observe (x-1.0)} \kw{else observe (x-1.0)} 
substituted for \kw{observe (x-1.0)}.
Let $A=\{(x,y)\in\Real^2:x+y>0\}$ and $B= \Real^2\setminus A$.
We have $\aqq{N}\nu=\Achoose p\ T\ T\ \nu = T(\nu|_A) + T(\nu|_B)$ 
where $p=\lambda x,y. (x+y>0)$ and 
$T=\Aconstrain\lambda x,\_.(x-1)$.
Since the definition of $\Aconstrain\lambda x,\_.(x-1)\mu=\given\mu x \cdot 1$ is linear in $\mu$ (where defined) and $\nu=\nu|_A+\nu|_B$,
we have $\aqq{M_{\mathrm{obs}}}=\aqq{M'}$.

However, if observations always yielded probability distributions, 
and \kw{if} statements reweighted the result of each branch 
by the probability that that branch was taken, 
the above equality would not hold.
In $M'$, the branch condition \kw{x+y>0} is \kw{true} with probability 0.5 a priori.
This reweighting semantics would 
after the observation of $\kw{x=1}$ 
give the same probability 
to $\kw{1+y>0}$ (the left branch being taken) 
and $\kw{1+y<0}$ (the right branch being taken).
In contrast, the original program $M_{\mathrm{obs}}$ 
yields $\Prob{\kw{1+y<0}}\approx0.16$.

\begin{FULL}%
\subsubsection*{Medical trial.}
  As another example, let us consider a simple 
  Bayesian evaluation of a medical trial~\cite{infer.net}.
  We assume a trial group of \ls$nTrial$ persons, of which \ls$cTrial$ were healthy at the end of the trial,
  and a control group of \ls$nControl$ persons, of which \ls$cControl$ were healthy at the end of the trial.
  Below, \ls$Beta(1.0,1.0)$ is the uniform distribution on the interval $[0.0, 1.0]$.
  We return the posterior distributions of the likelihood that 
  a member of the trial group (\kw{pTrial})
  and a member of the control group (\kw{pControl})
  is healthy at the end of the trial.

\begin{displaylisting}{Medical Trial:}
 let medicalTrial nTrial nControl cTrial cControl = 
   let pTrial = random(Beta(1.0,1.0))
   observe (cTrial == random (Binomial(nTrial,pTrial))); 
   let pControl = random(Beta(1.0,1.0))
   observe (cControl == random (Binomial(nControl,pControl)));
   pTrial, pControl
\end{displaylisting}

We can then compare this model to one where the treatment is ineffective, 
that is, where the members of the trial group and the control group 
have the same probability of becoming healthy.
Also here we give a uniform prior to the probability that the treatment is effective;
the posterior distribution of this variable will depend on the Bayesian evidence 
for the different models, that is, the ratio between the probabilities 
of the observed outcome in the two models.
This way of performing model comparison critically depends on the unnormalized nature of discrete observations as filtering.

\begin{displaylisting}{Model Selection:}
 let modelSelection nTrial nControl cTrial cControl = 
   let pEffective = random(Beta(1.0,1.0))
   if random(Bernoulli(pEffective)) then 
     medicalTrial nTrial nControl cTrial cControl
     ()
   else 
     let pAll = random(Beta(1.0,1.0))
     observe (cTrial   == random (Binomial(nTrial,pAll)))
     observe (cControl == random (Binomial(nControl,pAll)))
   pEffective
\end{displaylisting} 

\subsubsection*{Observation of Derived Variable}
\label{sec:failing-observations}
The following example, due to Chung-Chieh Shan, 
highlighted regularity problems with 
our original definition of observation~\cite{fun-esop11}.
\begin{displaylisting}{Observation of Derived Variable:}
 let x = random (Beta(1.0, 1.0)) in let y = x - 0.5 in observe y; x.
\end{displaylisting}   
Intuitively, this program should yield a point mass at \kw{x=0.5}, \kw{y=0}.
In our semantics, if $\mu$ is the measure before the observation (when starting from $\delta_{\Punit}$) we have
\begin{eqnarray*}
  F_0(x,y) &=& 1\text{~if~} x > 0.5\text{~and~}y > 0\\
  F_0(x,y) &=& 0\text{~if~} x < 0.5\text{~or~}y < 0 
\end{eqnarray*}
Otherwise, we have $F_0(x,y)=\inf\{ F_0(x',y')\mid x'\ge x \land y'\ge y \} = 1$ so
$\given\mu y{A}0 = 1$ iff $(0.5,0)\in A$ and otherwise 0; in particular we have
$\given\mu y{      x = 0.5    }0=1$.

The original definition of observation simply applied the limit of 
Equation~(\ref{eq:2}) to any $A$ (not only to rectangles $R_d$). 
Then the density of any null set would be 0, and in particular we would have
$\given\mu y{      x = 0.5    }0=0$. 
This would contradict countable additivity, 
since $\abs{\given\mu y{\cdot}0}=1$ but
$\given\mu y{x_1 < \abs{x-0.5}\le x_2}0=0$ when $0<x_1<x_2$.


\end{FULL}%

\section{Semantics by compilation to \csoft}
\label{sec:semantics-fg}
A naive implementation of the measure transformer semantics of the
previous section would work directly with measures of states, whose size
even in the discrete case could be exponential in the number of variables in scope.
For large models, this becomes intractable.
In this section, we instead give a semantics to \fun programs
%
%
%
by translation to the simple imperative language \imp.
We consider \imp to be a sublanguage of \csoft;
the \csoft program is then evaluated by Infer.NET by constructing a suitable factor graph~\cite{DBLP:journals/tit/KschischangFL01},
whose size will be linear in the size of the program.
The implementation advantage of translating \fsharp to \csoft, over simply
generating factor graphs directly \citep{FACTORIE}, is that
the translation preserves the structure of the input model
(including array processing in our full language),
which can be exploited by the various inference algorithms supported by Infer.NET.
%

\begin{HIDE}
  Our strategy is to first compile \fun to a lower-level imperative
  language \imp (a subset of \csoft), and then to compile \imp to
  factor graphs.  We introduce syntax for \imp and for factor graphs,
  and define the semantics of both using the measure transformer
  semantics of \Sref{sec:coarrows}.
  Theorem~\ref{thm:fun-to-imp-correct} and
  Theorem~\ref{thm:imp-to-fg-correct} establish the correctness of the
  first step, from \fun to \imp, and the second step, from \imp to
  factor graphs.
\end{HIDE}

\subsection{\imp: An Imperative Core Calculus}
\label{sec:core}
\imp is an imperative language, based on the static single assignment (SSA)
intermediate form.  It is a sublanguage of \csoft, the input language of
Infer.NET~\cite{infer.net}.
%
A composite statement $C$ is a sequence of statements,
each of which either stores the result of a primitive operation in a location,
observes the contents of a location to be zero,
or branches on the value of a location.
\imp shares the base types $b$ with \fun, but has no tuples.

\adg{NB removed from $()$ from $L$; as per section 1, $()$ is a constant $c$}
\newcommand{\Local}[1]{\kw{local}\ #1\ \kw{in}\ }
\begin{display}[.35]{Syntax of {\imp}:}
\clause{l,l',\ldots}{location (variable) in global store}\\
\clause{E,F ::= c \mid l \mid (l \otimes l)}{expression}\\
\category{I}{statement}\\
\entry{l \assign E}{assignment}\\
\entry{l \sample D(l_1,\ldots,l_n) }{random assignment}\\
\entry{\kw{observe}_b\ l }{observation}\\
\entry{\kw{if}\ l\ \kw{then}\ C_1\ \kw{else}\ C_2 }{conditional}\\
\entry{\Local{l \ty b} C}{local declaration (scope of $l$ is $C$)}\\
\clause{C ::= \kw{nil} \mid I \mid (C;C)}{composite statement}
\end{display}
When making an observation $\kw{observe}_b$, we make explicit the type $b$ of the observed location.
\begin{SHORT}
In the form $\kw{if}\ l\ \kw{then}_{\Sigma_1}\ C_1\
\kw{else}_{\Sigma_2}\ C_2$, the environments $\Sigma_1$ and $\Sigma_2$ declare the local variables assigned by the \kw{then} branch and the \kw{else} branch, respectively.
These annotations simplify type checking and denotational semantics.
\end{SHORT}
\begin{FULL}%
In a local declaration, $\Local{l \ty b} C$, the location $l$ is bound, with scope $C$.
Next, we derive an extended form of $\kw{local}$, which introduces a sequence of local variables.
%
\begin{display}{Extended Form of \kw{local}:}
\clause{ \Local{\Sigma}C \deq \Local{l_1 \ty b_1}\dots\Local{l_n \ty b_n}C \quad \mbox{where $\Sigma =  \emptyEnv, l_1 \ty b_1, \dots, l_n \ty b_n$} }
\end{display}
\end{FULL}%

The typing rules for \imp are standard.
We consider \imp typing environments $\Sigma$ to be a special case of \fun environments $\Gamma$, where
variables (locations) always map to base types.
\begin{FULL}%
If $\Sigma =  \emptyEnv, l_1 \ty b_1, \dots, l_n \ty b_n$,
we say $\Sigma$ is \emph{well-formed} and write $\Sigma \vdash \diamond$ to mean that the locations $l_i$ are pairwise distinct.
The judgment $\Sigma \vdash E : b$ means that the expression $E$ has type $b$ in the environment $\Sigma$.
\end{FULL}%
The judgment $\Sigma \vdash C : \Sigma'$ means that the composite
statement $C$ is well-typed in the initial environment $\Sigma$, yielding
additional bindings $\Sigma'$.
%
\begin{FULL}%
  \begin{display}{Judgments of the {\imp} Type System:}
  \clause{\Sigma \vdash \diamond}{environment $\Sigma$ is well-formed}\\
  \clause{\Sigma \vdash E : b}{in $\Sigma$, expression $E$ has type $b$}\\
  \clause{\Sigma \vdash C : \Sigma'}{given $\Sigma$, statement $C$ assigns to $\Sigma'$}
  \end{display}
  \begin{display}[.35]{Typing Rules for {\imp} Expressions and Commands:}
    \>
    \staterule{Imp Const}{
      \Sigma\vdash\diamond }{ \Sigma \vdash c : \typeOf{c} } \quad

    \staterule{Imp Loc}{\Sigma\vdash\diamond  \quad (l\tty b) \in \Sigma }{ \Sigma \vdash l : b } \quad


    \staterule{Imp Op}{ \Sigma \vdash l_1 : b_1 \quad \Sigma \vdash  l_2 : b_2 \quad \otimes : b_1,b_2 \to b_3 }{ \Sigma \vdash l_1 \otimes l_2 : b_3 }\\[\GAP]\>

    \staterule{Imp Assign}{ \Sigma \vdash E : b \qquad l \nin
      \dom(\Sigma) }{ \Sigma \vdash l \assign E : (\emptyEnv,l\tty b) } \quad

    \staterule{Imp Random}{
      D : (x_1 \ty b_1, \dots, x_n \ty b_n) \to b \qquad l \nin \dom(\Sigma) \\
     \Sigma \vdash l_1 : b_1 \quad \cdots \quad \Sigma \vdash l_n :
      b_n }{ \Sigma \vdash l \sample D(l_1,\ldots,l_n) : (\emptyEnv,l\tty b) 
    } \\[\GAP]\>

    \staterule{Imp Observe}{ \Sigma \vdash l : b }{ \Sigma \vdash
      \kw{observe}_b\ l : \emptyEnv } \quad

    \staterule{Imp Seq}{ \Sigma \vdash C_1:\Sigma' \qquad
      \Sigma,\Sigma' \vdash C_2:\Sigma'' }{ \Sigma \vdash
      C_1;C_2:\Sigma',\Sigma'' } \quad

    \staterule{Imp Nil}{ \Sigma \vdash \diamond }{ \Sigma \vdash
      \kw{nil}:\emptyEnv } \\[\GAP]\>

   \staterule{Imp If}{
  \Sigma \vdash l : \boolT 
  \quad \Sigma \vdash C_1: \Sigma'
  \quad \Sigma \vdash C_2: \Sigma'
  }{
  \Sigma \vdash \kw{if}\ l\ \kw{then}\ C_1\ \kw{else}\ C_2: \Sigma'
  } \quad

  \staterule{Imp Local}{
  \Sigma \vdash C : \Sigma' \quad (l \ty b) \in \Sigma'
  }{
  \Sigma \vdash \Local{l \ty b}C : (\Sigma' \setminus \Set{l \ty b})
  }
  \end{display}

\ADG{2013 rule Imp Local forces us to use the local variable; is that necessary, as the referee asks?  may bite us in auto-generated uses of local.}
\JB{In our auto-generated uses, we indeed always use the variable, but I don't think it is necessary as such.}
To treat sequences of local variables, we use the \emph{shuffle product} $\Sigma_1 + \Sigma_2$ of two environments,
defined below.
  \begin{display}{Typing Rule for Extended Form of \kw{local}:}
  \>\staterule{Sh Emp}{} {\emptyEnv\in \emptyEnv + \emptyEnv}\quad\;
    
  \staterule{Sh Left}{\Sigma \in \Sigma_1 + \Sigma_2\quad 
    \Sigma,x:b\vdash\diamond} 
  {(\Sigma,x:b) \in (\Sigma_1,x:b) + \Sigma_2}\quad\;
  
  \staterule{Sh Right}{\Sigma \in \Sigma_1 + \Sigma_2\quad 
    \Sigma,x:b\vdash\diamond} 
  {(\Sigma,x:b) \in \Sigma_1 + (\Sigma_2,x:b)}
  \quad\;
  \staterule{Imp Locals}{
  \begin{prog}
    \Sigma \vdash C : \Sigma'_1 \\
    \Sigma'_1 \in \Sigma_1 + \Sigma'
  \end{prog}
  }{
  \Sigma \vdash \Local{\Sigma_1}C : \Sigma'
  }
  \end{display}
\end{FULL}%
\begin{SHORT}
  \begin{display}[.35]{Part of the Type System for {\imp}: $\Sigma \vdash C : \Sigma'$}
\>
\staterule{Imp Seq}{
\Sigma \vdash C_1:\Sigma' \qquad \Sigma,\Sigma' \vdash C_2:\Sigma''
}{
\Sigma \vdash C_1;C_2: (\Sigma',\Sigma'')
}\qquad

\staterule{Imp Nil}{
\Sigma \vdash \diamond
}{
\Sigma \vdash \kw{nil}:\emptyEnv
}\qquad

\staterule{Imp Assign}{
\Sigma \vdash E : b \qquad l \nin \dom(\Sigma)
}{
\Sigma \vdash l \assign E : \emptyEnv,l\tty b
} \\[\GAP]\>

\staterule{Imp Observe}{
\Sigma \vdash l : b
}{
\Sigma \vdash \kw{observe}_b\ l : \emptyEnv
} 
\quad
\staterule{Imp If Locals}{
  \Sigma \vdash l : \boolT 
  \quad \Sigma \vdash C_1: \Sigma'_1
  \quad \Sigma \vdash C_2: \Sigma'_2
  \quad \Set{\Sigma'_i}=\Set{\Sigma_i,\Sigma'}
  }{
\Sigma \vdash \kw{if}\ l\ \kw{then}_{\Sigma_1}\ C_1\ \kw{else}_{\Sigma_2}\ C_2: \Sigma'
}
\end{display}
\end{SHORT}
\begin{FULL}%

\begin{lem}~
\begin{enumerate}
\item If $\Sigma,\Sigma' \vdash \diamond$ then $\dom(\Sigma) \cap \dom(\Sigma') = \emptyset$.
\item If $\Sigma \vdash E : b$ then $\Sigma \vdash \diamond$ and $\fv(E) \subseteq \dom(\Sigma)$.
\item If $\Sigma \vdash C : \Sigma'$ then $\Sigma,\Sigma' \vdash \diamond$.
\end{enumerate}
\end{lem}
\subsection{Measure Transformer Semantics of {\imp}}
A compound statement $C$ in \imp has a semantics as a
measure transformer $\impdt{C}$ generated from the set of combinators defined in
Section~\ref{sec:coarrows}.
An \imp program does not return a value, but
is solely a measure transformer on states
$\PArr{\State[\Sigma]} \State[\Sigma,\Sigma']$ (where
$\Sigma$ is a special case of $\Gamma$).
\begin{display}[.5]{Interpretation of Statements: $\impdt{C},\impdt{I}:\PArr{\State[\Sigma]} \State[\Sigma,\Sigma']$}
\clause{\impdt{\kw{nil}}\deq \Aarr \mathtt{id}}\\
\clause{\impdt{C_1;C_2}\deq \impdt{C_1} \Athen \impdt{C_2}}
\\[\GAP]
\clause{\impdt{l \assign c}\deq \Aarr \lambda s.\Padd l\ (s,c)}\\
\clause{\impdt{l \assign l'}\deq \Aarr \lambda s.\Padd l\ (s,\Plookup{l'}s)}\\
\clause{\impdt{l \assign l_1\otimes l_2}\deq \Aarr 
  \lambda s.\Padd l\ (s, {\otimes}(\Plookup{l_1}s, \Plookup{l_2}s)))}\\
\clause{\impdt{l \sample D(l_1,\ldots,l_n)}\deq \Aextend
  (\lambda s. \mu_{D(\Plookup{l_1}s,\ldots, \Plookup{l_n}s)})
  \Athen \Aarr(\Padd l)}\\
\clause{\impdt{\kw{observe}_b\ l}\deq 
  \Aconstrain \lambda s.\Plookup ls} \\
\clause{ \impdt{\kw{if}\ l\ \kw{then}\ C_1\ \kw{else}\ C_2 }\deq \Achoose(\lambda s.\Plookup ls)\ \impdt{C_1}\ \impdt{C_2} }\\
\clause{ \impdt{\Local{l \ty b}C} \deq \impdt{C} \Athen \Aarr (\Pdrop\ \Set l) }
\end{display}
\begin{lem}
If  $\Sigma \vdash C:\Sigma'$ then $\aqq{M} \in \PArr {\State[\Sigma]} {\State[\Sigma,\Sigma']}$.
\end{lem}

\begin{display}{Semantics of Extended Form of \kw{local}:}
\clause{\impdt{\Local\Sigma C} \deq \impdt{C} \Athen \Aarr (\Pdrop(\dom(\Sigma))) }
\end{display}

\end{FULL}%
\subsection{Translating from {\fun} to {\imp}}
The translation from \fun to \imp is a mostly routine compilation of functional code to imperative code.
The main point of interest is that \imp locations only hold values of base type, while \fun variables may hold tuples.
We rely on \emph{patterns} $p$ and \emph{layouts} $\rho$ to track the \imp locations corresponding to \fun environments.
\begin{display}[.15]{Notations for the Translation from {\fun} to {\imp}:}
\clause{ p ::= l\mid()\mid(p,p) }{pattern: group of \imp locations to represent \fun value}\\
\clause{ \rho ::= (x_i \mapsto p_i)^{i \in 1..n} }{layout: finite map from \fun variables to patterns}\\
\clause{ \Sigma \vdash p : t }{ in environment $\Sigma$, pattern $p$ represents \fun value of type $t$}\\
\clause{ \Sigma \vdash \rho : \Gamma }{in environment $\Sigma$, layout $\rho$ represents environment $\Gamma$}\\
\clause{ \rho \vdash M \Rightarrow C,p }{given $\rho$, expression $M$ translates to $C$ and pattern $p$}
\end{display}
\begin{FULL}%
\begin{display}[.35]{Typing Rules for Patterns $\Sigma \vdash p : t$ and Layouts $\Sigma \vdash \rho : \Gamma$:}
\ruleclause
\staterule{Pat Loc}{
  \begin{prog}
  \Sigma \vdash \diamond \\ (l:t) \in \Sigma
  \end{prog}
}{
\Sigma \vdash l : t
} \quad

\staterule{Pat Unit}{ \Sigma \vdash \diamond
}{
\Sigma \vdash () : \unitT
} \quad

\staterule{Pat Pair}{
\begin{prog}
\Sigma \vdash p_1 : t_1 \\ \Sigma \vdash p_2 : t_2
\end{prog}
}{
\Sigma \vdash (p_1,p_2) : t_1 * t_2
} \quad

\staterule{Layout}{
  \begin{prog}
    \locs(\rho)=\dom(\Sigma)  \\
    \Sigma \vdash \diamond \quad \dom(\rho)=\dom(\Gamma) \\
    \Sigma \vdash \rho(x) : t \quad \forall (x:t) \in \Gamma
  \end{prog}
}{
  \Sigma \vdash \rho : \Gamma 
}
\end{display}
The rule \ref{Pat Loc} represents values of base type by a single location.
The rules \ref{Pat Unit} and \ref{Pat Pair} represent products by a pattern for their corresponding components.
The rule \ref{Layout} asks that each entry in $\Gamma$ is assigned a pattern of suitable type by layout $\rho$.

The translation rules below depend on some additional notations.
We say $p \in \Sigma$ if every location in $p$ is in $\Sigma$.
Let $\locs(\rho) = \bigcup \{\fv(\rho(x)) \mid x \in \dom(\rho)\}$,
and let $\locs(C)$ be the environment listing the set of locations assigned by a command $C$.
\begin{display}[.35]{Rules for Translation: $p\sim p'$ and $p \assign p'$ and $p \vdash M \Rightarrow C,p$}
\> $\kw{()}\sim\kw{()}\qquad\qquad\qquad l\sim l'\qquad\qquad\qquad 
p_1\sim p_1'\land p_2\sim p_2'\implies(p_1,p_2)\sim(p_1',p_2')$\\[\GAP]
\> $\kw{()}\assign\kw{()}\deq\kw{nil}\qquad\qquad\qquad\qquad\qquad
(p_1,p_2)\assign(p_1',p_2')\deq p_1\assign p_1';p_2\assign p_2'$\\[\GAP]

\>\staterule{Trans Var}{
}{
\rho \vdash x \Rightarrow \kw{nil}, \rho(x)
} \quad

\staterule{Trans Const}{
c \neq \kw{()} \quad l \notin \locs(\rho)
}{
\rho \vdash c \Rightarrow (l \assign c), l
} \quad

\staterule{Trans Unit}{
}{
\rho \vdash \kw{()} \Rightarrow \kw{nil}, \kw{()}
} \\[\GAP]\>

\staterule{Trans Operator}{
\rho \vdash V_1 \Rightarrow C_1,l_1 \qquad \rho \vdash V_2 \Rightarrow C_2,l_2 \\
l \notin \locs(\rho) \cup \locs(C_1) \cup \locs(C_2) \qquad
\locs(C_1) \cap \locs(C_2) = \emptyset
}{
\rho \vdash V_1 \otimes V_2 \Rightarrow (C_1;C_2;l \assign l_1 \otimes l_2), l
}
 \\[\GAP]\>

\staterule{Trans Pair}{
\rho \vdash V_1 \Rightarrow C_1,p_1 \quad
\rho \vdash V_2 \Rightarrow C_2,p_2 \quad \locs(C_1) \cap \locs(C_2) = \emptyset
}{
\rho \vdash (V_1,V_2) \Rightarrow (C_1; C_2), (p_1,p_2)
} \\[\GAP]\>

\staterule{Trans Proj1}{
\rho \vdash V \Rightarrow C,(p_1,p_2)
}{
\rho \vdash V.1 \Rightarrow C,p_1
} \quad

\staterule{Trans Proj2}{
\rho \vdash V \Rightarrow C,(p_1,p_2)
}{
\rho \vdash V.2 \Rightarrow C,p_2
} \\[\GAP]\>

\staterule{Trans If}{
\begin{prog}
\rho \vdash V_1 \Rightarrow C_1,l \quad (\locs(\rho) \cup \locs(C_1) \cup \locs(C_2) \cup \locs(C_3)) \cap \fv(p) = \emptyset \\
\rho \vdash M_2 \Rightarrow C_2,p_2 \quad C'_2 = \Local{\locs(C_2)} (C_2;  p \assign p_2) \quad p_2 \sim p \\
\rho \vdash M_3 \Rightarrow C_3,p_3 \quad C'_3 = \Local{\locs(C_3)} (C_3;  p \assign p_3) \quad p_3 \sim p
\end{prog}
}{
\rho \vdash (\Lif{V_1}{M_2}{M_3}) \Rightarrow (C_1; \kw{if}\ l\ \kw{then}\ C'_2\ \kw{else}\ C'_3), p
} \\[\GAP]\>

\staterule{Trans Observe}{
\rho \vdash V \Rightarrow C,l \quad \mbox{$b$ is the type of $V$}
}{
\rho \vdash \kw{observe}\ V \Rightarrow (C;\kw{observe}_b\ l), \kw{()}
} \quad

\staterule{Trans Random}{
\rho \vdash V \Rightarrow C,p \quad l \notin \locs(\rho) \cup \locs(C)
}{
\rho \vdash \Sample{D(V)} \Rightarrow (C; l \sample D(p)), l
} \\[\GAP]\>

\staterule{Trans Let}{
\rho \vdash M_1 \Rightarrow C_1,p_1 \quad
x \notin \dom(\rho) \quad
\rho\{x\mapsto p_1\} \vdash M_2 \Rightarrow C_2,p_2 \quad
}{
\rho \vdash \kw{let}\ x = M_1\ \kw{in}\ M_2 \Rightarrow 
(\Local{(\locs(C_1)\setminus\fv(p_1))} C_1); C_2, p_2
}
\end{display}
In general, a \fun term $M$ translates under a layout $\rho$ to a series of commands $C$ and a pattern $p$.
The commands $C$ mutate the global store so that the locations in $p$ correspond to the value that $M$ returns.
The simplest example of this is in \ref{Trans Const}: the constant expression $c$ translates to an \imp program that writes $c$ into a fresh location $l$.
The pattern that represents this return value is $l$ itself.
The \ref{Trans Var} and \ref{Trans Unit} rules are similar.
In both rules, no commands are run.
For variables, we look up the pattern in the layout $\rho$; for unit, we return the unit location.
Translation of pairs \ref{Trans Pair} builds each of the constituent values and constructs a pair pattern.

More interesting are the projection operators.
Consider \ref{Trans Proj1}; the second projection is translated similarly by \ref{Trans Proj2}.
To find $V.1$, we run the commands to generate $V$, which we know must return a pair pattern $(p_1,p_2)$.
To extract the first element of this pair, we simply need to return $p_1$.
Not only would it not be easy to isolate and run only the commands to generate the values that go in $p_1$, it would be incorrect to do so.
For example, the \fun expressions constructing the second element of $V$ may observe values, and hence have non-local effects.

The translation for conditionals \ref{Trans If} is somewhat subtle.
First, we run the translated branch condition.
The return value of the translated branches is reassigned to a pattern $p$ of fresh locations:
using a shared output pattern allows us to avoid the $\phi$ nodes common in SSA compilers.
We use the \imp derived form where the local variables of the 
\kw{then} and \kw{else} branches of the conditional are restricted.
Instead, both branches write to a fresh shared target $p$,
in order to preserve well-typedness (Proposition~\ref{prop:fun-to-imp-static}).
\adg{We decided to let this slip, but strictly speaking we need to reconstruct the \fun types for these variables, perhaps by including types in the layout.}

The rule \ref{Trans Observe} translates \kw{observe} by running the commands to generate the value for $V$ and then observing the pattern.
(This pattern $l$ can only be a location, and not of the form $\kw{()}$ or $(p_1,p_2)$, as observations are only possible on values of base type.)

The rule \ref{Trans Random} translates random sampling in much the same way.
By $D(p)$, we mean the flattening of $p$ into a list of locations and passing it to the distribution constructor $D$.

Finally, the rule \ref{Trans Let} translates \kw{let} statements by running both expressions in sequence.
We translate $M_2$, the body of the let, with an extended layout, so that $C_2$ knows where to find the values written by $C_1$, in the pattern $p_1$.
Here the local variables of the let-bound expression are restricted using \kw{local}.
\begin{prop}\label{prop:fun-to-imp-static}
Suppose $\Gamma \vdash M :t$ and $\Sigma \vdash \rho : \Gamma$.
\begin{enumerate}
\item There are $C$ and $p$ such that $\rho \vdash M \Rightarrow C,p$.
\item Whenever $\rho \vdash M \Rightarrow C,p$, there is $\Sigma'$ such that $\Sigma \vdash C:\Sigma'$ and $\Sigma,\Sigma' \vdash p : t$.
\end{enumerate}
\end{prop}
  \begin{proof}
    By induction on the typing of $M$ (Appendix~\ref{sec:proof-prop}).\QED
  \end{proof}
We define operations $\metaF{lift}$ and $\metaF{restrict}$ to translate between
\fun variables ($\State[\Gamma]$) and \imp locations ($\State[\Sigma]$).
\[\begin{prog}
\metaF{lift}~\rho \deq \lambda s. \operatorname{flatten}\left\{ \rho(x) \mapsto \vqq{x}s \mid x \in \dom(\rho) \right\} \\
\metaF{restrict}~\rho \deq \lambda s. \left\{ x \mapsto \vqq{\rho(x)}s \mid x \in \dom(\rho) \right\}
\end{prog}\]
We let $\operatorname{flatten}$ take a mapping from patterns to values to a mapping from locations to base values.
Given these notations, we state that the 
compilation of \fun to \imp preserves the measure transformer semantics, 
modulo a pattern $p$ that indicates the locations of the various parts of the return value in the 
typing environment; 
an environment mapping $\rho$, which does the same translation for the initial typing environment;
and superfluous variables, removed by $\metaF{restrict}$.
\begin{thm}\label{thm:fun-to-imp-correct}
If $\Gamma \vdash M :t$ and $\Sigma \vdash \rho : \Gamma$ and $\rho \vdash M \Rightarrow C,p$ then:\\
$\aqq{M} = \Aarr (\metaF{lift}~\rho) \Athen \impdt{C} \Athen \Aarr (\lambda s.~(\metaF{restrict}~\rho~s, \vqq{p}~s)).$
\end{thm}
 \begin{proof}
    By induction on the typing of $M$\ (Appendix~\ref{sec:proof-thms}).\QED
  \end{proof}
\end{FULL}%
\begin{HIDE}
\subsection{Factor Graphs}

A factor graph~\cite{DBLP:journals/tit/KschischangFL01} represents a joint
probability distribution of a set of random variables as a collection of
multiplicative factors.  Factor graphs are an effective means of stating
conditional independence properties between variables, and enable efficient
algebraic inference
techniques~\cite{DBLP:conf/uai/Minka01,DBLP:journals/jmlr/WinnB05} as well
as sampling techniques~\cite[Chapter 12]{koller09:PGM}.  We use factor graphs
with \emph{gates}~\cite{DBLP:conf/nips/MinkaW08} for modelling if-then-else
clauses; gates introduce second-order edges in the graph.
\begin{display}[0.3]{Factor Graphs:}
\smallCategory{G}{\New{\ol{x:b}}\Set{e_1,\dots,e_m}}{graph (variables $\ol{x}$ distinct, bound in $e_1,\dots,e_m$)}\\
\clause{x,y,z,\dots}{nodes (random variables)}\\
\category{e}{edge}\\
\entry{\Fequal{x}{y}}{equality ($x=y$)}\\
\entry{\Fconstant{x}{c}}{constant ($x=c$)}\\
\entry{\Fop{x}{y}{z}}{binary operator ($x=y\otimes z$)}\\
\entry{\Fsample{x}D{y_1,\ldots,y_n}}{sampling ($x\sim D(y_1,\dots,y_n)$)}\\
\entry{\Fgate{x}{G_1}{G_2}}{gate (if $x$ then $G_1$ else $G_2$)}
\end{display}
In a graph $\New{\ol{x:b}}\Set{e_1,\dots,e_m}$, the variables $x_i$ are bound;
graphs are identified up to consistent renaming of bound variables.
We write $\Set{e_1,\dots,e_m}$ for $\New{\emptyEnv}\Set{e_1,\dots,e_m}$.
We write $\fv(G)$ for the variables occurring free in $G$.
\begin{FULL}%
  If $x\nin\Set{\ol{y}}$ we write 
  $\New{x:b'}\New{\ol{y:b}}\Set{e_1,\dots,e_m}$ for the graph $\New{x:b',\ol{y:b}} \Set{e_1,\dots,e_m}$.

As a first example, the coin flipping code in \fun from Section~\ref{sec:fun-syntax} corresponds to the following factor graph:
\begin{display}{Factor Graph for Coin Flipping Example:}
\clause{G_{\mathrm{F}} = \{
  \begin{prog}
    \Fconstant{p}{0.5}, \\
    \Fsample{x}{\kw{Bernoulli}}{p}, \\
    \Fsample{y}{\kw{Bernoulli}}{p}, \\
    \Fop[||]{z}{x}{y},\\
    \Fconstant{z}{\kw{true}} \}
  \end{prog}}
\end{display}
\end{FULL}%
For a second example, we give a factor graph $G_{\mathrm{E}}$
corresponding to the epidemology example of Section~\ref{sec:Epidemiology}
(where $B=\kw{Bernoulli}$).
\begin{display}{Factor Graph for Epidemiology Example:}
\clause{G_{\mathrm{E}}= \{
  \begin{prog}
    \Fconstant{p_d}{0.01},\Fsample{\kw{has_disease}}{B}{p_d}, \\
    \graphF{Gate} (\kw{has_disease}, \\
    \quad \New{p_p : \realT}\{\Fconstant{p_p}{0.8},\Fsample{\kw{positive_result}}{B}{p_p}\}, \\
    \quad \New{p_n : \realT}\{\Fconstant{p_n}{0.096},\Fsample{\kw{positive_result}}{B}{p_n}\}), \\
    \Fconstant{\kw{positive_result}}{\kw{true}}\}
  \end{prog}}
\end{display}
A factor graph typically denotes a probability distribution.
The probability (density) of an assignment of values to variables is
equal to the product of all the factors, averaged over all assignments
to local variables.
Here, we give a slightly more general semantics of factor graphs as
measure transformers; the input measure corresponds to a
prior factor over all variables that it mentions.
Below, we use the Iverson brackets, where $[p]$ is $1$ when $p$ is true and
$0$ otherwise.  We let $\delta(x=y)\deq\delta_0(x-y)$ when $x,y$
denote real numbers, and $[x=y]$ otherwise.  
\begin{display}[.3]{Semantics of Factor Graphs: 
    $\Pqq{G}\in \PArr{\State[\Sigma]}{\State[\Sigma,\Sigma']}$}
\clause{\Pqq{G}\ \mu\ A\deq 
  \int_A\left(\pqq{G}\ s\right)\;d(\mu\times\lambda)(s)}{}\\[\GAP]
\clause{\pqq{\New{\ol{x:b}}\Set{\ol{e}}}s\deq
  \int_{\Vals{*_i b_i}}\prod_{j}(\pqq{e_j}(s,\ol{x}))\,d\lambda(\ol{x})}
\\[\GAP]
\clause{\pqq{\Fequal{l}{l'}}s \deq \delta(\Plookup ls=\Plookup{l'}s)}\\
\clause{\pqq{\Fconstant{l}{c}}s \deq \delta(\Plookup ls=c)}\\
\clause{\pqq{\Fop{l}{w_1}{w_2}}s\deq\delta(\Plookup ls = {\otimes}(\Plookup{w_1}s,\Plookup{w_2}s))}\\
\clause{\pqq{\Fsample{l}D{v_1,\ldots,v_n}}s \deq \mu_{D(\Plookup{v_1}s,\dots,\Plookup{v_n}s)}(\Plookup{l}s)}\\
\clause{\pqq{\Fselect{l}v{y}}s \deq \prod_i{\delta(l=y_i)}^{[v=i]}}\\
\clause{\pqq{\Fgate{v}{G_1}{G_2}}s \deq 
  (\pqq{G_1}s)^{[\Plookup vs]}(\pqq{G_2}s)^{[\neg \Plookup vs]}}
\end{display}
\jb{checked the clause for new; --- changed indexing variable over edges}

\subsection{Factor Graph Semantics for {\imp}}
\label{sec:factor-graph}
An \imp statement has a straightforward semantics as a factor graph.
Here, observation is defined by the value of the variable being the constant $0_b$.
We require that all local variables are declared at top-level, 
i.e.~that in a composite statement $C_1;C_2$, neither $C_1$ nor $C_2$ are of the form
$\kw{local}\ l:b\ \kw{in}\ C$.  Every \imp program can be rewritten in this form 
by pulling local variable declarations to the front.  Moreover, if
 $\rho \vdash M \Rightarrow C,p$ then all local variables of $C$ are declared at top-level.
\jb{noted insertion of singleton sets, to turn everything into a graph}
\begin{display}[.3]{Factor Graph Semantics of \imp: $G = \impfg C$ and $\{e_1,...e_n\} = \impeg C$}
\clause{\impfg{ \Local{l \ty b} C } \deq \New{l \ty b} \impfg{C} }\\
\clause{\impfg{ C } \deq \New{\varepsilon} \impeg{C} }{if $C \neq \kw{local}\ l:b\ \kw{in} C'$. }\\[\GAP]
\clause{\impeg{\kw{nil}}\deq \Femp}\\
\clause{\impeg{C_1;C_2}\deq \impeg{C_1}\cup\impeg{C_2}}\\[\GAP]
\clause{\impeg{l \assign c}\deq \Set{\Fconstant{l}{c}}}\\
\clause{\impeg{l \assign l'}\deq \Set{\Fequal{l}{l'}}}\\
\clause{\impeg{l \assign l_1\ \otimes\ l_2}\deq \Set{\Fop{l}{l_1}{l_2}}}\\
\clause{\impeg{l \sample D(l_1,\ldots,l_n)}\deq \Set{\Fsample{l}D{l_1,\ldots,l_n}} }\\
\clause{\impeg{\kw{observe}_b\ l}\deq \Set{\Fconstrain{l}{b}} }\\
\clause{\impeg{\kw{if}\ l\ \kw{then}\ C_1\ \kw{else}\ C_2 }\deq \Set{\Fgate{l}{\impfg{C_1}}{\impfg{C_2}}}}
\end{display}
\begin{FULL}%
\begin{display}[.3]{Factor Graph Semantics of Derived Forms:}
\clause{\impfg{ \Local\Sigma C } \deq \New{\Sigma} \impfg{C} }\\
\clause{\impfg{\kw{if}\ l\ \kw{then}_{\Sigma_1}\ C_1\ \kw{else}_{\Sigma_2}\ C_2 }\deq
  \Set{\Fgate{l}{\New{\Sigma_1}\impfg{C_1}}{\New{\Sigma_2}\impfg{C_2}}}}
\end{display}
\end{FULL}%
The following theorem asserts that
the two semantics for \imp coincide
for compatible measures, which are defined as follows.
If $T:\PArr tu$ is a measure transformer composed from the combinators
of Section~\ref{sec:coarrows} and $\mu\in\DIST{t}$, we say that $T$ is \emph{compatible} with $\mu$ if
every application of $\Aconstrain f$ to some $\mu'$ in the evaluation of
$T(\mu)$ satisfies either that $f$ is
discrete or that $\mu$ has a continuous density on some $\varepsilon$-neighbourhood of $f^{-1}(0.0)$.  This restriction is needed to ensure that the probabilistic semantics of the factor graph is well-defined.

\begin{SHORT}
The statement of the theorem needs some additional notation.
If $\Sigma \vdash p : t$ and $s \in \State[\Sigma]$, we write $p~s$ for the reconstruction of an element of $\tqq{t}$  by looking up the locations of $p$ in the state $s$.
We define as follows operations $\metaF{lift}$ and $\metaF{restrict}$ to translate between
states consisting of \fun variables ($\State[\Gamma]$) and states consisting of \imp locations ($\State[\Sigma]$),
where $\operatorname{flatten}$ takes a mapping from patterns to values to a mapping from locations to base values.
\[\begin{prog}
\metaF{lift}~\rho \deq \lambda s. \operatorname{flatten}\left\{ \rho(x) \mapsto \vqq{x}s \mid x \in \dom(\rho) \right\} \\
\metaF{restrict}~\rho \deq \lambda s. \left\{ x \mapsto \vqq{\rho(x)}s \mid x \in \dom(\rho) \right\}
\end{prog}\]
\begin{thm}\label{thm:fun-to-imp-correct}
If $\Gamma \vdash M :t$ and $\Sigma \vdash \rho : \Gamma$ and $\rho \vdash M \Rightarrow C,p$
and measure $\mu\in\Dist$ is compatible with $\aqq{M}$ then there exists $\Sigma'$ such that $\Sigma \vdash C : \Sigma'$ and:
\\
$\aqq{M}\ \mu = (\Aarr (\metaF{lift}~\rho) \Athen \Pqq{\impfg{C}}\Athen \Aarr (\lambda s.~(\metaF{restrict}~\rho~s, p~s)))\ \mu$.
\end{thm}
\begin{proof}
  Via a direct measure transformer semantics for \imp. 
  The proof is by induction on the typing judgments 
  $\Gamma \vdash M :t$ and $\Sigma \vdash C:\Sigma'$.\QED
\end{proof}
\ADG{P1 an example to illustrate the theorem would be nice}
\end{SHORT}

\begin{FULL}%
  \begin{thm}\label{thm:imp-to-fg-correct}
    If $\Sigma\vdash C:\Sigma'$ and $\mu\in\Dist[\Sigma]$ is
    compatible with $\impdt{C}$ then $\impdt{C}\ \mu =
    \Pqq{\impfg{C}}\ \mu$.
  \end{thm}
  \begin{proof}
    By induction on the typing of $C$
    (Appendix~\ref{sec:proof-integrals}).\QED
  \end{proof}
\end{FULL}%
\end{HIDE}

\begin{FULL}%
\section{Adding Arrays and Comprehensions}
\label{sec:arrays}
\adg{If we really need names for these, we could use \fun[] and \imp[]}
To be useful for machine learning, our language must support large datasets.
To this end, we extend \fun and \imp with arrays and comprehensions.
We offer three examples, after which we present the formal semantics,
which is based on unrolling.

\subsection{Comprehension Examples in \fun}
\label{sec:listcomps}

Earlier, we tried to estimate the skill levels of three competitors in head-to-head games.
Using comprehensions, we can model skill levels for an arbitrary number of players and games:
\begin{display}{TrueSkill:}
\>
\begin{lstlisting}
let trueskill (players:int[]) (results:(bool*int*int)[]) =
    let skills = [for p in players -> random (Gaussian(10.0,20.0))]
    for (w,p1,p2) in results do
        let perf1 = random (Gaussian(skills.[p1], 1.0))
        let perf2 = random (Gaussian(skills.[p2], 1.0))
        if w // win?
        then observe (perf1 > perf2) // first player won
        else observe (perf1 = perf2) // draw
    skills
\end{lstlisting}
\end{display}
First, we create a prior distribution for each player: 
we assume that skills are normally distributed around 10.0, with variance 20.0.
Then we look at each of the results---this is the comprehension.
The result of the head-to-head matches is an array of triples: a Boolean and two indexes.
If the Boolean is true, then the first index represents the winner and the second represents the loser.
If the Boolean is false, then the match was a draw between the two players.
The probabilistic program walks over the results, and observes that either 
the first player's performance---normally distributed around their skill level---was greater than the second's performance, or
that the two players' performances were equal.
Returning \kw{skills} after these observations allows us to inspect the posterior distributions.
Our original example can be modelled with $\kw{players} = [0;1;2]$ (IDs for Alice, Bob, and Cyd, respectively) and $\kw{results} = [(\kw{true},0,1);(\kw{true},1,2);(\kw{true},0,2)]$.

As another example, we can generalize the simple Bayesian classifier of Section~\ref{sec:coarrows} to arrays of categories and measurements, as follows: 
\mmg{This isn't {\em exactly} what runs, since the current interpreter requires some massaging: we need an Array.toList on measurements.}
\begin{display}{Bayesian Inference Over Arrays:}
\>
\begin{lstlisting}
let trainF (catIds:int[]) (trainData:(int*real)[]) fMean fVariance = 
   let priors = [for cid in catIds -> random (Gaussian(fMean,fVariance))]
   for (cid,m) in trainData do observe (m - random (Gaussian(priors.[cid],1.0)))
   priors
let catIds:int[] = (* ... *)
let trainingData:(int*real)[] = (* ... *)
\end{lstlisting}
\end{display}
The function \kw{trainF} is a probabilistic program for training a naive
Bayesian classifier on a single feature.
Each category of objects---modelled by the array \kw{catIds}---is given a normally distributed prior on the weight of objects in that category;
we store these in the \kw{priors} array.
Then, for each measurement \kw{m} of some object of category \kw{cid} in
the \kw{trainingData} array, we observe that \kw{m} is normally distributed
according to the prior for that category of object. 
We then return the posterior distributions, which have been
appropriately modified by the observed weights.
We can train using this model by issuing a command such as \kw{trainF catIds trainingData 20.0 5.0},
which runs inference to compute for each category its posterior distribution for this feature.

\begin{FULL}
As a third example, consider the adPredictor component of the Bing search engine, which estimates the click-through rates for particular users on advertisements \citep{Graepel2010.Bing}.
We describe a probabilistic program that models (a small part of) adPredictor.
Without loss of generality, we use only two features to make our prediction: the advertiser's listing and the phrase used for searching.
In the real system, many more (undisclosed) features are used for prediction.

\label{disp:adPredictor}
\begin{display}{adPredictor in {\fsharp}:}
\>
\begin{lstlisting}
let read_lines filename count line = (* ... *)
[<RegisterArray>]
let imps = (* ... *)
[<ReflectedDefinition>]
let probit b x = 
    let y = random (Gaussian(x,1.0))
    observe (b == (y > 0.0))
[<ReflectedDefinition>]
let ad_predictor (listings:int[]) (phrases:int[]) impressions =
    let lws = [for l in listings -> random (Gaussian(0.0,0.33))]
    let pws = [for p in phrases  -> random (Gaussian(0.0,0.33))]
    for (clicked,lid,pid) in Array.toList impressions do
        probit clicked (lws.[lid] + pws.[pid])
    lws,pws
\end{lstlisting}
\end{display}
The \kw{read_lines} function loads data from a file on disk.
The data are formatted as newline-separated records of comma-separated values.
There are three important values in each record:
a field that is 1 if the given impression lead to a click, and a 0 otherwise;
a field that is the database ID of the listing shown;
a field that is the part of the search phrase that led to the selection of the listing.
We preprocess the data in three ways, which are elided in the code above.
First, we convert the 1/0-valued Boolean to a \kw{true}/\kw{false}-valued Boolean.
Second, we normalize the listing IDs so that they begin at $0$, that is, so that we can use them as array indexes.
Third, we collect unique phrases and assign them fresh, $0$-based IDs.
We define \kw{imps}---a list of advertising impressions (a listing ID and a phrase ID) and whether or not the ad was clicked---in terms of this processed data.
The \kw{[<RegisterArray>]} attribute on the definition of \kw{imps} instructs the compiler to simply evaluate this \fsharp expression, yielding a deterministic constant.
Finally, \kw{ad_predictor} defines the model.
We use the \kw{[<ReflectedDefinition>]} attribute on \kw{ad_predictor} to mark it as a probabilistic program, which should be compiled and sent to Infer.NET.
Suppose we have stored the collated listing and phrase IDs in \kw{ls} and \kw{ps}, respectively;
we can train on the impressions by calling \kw{ad_predictor ls ps imps}.
\end{FULL}

\subsection{Formalizing Arrays and Comprehensions in \fun}

We introduce syntax for arrays in \fun,
and give interpretations of this extended syntax in terms of the core languages,
essentially by treating arrays as tuples and by unfolding iterations.
We work with non-empty zero-indexed arrays of statically known size
(representing, for example, statically known experimental data).

\jb{clarified that sizes of ranges are defined statically, and that arrays are non-empty zero-indexed}
There are three array operations: array literals, indexing, and array comprehension.  First, let $\Range$ be a set of {\em ranges} $r$.
Ranges allow us to differentiate arrays of different sizes.
Moreover, limitations in the implementation of Infer.NET disallow
nested iterations on the same range.
Here we disallow nested iterations altogether---they are not needed for
our examples and they would significantly complicate the formalization.
We assign sizes to ranges using the function $\abs\cdot : \Range \rightarrow \Int^{+}$.
In the metalanguage, arrays over range $r$ correspond to tuples of length $\abs r$.
\begin{display}[0.35]{Extended Syntax of \fun:}
\smallCategory{t}{\dots \mid {t[r]}}{type}\\
\extendCategory{M,N}{expression}\\
\entry{ [ V_1 ; \ldots ; V_n ] }{array literal} \\
\entry{ V_1 . [ V_2 ]_r }{indexing} \\
\entry{ [\Lforeach{x}{r}{V}{M}] }{comprehension}
\end{display}
First, we add arrays as a type: $t[r]$ is an array of elements of type $t$ over the range $r$.
In the array type $t[r]$, we require that the type $t$ contains no array type $t'[r']$, that is, we do not consider nested arrays.
Indexing, $V_1.[V_2]_r$, extracts elements out of an array, where the
index $V_2$ is computed modulo the size $|r|$ of the array $V_1$.
A comprehension $[\Lforeach{x}{r}{V}{M}]$ maps over an array $V$, 
producing a new array where each element is determined by evaluating $M$ with the corresponding element of array $V$ bound to $x$. 
To simplify the formalization, we here require that the body $M$ of the comprehension contains neither array literals nor comprehensions.
We attach the range to indexing and comprehensions so that the measure transformer
semantics can be given simply; the range can be inferred easily, and
need not be written by the programmer.  We elide the range in our code
examples.

We here do not distinguish comprehensions that produce values---like
the one that produces $\kw{skills}$---and those that do not---like the one that observes player performances according to $\kw{results}$.
For the sake of efficiency, our implementation does distinguish these two uses.
In some of the code examples, we write $\kw{for}\ x\ \kw{in}\ V\ \kw{do}\ M$ to mean $[\Lforeach{x}{r}{V}{M}]$.
We do so only when $M$ has type $\unitT$ and we intend to ignore the result of the expression.

We encode arrays as tuples.
For all $n>0$, we define $\pi_n(M,N)$ with $M : t^n$ and $N :
\kw{int}$ and if $N\%n=i$ we expect $\pi_n((V_0,\dots,V_{n-1}),N)=
V_i$.

\begin{display}[.7]{Derived Types and Expressions for Arrays in \fun:}
\clause{ \pi_1(M,N) := M }\\
\clause{ \pi_n(M, N) := \kw{if}\ N\kw{\%}n \kw{== 0 then}\ M.1\ \kw{else}\ \pi_{n-1}(M.2, N-1) }{for $n>1$}\\[\GAP]
\clause{t[r] := t^{|r|} \quad \mbox{where $t^1 := t$ and $t^{n+1} := t * t^n$}} \\[\GAP]
\clause{[V_0;...;V_{n-1}] := (V_0, \ldots ,V_{n-1})} \\
\clause{V_1[V_2]_r := \pi_{\abs{r}}(V_1,V_2)} \\
\clause{\Lforeach{x}{r}{V}{M} :=} \\
\clause{\qquad
  \kw{let}\ y_0 = (\kw{let}\ x=\pi_{|r|}(V,0)\ \kw{in}\ M)\ \kw{in}
} \\
\clause{\qquad\cdots
} \\
\clause{\qquad
  \kw{let}\ y_{|r|-1} = (\kw{let}\ x=\pi_{|r|}(V,|r|-1)\ \kw{in}\ M)\ \kw{in}
}\\
\clause{\qquad
  (y_0;\ldots;y_{|r|-1}) 
\quad\text{where $y_1, \ldots, y_{|r|}$ are fresh for $M$ and $V$.}
}
\end{display}
Our derived forms for arrays yield programs whose size grows linearly with the data over
which they compute---we implement $V[i]_r$ with $O(|r|)$ projections.
To avoid this problem, our implementation takes advantage of support
for arrays in the Infer.NET factor graph library (see Section~\ref{sec:arraysInImp}).

The static semantics of these new constructs is straightforward; we
give the derived rules for \ref{Fun Array}, \ref{Fun Index}, and
\ref{Fun For}.
By adding these as derived forms in \fun, we do not need to extend \imp at all.
On the other hand, our formalization does not reflect that our
implementation preserves the structure of array comprehensions when
going to Infer.NET.

\begin{display}[0.35]{Extended Typing Rules for {\fun} Expressions: $\Gamma \vdash M : t$}
\>\staterule{Fun Array}{
\Gamma \vdash V_i : t \quad \forall i \in 0..n-1
}{
\Gamma \vdash [ V_0; \ldots; V_{n-1} ] : t[r_n]
} \quad~~
\staterule{Fun Index}{
\Gamma \vdash V_1 : t[r] \quad \Gamma \vdash V_2 : \intT
}{
\Gamma \vdash V_1[V_2]_r : t
} \quad~~
\staterule{Fun For}{
\Gamma \vdash V : t[r] \quad \Gamma,x:t \vdash M : t'
}{
\Gamma \vdash [\Lforeach{x}{r}{V}{M}] : t'[r]
}
\end{display}

%
The rule \ref{Fun Array} uses the notation $r_n$ for the
\emph{concrete range} of size $n$; we assume there is a unique such
range for each $n>0$.
This rule can be derived using repeated applications of \ref{Fun Pair}.
The rule \ref{Fun Index} checks that the array $V_1$ is non-empty array and the index $V_2$ is an integer;
the actual index is the value of $V_2$ modulo the size of the array, as in the meta-language.
We can derive this rule for a given $n$ by induction on $n$, using
repeated applications of \ref{Fun If}; we use \ref{Fun Proj1} in the
$\kw{then}$ case and \ref{Fun Proj2} in the $\kw{else}$ case.
%
The rule \ref{Fun For} requires
that the source expression $V$ is an array, 
and that the body $M$ is well-typed assuming a suitable type for $x$.
We can derive \ref{Fun For} using repeated applications of \ref{Fun Let},
with \ref{Fun Pair} to type the final result.

\MySubsection{Arrays in \imp}\label{sec:arraysInImp}
We now sketch our structure-preserving implementation strategy.
We work in a version of \imp with arrays and iteration over ranges,
and we extend both the assignment form and expressions to permit array indexing.  
Inside the body of an iteration over a range, the name of the range can be used as an index.
\begin{display}[0.35]{Extended Syntax of \imp:}
\smallCategory{E}{\ldots\mid l[l'] \mid l[r] }{expression} \\
\extendCategory{I}{statement}\\
\entry{l[r] \assign E}{assignment to array item} \\
\entry{\Lfor{r}{C}}{iteration over ranges}
\end{display}
We require that every occurrence of an index $r$ is inside an iteration $\Lfor rC$.
Inside such an iteration, every assignment to an array variable must be at index $r$.
We also extend patterns to include range indexed locations, 
and write $(p_1,p_2)[r]$ for $(p_1[r],p_2[r])$.

Our compiler translates comprehensions over variables of array type as
an iteration over the translation of the body of the comprehension.
We add to $\rho$ the fact that the comprehension variable corresponds to
the array variable indexed by the range.
We invent a fresh array result pattern $p'$, 
and assign the result of the translated body to $p'[r]$.
Finally, we hide the local variables of the translation of the body of the comprehension,
in order to avoid clashes in the unrolling semantics of the loop.
This compilation corresponds to the rule \ref{Trans For} below.
In particular, the sizes of ranges are never needed in our compiler,
so compilation is not data dependent.
\begin{display}[0.35]{Compilation of comprehensions:}
\>\staterule{Trans For}{
\rho\{x\mapsto \rho(z)[r]\} \vdash M \Rightarrow C,p 
\qquad  p[r] \sim p' 
\qquad (\locs(\rho) \cup \locs(C)) \cap \fv(p') = \emptyset 
}{
\rho \vdash [\Lforeach{x}{r}{z}{M}] \Rightarrow \Lfor{r}{\Local{\locs(C)} (C;p'[r] \assign p)},p'
}
\end{display}
\end{FULL} 

\section{Implementation Experience}\label{sec:experience}

We implemented a compiler from \fun to \imp in \fsharp.
We wrote two backends for \imp:
an exact inference algorithm based on a direct implementation of measure transformers for discrete measures,
and an approximating inference algorithm for continuous measures, using Infer.NET~\cite{infer.net}.
The translation of Section~\ref{sec:semantics-fg} formalizes our translation of \fun to \imp. 
Translating \imp to Infer.NET is relatively straightforward, and amounts to a syntax-directed series of calls to Infer.NET's object-oriented API.

\begin{FULL}%
The frontend of our compiler takes (a subset of) actual \fsharp code as its input.
To do so, we make use of \fsharp's {\em reflected definitions}, which allow programmatic access to ASTs.
This implementation strategy is advantageous in several ways.
First, there is no need to design new syntax, or even write a parser.
Second, all inputs to our compiler are typed ASTs of well typed \fsharp programs.
Third, a single file can contain both ordinary \fsharp code as well as reflected definitions.
This allows a single module to both read and process data, and to specify a probabilistic model for inference from the data.

Functions computing array values containing deterministic data are tagged with an attribute \kw{RegisterArray}, to signal to the compiler that they do not need to be interpreted as \fun programs.
Reflected definitions later in the same file are typed with respect to these registered definitions and then run in Infer.NET with the pre-processed data;
we further discuss this idea below.
\end{FULL}%

Below follows some statistics on a few of the examples we have implemented.
The number of  lines of code includes \fsharp code that loads and processes data from disk 
before loading it into Infer.NET.
The times are based on an average of three runs.
All of the runs are on a four-core machine with 4GB of RAM.
The Naive Bayes program is the naive Bayesian classifier of the earlier examples.
The Mixture model is another clustering/classification model.
\begin{FULL}%
TrueSkill and adPredictor were described earlier.
\end{FULL}%
\begin{SHORT}
TrueSkill is a tournament ranking model, and adPredictor is a
simplified version of a model to predict the likelihood
that a display advertisment will be clicked.
\end{SHORT}
TrueSkill spends the majority of its time (64\%) in Infer.NET, performing inference.
AdPredictor spends most of the time in pre-processing (58\%), and only 40\% in inference.
The time spent in our compiler is negligible, never more than a few hundred milliseconds.
\begin{display}{Summary of our Basic Test Suite:}
\quad\begin{tabular}[\linewidth]{c|c|c|c|c}
& LOC & Observations & Variables & Time \\
\hline
Naive Bayes & 28 & 9 & 3 & $<$1s \\
Mixture & 33 & 3 & 3 & $<$1s \\
TrueSkill & 68 & 15,664 & 84 & 6s \\
adPredictor & 78 & 300,752 & 299,594 & 3m30s
\end{tabular}
\end{display}
In summary, our implementation strategy allowed us to build an effective prototype quickly and easily: 
the entire compiler is only 2079 lines of \fsharp; the Infer.NET backend is 600 lines; the discrete backend is 252 lines.
Our implementation, however, is only a prototype, and has limitations.
Our discrete backend is limited to small models using only finite measures.
Infer.NET supports only a limited set of operations on specific combinations of probabilistic and deterministic arguments.
It would be useful in the future to have an enhanced type system able to  detect errors arising from illegal combinations of operators in Infer.NET. 
The reflected definition facility is somewhat limited in \fsharp.
In the adPredictor example on page~\pageref{disp:adPredictor}, a call to \kw{Array.toList} is required because \fsharp does not reflect definitions that contain comprehensions over arrays---only lists.
(The \fsharp to \fun compiler discards this extra call as a no-op, so there is no runtime overhead.)

\section{Related Work}
\label{sec:related}

\adg{responded to reviewer A:
In part, I really liked the paper, but overall I felt that this work
suffered from ignorance of the extant literature on probabilistic programs.
The citations are to work coming from the machine learning community (which
as a whole seems to be unaware of many things in programming languages that
might be relevant to them; the dual also holds of course).
The semantics given is essentially a trivial exercise compared to what has
been known since the 1980s. In 1985 Dexter Kozen developed a
semantics for probabilistic programs including looping (later a simplified
presentation of this same semantics was given by Panangaden circa
1998, see the recent book Labelled Markov Processes); also in the late 1980s and
early 1990s the semantics of probabilistic and nondeterministic programs
based on weakest preconditions was developed by Karen Seidel, Annabelle
McIver, Carroll Morgan and others from PRG Oxford (the best source is the
recent book by McIver and Morgan); this semantics works in a way dual to
the forward semantics (the type you have in this paper). In POPL99,
Gupta, Jagadeesan and Panangaden gave a semantics for probabilistic
concurrent constraint programs also including recursion. This semantics also
rejects runs that fail the constraint. Your own language relies on constraints
in a crucial way and it seems clear to me that what you want is a
constraint programming language rather than having constraints tossed in on top of
a functional language. The measure theory used in this paper is very
elementary and hardly a technical advance; it is when you deal with
recursion that subtle and interesting things happen. The domain theory paper by
Saheb-Djaromi that you cite is just the start of a long and intricate story. I
can only suggest that you take a look at some of the sources cited.
}



\subsubsection*{Formal Semantics of Probabilistic Languages}
There is a long history of formal semantics for probabilistic languages with sampling primitives, often combined with recursive computation.
One of the first semantics is for Probabilistic LCF \citep{DBLP:conf/mfcs/Saheb-Djahromi78},
which augments the core functional language LCF with weighted binary choice, for discrete distributions.
\begin{FULL}(Apart from its inclusion of observations, \bfun is a first-order terminating form of Probabilistic LCF.)\,\end{FULL}
\citeT{Kozen}{DBLP:journals/jcss/Kozen81} develops a probabilistic semantics for while-programs augmented with random assignment.
He develops two provably equivalent semantics; one more operational, and the other a denotational semantics using partially ordered Banach spaces.
\imp is simpler than Kozen's language, as \imp has no unbounded while-statements, so the semantics of \imp need not deal with non-termination.
On the other hand, observations are not present in Kozen's language, although discrete observations can be encoded using possibly non-terminating while loops.

\citeT{Jones and Plotkin}{DBLP:conf/lics/JonesP89} investigate the probability monad, and apply it to languages with discrete probabilistic choice.
\citeT{Ramsey and Pfeffer}{DBLP:conf/popl/RamseyP02} 
give a stochastic $\lambda$-calculus with a measure-theoretic semantics in the probability monad, 
and provide an embedding within Haskell;
they do not consider observations.  
We can generalize the semantics of
\kw{observe} to the stochastic $\lambda$-calculus as filtering in the probability monad 
(yielding what we may call a sub-probability monad), 
as long as the events that are being
observed are discrete.
In their notation, we can augment their language with a failure construct defined by ${\cal P}\qq{ \kw{fail} }\rho = \mu_0$ where we define $\mu_0(A)=0$ for all measurable sets $A$.
Then, we can define $\kw{observe}\ v = (\kw{if}\ v=\kw{true}\ \kw{then}\ ()\ \kw{else}\ \kw{fail})$.
However, as discussed in Section~\ref{sec:discussion},
zero-probability observations of real variables do not translate easily to the probability monad, as the following example shows.
Let $N$ be an expression denoting a continuous distribution, for example, \ls$random (Gaussian(0.0,1.0))$, and let \ls$f x$ = \ls$observe x$.
Suppose there is a semantics for $\qq{\kw{f x}} \{\kw{x} \mapsto r\}$ for real $r$ in the probability monad.
The probability monad semantics of the program  \kw{let x =} $N$ \kw{in f x}
of the stochastic $\lambda$-calculus is
$\qq{N} \gg= \lambda y. \qq{\kw{f x}} \{\kw{x} \mapsto y\}$, 
which yields the measure
$\mu(A)= \int_{\Real} (\mathbf{M}\qq{\qq{\kw{f x}} \{\kw{x} \mapsto y\}})(A)\; d \mathbf{M}[N](y)$.
Here the probability $(\mathbf{M}\qq{\qq{\kw{f\ x}} \{x \mapsto y\}})(A)$ is zero except when $y=0$, 
where it is some real number.
Since the $N$-measure of $y=0$ is zero, the whole integral is zero for all $A$ (in particular $\mu(\Real)=0$),
whereas the intended semantics is that \ls$x$ is constrained to be zero with probability 1 
(so in particular $\mu(\Real)=1$).

The probabilistic concurrent  constraint programming language Probabilistic cc of \citeT{Gupta, Jagadeesan, and Panangaden}{DBLP:conf/popl/GuptaJP99} is also intended for describing probability
distributions using independent sampling and constraints.
Our use of observations loosely corresponds to constraints on random variables in Probabilistic cc.
In the finite case, Probabilistic cc also relies on a sampling semantics with observation (constraints) denoting filtering.
To admit continuous distributions, Probabilistic cc adds general fixpoints and defines the semantics of a program as the limit of finite unrollings of its fixpoints, if defined.
This can lead to surprising results, such as that the distribution resulting from observing that two apparently uniform distributions are equal may not itself be uniform.
In contrast, we work directly with standard distributions
and have a less syntactic semantics of observation that appears to be easier to anticipate.

\citeT{McIver and Morgan}{MM05:AbstractionRefinementProofForProbabilisticSystems}
develop a theory of abstraction and refinement for probabilistic while
programs, based on weakest preconditions. They reject a
subdistribution transformer semantics in order to admit demonic
nondeterminism in the language.

We conjecture that \fun and \imp could in principle be conferred semantics within a probabilistic language supporting general recursion,
by encoding discrete observations by placing the whole program within a conditional sampling loop,
and by encoding Gaussian and other continuous distributions as repeated sampling using recursive functions.
Still, dealing with recursion would be a non-trivial development, and would raise issues of computability.
Ackerman, Freer, and Roy \cite{DBLP:conf/lics/AckermanFR11} show the uncomputability of conditional distributions in general,
establishing limitations on constructive foundations of probabilistic programming.
We chose when formulating the semantics of \fun and \imp to include some distributions as primitive, and to exclude recursion;
compared to encodings within probabilistic languages with recursion, this choice has the advantage of compositionality (rather than relying on a global sampling loop)
and of admitting a direct (if sometimes approximate) implementation (via message-passing algorithms on factor graphs, with efficient implementations of primitive distributions).

\adg{React to referee comment: conjecture ... in principle". This language is understating the work it would take to lift this analysis.
I think it's funny to suggest the use of while loops to implement observation when earlier this same tact in other languages is not conflated with "having observations."
Also, as for the prospect of giving semantics using recursive sampling functions to implement continuous random variables, see Ackerman, Freer, Roy 2010 for inherent limitations to implementing observe continuously and computably in the presence of general recursion. In particular, this conjecture comes with it severe restrictions.}

\adg{Cut for now:
Several of these languages implement a limited form of observations as conditional sampling; this
entails using a loop around the sampling operation, a post-processing step, or a conditional sampling primitive.}

\adg{Since observations can happen at any time and of any expression, they
intuitively correspond to having the whole program inside a conditional
sampling loop.  In our setting observations are local, and we do not need a
fixpoint semantics to express them.
We wrote "Several of these languages implement a limited form of observations
as conditional sampling" in section 7.  
By "limited" we meant that there are some languages that only permit a
sample from a conditional distribution (Gaussian such that...), which seems to
be less general.  On the other hand, recursion-based conditional sampling 
appears to us likely to be more general than Fun.
Still, our semantics uses relatively elementary measure-theory and is
compositional, whereas a postulated derived semantics based on a language with
recursion requires a more sophisticated framework, and depends on a global
loop, so is less compositional.  Developing the latter semantics is
interesting future work, but the direct and relatively elementary semantics of
the present paper remains a new contribution.}

\begin{FULL}%
Recent work on semantics of probabilistic programs within interactive theorem provers
includes the mechanization of measure theory \cite{Hurd01} and Lebesgue integration \cite{MHT10:LebesgueIntegration} in HOL,
and a framework for proofs of randomized algorithms in Coq \cite{DBLP:journals/scp/AudebaudP09} which also allows for discrete observations.

\end{FULL}%

\subsubsection*{Probabilistic Languages for Machine Learning}
\citeT{Koller et al.}{DBLP:conf/aaai/KollerMP97} proposed representing a
probability distribution using first-order functional programs with discrete random choice,
and proposed an inference algorithm for Bayesian networks and stochastic context-free grammars.
Observations happen outside their language, by returning the distributions $\Prob{A\land B}, \Prob{A\land \neg B},\Prob{\neg A}$
which can be used to compute $\Prob{B\mid A}$.
Their work was subsequently developed by Pfeffer into the language IBAL~\cite{pfeffer07:ibal},
which has observations and uses a factor graph semantics, 
but only works with discrete datatypes.

\ADG{P1 Tom points out some of these works are libraries and some are languages, so our classification needs to be clearer.}

\citeT{Park et al.}{DBLP:conf/popl/ParkPT05} propose  $\lambda_\circ$,
the first probabilistic language with formal semantics applied to actual machine learning problems involving continuous distributions.
The formal basis is sampling functions,
which uniformly supports both discrete and continuous probability distributions,
and inference is by Monte Carlo importance sampling methods.
The calculus $\lambda_\circ$ enables conditional sampling via fixpoints and rejection, 
and its implementation allows discrete observations only.

\ADG{2013 the ref asks about discrete observations only: did you check?}
\JB{Yes, changed to talk about their implementation as described in the paper.}


%
%
HANSEI \citep{monolingual2009,KS09:EmbeddedProbabilisticProgramming} 
is an embedding of a probabilistic language as a programming library in OCaml,
based on explicit manipulation of discrete probability distributions as lists,
and sampling algorithms based on coroutines.
HANSEI uses an explicit \texttt{fail} statement, which is
equivalent to \ls{observe false} and so cannot be used for conditioning on
zero probability events.
Infer.NET \citep{infer.net} is a software library that implements the approximate deterministic algorithms
expectation propagation \cite{DBLP:conf/uai/Minka01} and variational message passing \cite{DBLP:journals/jmlr/WinnB05},
as well as Gibbs sampling, a nondeterministic algorithm.
Infer.NET models are written in a probabilistic subset of C\#, known as \csoft  \cite{Csoft:WM09}.
\csoft allows $\kw{observe}$ on zero probability events, 
but does not have a continuous semantics other than as factor graphs 
and is currently only implemented as an internal language of Infer.NET.
This paper gives a higher-level semantics of \csoft (or \imp) programs as distribution transformers.

Although there are many Bayesian modelling languages,
\csoft and IBAL are the only previous languages 
implemented by a compilation to factor graphs.
Probabilistic Scheme~\citep{Radul:2007:RPL:1297081.1297085} is a probabilistic form of the untyped functional language Scheme,
limited to discrete distributions, and with a construct for reifying the distribution induced by a thunk as a value.
Church \citep{DBLP:conf/uai/GoodmanMRBT08} is another probabilistic form of Scheme,
equipped with conditional sampling
and a mechanism of stochastic memoization.
In MIT-Church, queries are implemented using Markov chain Monte Carlo methods.
WinBUGS \citep{winbugs} is a popular implementation of the BUGS language \cite{GTS94:Bugs} 
for explicitly describing distributions suitable for MCMC analysis.

FACTORIE \citep{FACTORIE} is a Scala library for explicitly constructing factor graphs.
Blaise \cite{blaise} is a software library for building MCMC samplers in Java, 
that supports compositional construction of sophisticated probabilistic models,
and decouples the choice of inference algorithm from the specification of the distribution.

A recent paper \cite{modelLearner} based on \fun describes a model-learner pattern which captures common probabilistic programming patterns in machine learning,
including various sorts of mixture models.


%


\subsubsection*{Other Uses of Probabilistic Languages}
Probabilistic languages with formal semantics
find application in many areas apart from machine learning,
including databases~\citep{DBLP:journals/cacm/DalviRS09},
model checking~\citep{DBLP:journals/entcs/KwiatkowskaNP06},
differential privacy~\citep{DBLP:conf/sigmod/McSherry09,cdp},
information flow~\citep{DBLP:conf/csfw/Lowe02},
and cryptography~\citep{DBLP:journals/joc/AbadiR02}.
A recent monograph on semantics for labelled Markov processes~\cite{Panangaden09:LMP}
focuses on bisimulation-based equational reasoning.
The syntax and semantics of \imp is modelled on the probabilistic language pWhile \citep{DBLP:conf/popl/BartheGB09} without observations.
%

\citeT{Erwig and Kollmansberger}{DBLP:journals/jfp/ErwigK06} describe a library for probabilistic functional programming in Haskell.
The library is based on the probability monad, and uses a finite representation suitable for small discrete distributions;
the library would not suffice to provide a semantics for \fun or \imp with their continuous and hybrid distributions.
\begin{FULL}%
Their library has similar functionality to that provided by our combinators for discrete distributions listed in the technical report.
\end{FULL}%





\section{Conclusion}
\label{sec:conc}

\begin{FULL}%
We advocate probabilistic functional programming with observations and comprehensions as a modelling language for Bayesian reasoning.
We developed a system based on the idea,
invented new formal semantics to establish correctness,
and evaluated the system on a series of typical inference problems.
\end{FULL}%

Our direct contribution is a rigorous semantics for a probabilistic
programming language with zero-probability observations on continuous variables.
%
We have shown that probabilistic functional programs with iteration over arrays,
but without the complexities of general recursion,
are a concise representation for complex probability distributions arising in machine learning.
An implication of our work for the machine learning community
is that probabilistic programs can be written directly within an existing declarative language
(\fun---a subset of \fsharp), linked by comprehensions to large datasets,
and compiled down to lower level Bayesian inference engines.
%

For the programming language community, our new semantics suggests some novel directions
for research.
What other primitives are possible---non-generative models, inspection of
distributions, on-line inference on data streams?
Can we verify the transformations performed by machine learning compilers such as Infer.NET compiler for \csoft?
What is the role of type systems for such probabilistic languages?
Avoiding (discrete) zero probability exceptions, 
and ensuring that we only generate \csoft programs suitable for our back-end, 
are two possibilities, but we expect there are more.
\subsubsection*{Acknowledgements}
We gratefully acknowledge discussions with and comments from Ralf Herbrich, Oleg Kiselyov, Tom Minka, Aditya Nori, Robert Simmons, Nikhil Swamy, Dimitrios Vytiniotis and John Winn.
Chung-Chieh Shan highlighted an issue with our original definition of observation.
The comments by the anonymous reviewers were most helpful, 
in particular regarding the definition of conditional density.

\begin{FULL}%
\appendix\clearpage
\section{Detailed Proofs}\label{app:detailed-proofs}

Our proofs are structured as follows.
\begin{iteMize}{$\bullet$}
\item
Appendix~\ref{sec:proof-prop} gives a proof of Proposition~\ref{prop:fun-to-imp-static}.
\item
Appendix~\ref{sec:proof-thms} gives a proof of Theorem~\ref{thm:fun-to-imp-correct}.
\end{iteMize}

\subsection{Proof of Proposition~\ref{prop:fun-to-imp-static}}\label{sec:proof-prop}

We begin with a series of lemmas.

\begin{lem}[Pattern agreement weakening]
  \label{lem:weakenpat}
  If $\Sigma \vdash p : t$ and $\Sigma,\Sigma' \vdash \diamond$, 
  then $\Sigma,\Sigma' \vdash p : t$.
\end{lem}
  \begin{proof}
    By induction on $t$. \QED
  \end{proof}

\begin{lem}[Expression and statement heap weakening]
  \label{lem:weakenSEC}\hfill
  \begin{enumerate}[\em(1)]
  \item If $\Sigma \vdash E : b$ and $\Sigma,\Sigma' \vdash \diamond$, 
    then $\Sigma,\Sigma' \vdash E : b$

  \item If $\Sigma \vdash I : \Sigma'$ and $\Sigma,\Sigma',\Sigma'' \vdash \diamond$,
    then $\Sigma,\Sigma'' \vdash I : \Sigma'$

  \item If $\Sigma \vdash C : \Sigma'$ and $\Sigma,\Sigma',\Sigma'' \vdash \diamond$,
    then $\Sigma,\Sigma'' \vdash C : \Sigma'$.
  \end{enumerate}
\end{lem}
  \begin{proof}
    By induction on $E$, $I$, and $C$, respectively. \QED
  \end{proof}

\begin{lem}[Pattern agreement uniqueness]
  \label{lem:patunique}
  If $\Sigma \vdash p : t$ and $\Sigma' \vdash p' : t$
  then $p \sim p'$.
\end{lem}
  \begin{proof}
    By induction on $t$. \QED
  \end{proof}

\begin{lem}[Pattern creation]
  \label{lem:patcreate}
  If $\Sigma \vdash p : t$ then there exists $\Sigma'$ such that $\Sigma,\Sigma' \vdash \diamond$
  and $\Sigma' \vdash p' : t$ and $\dom(\Sigma') = \fv(p')$.
\end{lem}
  \begin{proof}
    By induction on $t$, and the assumption that there always exist new,
    globally fresh locations. \QED
  \end{proof}

\begin{lem}[Pattern assignment]
  \label{lem:patassign}
  If $\Sigma \vdash p : t$ and $\Sigma' \vdash p' : t$ and
  $\Sigma,\Sigma' \vdash \diamond$, 
  then $\Sigma \vdash p' \assign p : \Sigma''$, where $\Sigma'' \subseteq \Sigma'$.
\end{lem}
  \begin{proof}
    By induction on $t$.

    \begin{iteMize}{$\bullet$}
    \item $(t=\unitT)$ Trivial: $p' \assign p = \kw{nil}$, so $\Sigma''
      = \emptyEnv \subseteq \Sigma'$.

    \item $(t=\boolT)$ $\Sigma \vdash l : \boolT$ and $\Sigma' \vdash l'
      : \boolT$, so $l:\boolT \in \Sigma$ and $l':\boolT \in \Sigma'$. So
      $l:\boolT \vdash l' \assign l : (l':\boolT) \subseteq \Sigma'$.

    \item $(t=\intT)$ Similar.
      
    \item $(t=\realT)$ Similar.
      
    \item $(t=t_1 * t_2)$ $\Sigma \vdash p_1,p_2 : t_1 * t_2$ and
      $\Sigma' \vdash p_1',p_2' : t_1 * t_2$.
      Both $\Sigma$ and $\Sigma'$ factor into contexts that type $p_1$
      and $p_2$ (resp. $p_1'$ and $p_2'$) individually; call them
      $\Sigma_1$ and $\Sigma_2$ (resp. $\Sigma_1'$ and $\Sigma_2'$).
      By the IHs, we have $\Sigma_1 \vdash p_1' \assign p_1 :
      \Sigma_1'' \subseteq \Sigma_1'$ and $\Sigma_2 \vdash p_2' \assign
      p_2 : \Sigma_2'' \subseteq \Sigma_2'$.
      We can then see $\Sigma \vdash p_1' \assign p_1;p_2' \assign p_2
      : \Sigma_1'',\Sigma_2'' \subseteq \Sigma_1',\Sigma_2'$.\qedhere
    \end{iteMize}
  \end{proof}

\noindent The purpose of this subsection is to prove the following.
\begin{restate}{Proposition~\ref{prop:fun-to-imp-static}}
  Suppose $\Gamma \vdash M : t$ and $\Sigma \vdash \rho : \Gamma$.
  \begin{enumerate}[\em(1)]
  \item\label{p1} There are $C$ and $p$ such that $\rho \vdash M \Rightarrow C,p$.
  \item\label{p2} Whenever $\rho \vdash M \Rightarrow C,p$, there is $\Sigma'$ such that $\Sigma \vdash C:\Sigma'$ and $\Sigma,\Sigma' \vdash p : t$.
  \end{enumerate}
\end{restate}
\begin{proof}
  By induction on the typing of $M$, leaving $\Sigma$ and $\rho$ general.
  \begin{enumerate}[\hbox to8 pt{\hfill}] 
  \item\noindent{\hskip-12 pt\ref{Fun Var}} $\Gamma \vdash x : t$. For (\ref{p1}), we have $C = \kw{nil}$
    and $p = \rho(x)$.
    For (\ref{p2}), let $\Sigma' = \emptyEnv$.  By assumption,
    $\Sigma,\Sigma' \vdash \rho(x) : t$ and $\Sigma \vdash \kw{nil} :
    \Sigma'$ immediately.
    
  \item\noindent{\hskip-12 pt\ref{Fun Const}} $\Gamma \vdash c : ty(c)$.  For (\ref{p1}), we have:
    \[\begin{array}{l}
      l \not\in \locs(\rho) \\
      ty(c) = b \text{ for some base type b} \\
      \rho \vdash c \Rightarrow l \assign c,l 
    \end{array}\]
    For (\ref{p2}), let $\Sigma' = l:ty(c)$.  We have $\Sigma,\Sigma'
    \vdash l : ty(c)$ and $\Sigma \vdash l \assign c : \Sigma'$.

  \item\noindent{\hskip-12 pt\ref{Fun Operator}} $\Gamma \vdash V_1 \otimes V_2 : b_3$, where $\otimes$ has type $b_1*b_2 \rightarrow b_3$.  By inversion and the IH:
    \[\begin{array}{lr} 
      \Gamma \vdash V_1 : b_1 & \\
      \rho \vdash V_1 \Rightarrow C_1,l_1 & (IH_1) \\
      \exists \Sigma_1 & (IH_2) \\
      \quad \Sigma,\Sigma_1 \vdash l_1 : b_1 &\\
      \quad \Sigma \vdash C_1 : \Sigma_1 &\\
      \Gamma \vdash V_2 : b_2 &\\
      \rho \vdash V_2 \Rightarrow C_2,l_2 & (IH_2) \\
      \exists \Sigma_2 & (IH_2) \\
      \quad \Sigma,\Sigma_2 \vdash l_2 : b_2 &\\
      \quad \Sigma \vdash C_2 : \Sigma_2 &
    \end{array}\]
    We have for (\ref{p1}), by \ref{Trans Operator}: $\rho \vdash V_1
    \otimes V_2 \Rightarrow C_1;C_2;l \assign l_1 \otimes l_2,l$.
    Let $\Sigma' = \Sigma_1,\Sigma_2,l:b_3 \vdash \diamond$.
    By weakening we find for (\ref{p2}): $\Sigma,\Sigma' \vdash l :
    b_3$ and $\Sigma \vdash C_1; C_2; l \assign l_1 \otimes l_2 :
    \Sigma'$.

  \item\noindent{\hskip-12 pt\ref{Fun Pair}} $\Gamma \vdash (M_1, M_2) : t_1 * t_2$. By inversion and the IH:
    \[\begin{array}{lr}
      \Gamma \vdash M_1 : t_1 & \\
      \rho \vdash M_1 \Rightarrow C_{M_1},p_1 & (IH_1)\\
      \exists \Sigma_1 & (IH_2) \\
      \quad \Sigma,\Sigma_1 \vdash p_1 : t_1 &\\
      \quad \Sigma \vdash C_{M_1} : \Sigma_1 &\\
      \Gamma \vdash M_2 : t_2 &\\
      \rho \vdash M_2 \Rightarrow C_{M_2},p_2 & (IH_1)\\
      \exists \Sigma_2 & (IH_2) \\
      \quad \Sigma,\Sigma_2 \vdash p_2 : t_2 &\\
      \quad \Sigma \vdash C_{M_2} : \Sigma_2 &
    \end{array}\]
    We have for (\ref{p1}):
    $\rho \vdash (M_1,M_2) \Rightarrow C_{M_1}; C_{M_2}, (p_1,p_2)$.
    Let $\Sigma' = \Sigma_1,\Sigma_2 \vdash \diamond$.
    By weakening we find for (\ref{p2}): $\Sigma,\Sigma' \vdash
    (p_1,p_2) : t_1 * t_2$ and $\Sigma \vdash C_{M_1}; C_{M_2} :
    \Sigma'$.

  \item\noindent{\hskip-12 pt\ref{Fun Proj1}} $\Gamma \vdash M.1 : t_1$. By inversion and the IH:
    \[\begin{array}{lr}
      \Gamma \vdash M : t_1 * t_2 &\\
      \rho \vdash M \Rightarrow C_M,p & (IH_1) \\
      \exists \Sigma' & (IH_2) \\
      \quad \Sigma,\Sigma' \vdash p : t_1 * t_2 &\\
      \quad \Sigma \vdash M : S' &
    \end{array}\]
    By inversion, $p = (p_1,p_2)$, such that $\Sigma,\Sigma' \vdash
    p_1 : t_1$ and $\Sigma,\Sigma' \vdash p_2 : t_2$.
    We now have $\rho \vdash M.1 \Rightarrow C_M, p_1$ for (\ref{p1}).
    We use $\Sigma'$ to show $\Sigma,\Sigma' \vdash p_1 : t_1$ and
    $\Sigma \vdash C_M : \Sigma'$ for (\ref{p2}).

  \item\noindent{\hskip-12 pt\ref{Fun Proj2}} $\Gamma \vdash M.2 : t_2$. Analogous to the previous case.

  \item\noindent{\hskip-12 pt\ref{Fun If}} $\Gamma \vdash \Lif{M_1}{M_2}{M3} : t$.  We have:
    \[\begin{array}{lr}
      \Gamma \vdash M_1 : \boolT &\\
      \rho \vdash M_1 \Rightarrow C_{M_1},p_1 & (IH_1) \\
      \exists \Sigma_1 & (IH_2) \\
      \quad \Sigma,\Sigma_1 \vdash p_1 : \boolT &\\
      \quad \Sigma \vdash C_{M_1} : \Sigma_1 &\\
      \Gamma \vdash M_2 : t &\\
      \rho\{x \mapsto p_l\} \vdash M_2 \Rightarrow C_{M_2}, p_2 & (IH_1) \\
      \exists \Sigma_2 & (IH_2) \\
      \quad \Sigma,\Sigma_2 \vdash p_2 : t &\\
      \quad \Sigma \vdash C_{M_2} : \Sigma_2 &\\
      \Gamma \vdash M_3 : t &\\
      \rho\{x \mapsto p_r \} \vdash M_3 \Rightarrow C_{M_3}, p_3 & (IH_1)\\
      \exists \Sigma_3 & (IH_2) \\
      \quad \Sigma,\Sigma_3 \vdash p_3 : t &\\
      \quad \Sigma \vdash C_{M_3} : \Sigma_3 &\\
    \end{array}\]
    By inversion, $p_1 = l$ and $\Sigma,\Sigma_1 \vdash l : \boolT$.
    By pattern agreement uniqueness (Lemma~\ref{lem:patunique}),
    $p_2 \sim p_3$.
    Let $\Sigma_{p'} \vdash p' : t$, for $\dom(\Sigma_{p'})=fv(p)$ (by
    Lemma~\ref{lem:patcreate}). 
    We have $(\locs(\rho) \cup \locs(C_1) \cup \locs(C_2) \cup
    \locs(C_3)) \cap fv(p) = \emptyset$.
    We also have $p' \sim p_2$ and $p' \sim p_3$.
    We now have for (\ref{p1}):
    \[\begin{array}{l}
      \rho \vdash \Lif{M_1}{M_2}{M_3} \Rightarrow \\
      C_{M_1}; 
      \Lif{l}{
        \Local{\locs(C_2)} C_{M_2}; [[p' \assign p_2]]}{
        \Local{\locs(C_3)} C_{M_3}; [[p' \assign p_3]]}, p' 
    \end{array}\]
    Finally, let $\Sigma_f = \Sigma_2 \cap \Sigma_3 \cap \Sigma_{p'}
    \vdash \diamond$ and $\Sigma' = \Sigma_1,\Sigma_f \vdash
    \diamond$.
    By pattern assignment, we can see $\Sigma_f \vdash [[p' \assign
    p_2]]$ and $\Sigma_f \vdash [[p' \assign p_3]]$.
    By weakening (Lemmas~\ref{lem:weakenpat}, and~\ref{lem:weakenSEC})
    we have what we need for (\ref{p2}):
    \[\begin{array}{l}
      \Sigma,\Sigma' \vdash p' : t \\
      \Sigma \vdash C_{M_1}; \Lif{l}{...}{...} : \Sigma'
     \end{array}\]

  \item\noindent{\hskip-12 pt\ref{Fun Let}} $\Gamma \vdash \kw{let}~x = M_1~\kw{in}~M_2 : t_2$. We have:
    \[\begin{array}{lr}
      \Gamma \vdash M_1 : t_1 &\\
      \rho \vdash M_1 \Rightarrow C_{M_1},p_1 &(IH_1) \\
      \exists \Sigma_1 & (IH_2) \\
      \quad \Sigma,\Sigma_1 \vdash p_1 : t_1 &\\
      \quad \Sigma \vdash C_{M_1} : \Sigma_1 &\\
      \Gamma,x:T_1 \vdash M_2 : t_2 &
    \end{array}\]
    Next, note that $\Sigma,\Sigma_1 \vdash \rho\{x \mapsto p_1\} :
    \Gamma,x:T_1$.
    We can now apply the IH to $M_2$'s typing derivation to see:
    \[\begin{array}{lr}
      \quad \rho\{x \mapsto p_1\} \vdash M_2 \Rightarrow C_{M_2},p_2 & (IH_1)\\
      \exists \Sigma_2 & (IH_2) \\
      \quad \Sigma,\Sigma_2 \vdash p_2 : t_2 &\\
      \quad \Sigma \vdash C_{M_2} : \Sigma_2 &
    \end{array}\]
    First, we have:
    $\rho \vdash \kw{let}~x = M_1~\kw{in}~M_2 \Rightarrow
    (\Local{(\locs(C_{M_1})\setminus\fv(p_1))} C_{M_1}); C_{M_2}, p_2$ for (\ref{p1}).
    For (\ref{p2}), let $\Sigma'_1=\Sigma_1|_{\fv(p_1)}$ and 
    $\Sigma' = \Sigma'_1,\Sigma_2 \vdash \diamond$.
    By weakening, we find $\Sigma,\Sigma' \vdash p_2 : t_2$ and
    $\Sigma \vdash     (\Local{(\locs(C_{M_1})\setminus\fv(p_1))} C_{M_1}); C_{M_2} : \Sigma'$.
    
  \item\noindent{\hskip-12 pt\ref{Fun Observe}} $\Gamma \vdash \kw{observe}_b~ E : \unitT$.  By the IH, with
    $\Sigma' = \emptyEnv$ from $IH_2$.

  \item\noindent{\hskip-12 pt\ref{Fun Random}} $\Gamma \vdash \kw{random} (D(V)) : b_{n+1}$.  We have:
    \[\begin{array}{l}
      D : (x_1:b_1 * ... * x_n:b_n) \rightarrow b_{n+1} \\
      \Gamma \vdash V : (b_1*...*b_n)
    \end{array}\]
    We have, by the IH:
    \[\begin{array}{lr}
      \rho \vdash V \Rightarrow C,p & (IH_1)\\
      \exists \Sigma' & (IH_2) \\
      \quad \Sigma,\Sigma' \vdash p : t & (*) \\
      \quad \Sigma \vdash C : \Sigma' &
    \end{array}\]
    So $\rho \vdash \kw{random} (D(V)) \Rightarrow C;l \sample D(p),l$, for
    (\ref{p1}).
    We find (\ref{p2}) by (*) and by (Imp Seq), (Imp Random), and the IH
    $\Sigma \vdash C;l : \Sigma',l$, where $\Sigma',l \vdash l :
    b_{n+1}$.\qedhere    
  \end{enumerate}
\end{proof}

\subsection{Proof of Theorem~\ref{thm:fun-to-imp-correct}}\label{sec:proof-thms}

We use the following lemma.
\begin{lem}[Value equivalence]
  \label{lem:valequiv}
  If $\Gamma \vdash V : t$ and $\Sigma \vdash \rho : \Gamma$ and $\rho
  \vdash V \Rightarrow C,p$ 
  then $\impdt{C} = \Aarr f$, where $f$ is either $id$ or a series of
  (independent) calls to $\Padd$:
  \[f = \lambda s.~\Padd l_1 (\Padd l_2 (... (\Padd l_n (s,c_n))
  ...,c_2) ,c_1) \] 
  where each of the $l_i$ are distinct, and \[\aqq{V} = \Aarr
  (\metaF{lift}~\rho) \Athen \impdt{C} \Athen \Aarr (\lambda s.~
  \metaF{restrict}~\rho~ s, \vqq p s)\]
\end{lem}
  \proof
    By induction on the derivation of $\Gamma \vdash V : t$.
    \begin{enumerate}[\hbox to8 pt{\hfill}] 
      
    \item\noindent{\hskip-12 pt\ref{Fun Var}} $\Gamma \vdash x : t$, so $x:t \in \Gamma$ and $\Sigma
      \vdash \rho(x) : t$.
      We have $\rho \vdash x \Rightarrow \kw{nil},\rho(x)$, so $f =
      id$.
      \[\begin{array}{rl}
         & \aqq{x} \\
        = & \Aarr (\lambda s.~ (s, \vqq{x} s)) \\
        = & \Aarr (\lambda s.~ (s, \Plookup{x} s)) \\
        = & \Aarr (\lambda s.~ 
              (\metaF{restrict}~ \rho (\metaF{lift}~ \rho), 
               \vqq p (\metaF{lift}~ \rho~ s))) \\
        = & \metaF{lift}~ \rho \Athen (\lambda s.~ (\metaF{restrict}~ \rho~ s, \vqq p s)) \\
        = & \metaF{lift}~ \rho \Athen \Aarr id \Athen 
              (\lambda s.~ (\metaF{restrict}~ \rho~ s, \vqq p s)) \\
        = & \metaF{lift}~ \rho \Athen \aqq{x} \Athen
              (\lambda s.~ (\metaF{restrict}~ \rho~ s, \vqq p s))
        \end{array}\]
        
      \item\noindent{\hskip-12 pt\ref{Fun Const}} $\Gamma \vdash c : ty(c)$.  We have $\rho \vdash c
        \Rightarrow l \assign c, l$, so $f = \lambda s.~ \Padd{l}~ (s,c)$.
        \[\begin{array}{rl}
           & \aqq{c} \\
          = & \Aarr (\lambda s.~ s, c) \\
          = & \Aarr (\lambda s.~ 
                \metaF{restrict}~ \rho (\metaF{lift}~ \rho~ s), 
                \vqq l (\Padd{l}~ (\metaF{lift}~ \rho~ s,c))) \\
          = & \Aarr (\metaF{lift}~ \rho) \Athen \Aarr 
                (\lambda s.~ \metaF{restrict}~ \rho~ s, \vqq l (\Padd{l}~ (s,c))) \\
          = & \Aarr (\metaF{lift}~ \rho) \Athen \Aarr (\lambda s.~ \Padd{l}~ (s,c)) \Athen 
                \Aarr (\lambda s.~ \metaF{restrict}~ \rho~ s, \vqq l s) \\
          = & \Aarr (\metaF{lift}~ \rho) \Athen \impdt{l \assign c} \Athen 
                \Aarr (\lambda s.~ \metaF{restrict}~ \rho~ s, \vqq l s)
        \end{array}\]
        
      \item\noindent{\hskip-12 pt\ref{Fun Pair}} $\Gamma \vdash (V_1, V_2) : t_1 * t_2$.  We have $\rho
        \vdash V_1,V_2 \Rightarrow C_1;C_2,(p_1,p_2)$.
        By the IH, $\impdt{C_1} = \Aarr f_1$ and $\impdt{C_2} = \Aarr
        f_2$, where $f_1$ and $f_2$ are either $id$ or $\Padd{}$s.
        We also have:
        \[\begin{array}{rl}
           & \aqq{Vi} \\
          = & \Aarr (\lambda s.~ s, \vqq{Vi}s) \\
          = & \Aarr (\metaF{lift}~ \rho) \Athen \impdt{C_i} \Athen
            \Aarr (\lambda s.~ \metaF{restrict}~ \rho~ s, \vqq{p_i} s) \\
          = & \Aarr (\metaF{lift}~ \rho) \Athen \Aarr f_i \Athen
            \Aarr (\lambda s.~ \metaF{restrict}~ \rho~ s, \vqq{p_i} s) \\
          = & \Aarr (\lambda s.~ 
            \metaF{restrict}~ \rho (f_i (\metaF{lift}~ \rho~ s)),
            \vqq{p_i} (f_i ~(\metaF{lift}~ \rho~ s))) \\
          = & \Aarr (\lambda s.~ s, \vqq{p_i} (f_i~ (\metaF{lift}~ \rho~ s)))
        \end{array}\]
        So $\vqq{Vi}s = \vqq{p_i} (f_i (\metaF{lift}~ \rho~ s))$.
        Let $f = f_1;f_2$.  We derive:
        \[\begin{array}{rl}
           & \aqq{V_1,V_2}  \\
          = & \Aarr (\lambda s.~ s, (\vqq{V_1}s,\vqq{V_2}s)) \\
          = & \Aarr (\lambda s.~ s, 
            (\vqq{p_1} (f_1~ (\metaF{lift}~ \rho~ s)),
             \vqq{p_2} (f_2~ (\metaF{lift}~ \rho~ s))) \text{\quad by weakening/independence} \\
          = & \Aarr (\lambda s.~ s, 
            (\vqq{p_1} ((f_1;f_2) (\metaF{lift}~ \rho~ s)), 
             \vqq{p_2} ((f_1;f_2) (\metaF{lift}~ \rho~ s))) \\ 

          = & \Aarr (\lambda s.~ \metaF{restrict}~ \rho~ (f_1;f_2 (\metaF{lift}~ \rho~ s)), \\
          &  \qquad (\vqq{p_1} ((f_1;f_2) (\metaF{lift}~ \rho~ s)), 
             \vqq{p_2} ((f_1;f_2) (\metaF{lift}~ \rho~ s))) \\ 

          = & \Aarr (\metaF{lift}~ \rho) \Athen \Aarr (f_1;f_2) \Athen 
              \Aarr (\lambda s.~ \metaF{restrict}~ \rho~ s, (\vqq{p_1} s, \vqq{p_2} s)) \\

          = & \Aarr (\metaF{lift}~ \rho) \Athen \impdt{C_1} \Athen \impdt{C_2} \Athen 
              \Aarr (\lambda s.~ \metaF{restrict}~ \rho~ s, \vqq{(p_1,p_2)} s) \\

          = & \Aarr (\metaF{lift}~ \rho) \Athen \impdt{C_1;C_2} \Athen 
              \Aarr (\lambda s.~ \metaF{restrict}~ \rho~ s, \vqq{(p_1,p_2)} s)
          \rlap{\hbox to 78 pt{\hfill\qEd}}
          \end{array}\]
      \end{enumerate}

\begin{restate}{Theorem~\ref{thm:fun-to-imp-correct}}
$\Gamma \vdash M : t$ and $\Sigma \vdash \rho : \Gamma$ and $\rho \vdash M \Rightarrow C,p$ then:
\[\aqq{M} = \Aarr (\metaF{lift}~\rho) \Athen \impdt{C} \Athen \Aarr (\lambda s.~(\metaF{restrict}~\rho~s, \vqq ps))\]
\end{restate}

\proof
    By induction on $\Gamma \vdash M : t$.
    \begin{enumerate}[\hbox to8 pt{\hfill}] 

    \item\noindent{\hskip-12 pt\ref{Fun Var}} By the value lemma.
    \item\noindent{\hskip-12 pt\ref{Fun Const}} By the value lemma.
    \item\noindent{\hskip-12 pt\ref{Fun Pair}} By the value lemma.
    \item\noindent{\hskip-12 pt\ref{Fun Operator}} $\Gamma \vdash V_1 \otimes V_2 : b_3$ and $\rho \vdash V_1
      \otimes V_2 \Rightarrow (C_1;C_2; l \assign l_1 \otimes l_2),l$.
      We have $\aqq{V_1 \otimes V_2} = \Aarr (\lambda s.~ s,
      {\otimes}(\vqq{V_1}s, \vqq{V_2}s))$.
      By the value lemma (Lemma~\ref{lem:valequiv}):
      \[\begin{array}{rl}
        & \aqq{V_i} \\ 
        = & \Aarr (\lambda s.~ s, \vqq{V_i}s) \\
        = & \Aarr (\metaF{lift}~ \rho) \Athen \impdt{C_i} \Athen
          \Aarr (\lambda s.~ \metaF{restrict}~ \rho~ s, \vqq{l_i} s) \\
        = & \Aarr (\metaF{lift}~ \rho) \Athen \Aarr f_i \Athen 
          \Aarr (\lambda s.~ \metaF{restrict}~ \rho~ s, \vqq{l_i} s) \\
        =  & \Aarr (\lambda s.~ 
          \metaF{restrict}~ \rho~ (f_i~ (\metaF{lift}~ \rho~ s)), 
          \vqq{l_i} (f_i~ (\metaF{lift}~ \rho~ s)))  \\
        = & \Aarr (\lambda s.~ s, \vqq{l_i} (f_i (\metaF{lift}~ \rho~ s))) \\
        = & \Aarr (\lambda s.~ s, \vqq{l_i} ((f_1;f_2) (\metaF{lift}~ \rho~ s)) \qquad \text{by weakening/independence}
      \end{array}\]
      So $\vqq{V_i}s = \vqq{l_i} ((f_1;f_2) (\metaF{lift}~ \rho~ s))$.
      We derive:
      \[\begin{array}{rl}
        & \aqq{V_1 \otimes V_2} \\
        = & \Aarr (\lambda s.~ s, \vqq{V_1}s \otimes \vqq{V_2}s) \\
        = & \Aarr (\lambda s.~ s, {\otimes} (
          \vqq{l_1} ((f_1;f_2) (\metaF{lift}~ \rho~ s)), 
          \vqq{l_2} ((f_1;f_2) (\metaF{lift}~ \rho~ s)))) \\
        = & \Aarr (\metaF{lift}~ \rho) \Athen \Aarr (f_1;f_2) \Athen 
          \Aarr (\lambda s.~ \metaF{restrict}~ \rho~ s, {\otimes}(\vqq{l_1} s, \vqq{l_2} s))))\\
        = & \Aarr (\metaF{lift}~ \rho) \Athen \impdt{C_1} \Athen \impdt{C_2} \Athen 
          \Aarr (\lambda s.~ \metaF{restrict}~ \rho~ s, {\otimes}(\vqq{l_1} s, \vqq{l_2} s))))\\
        = & \Aarr (\metaF{lift}~ \rho) \Athen \impdt{C_1} \Athen \impdt{C_2} \Athen \impdt{l\assign l_1\otimes l_2} \Athen
          \Aarr (\lambda s.~ \metaF{restrict}~ \rho~ s, \vqq{l} s)))\\
        = & \Aarr (\metaF{lift}~ \rho) \Athen \impdt{C_1;C_2;l\assign l_1\otimes l_2} \Athen 
          \Aarr (\lambda s.~ \metaF{restrict}~ \rho~ s, \vqq{l} s)))
      \end{array}\]

    \item\noindent{\hskip-12 pt\ref{Fun Proj1}} $\Gamma \vdash V.1 : t_1$ and $\Gamma \vdash V :
      t_1 * t_2$.  We have $\rho \vdash V \Rightarrow C, (p_1,p_2)$
      and $\rho \vdash V.1 \Rightarrow C, p_1$.
      By the value lemma as before, we can conclude $\vqq{V}s =
      \vqq{(p_1,p_2)} (f~ (\metaF{lift}~ \rho~ s))$.  Therefore:
      \[\begin{array}{rl}
        & \aqq{V.1}  \\ 
        = & \Aarr (\lambda s.~ s, \mathtt{fst}~ \vqq{V}s) \\
        = & \Aarr (\lambda s.~ s, \mathtt{fst}~ (\vqq{(p_1,p_2)}(f~ (\metaF{lift}~ \rho~ s))) \\
        = & \Aarr (\lambda s.~ s, \vqq{p_1} (f~ (\metaF{lift}~ \rho~ s)) \\
        = & \Aarr (\metaF{lift}~ \rho) \Athen \Aarr f \Athen 
          \Aarr (\lambda s.~ \metaF{restrict}~ \rho~ s, \vqq{p_1} s) \\
        = & \Aarr (\metaF{lift}~ \rho) \Athen \impdt{C} \Athen 
          \Aarr (\lambda s.~ \metaF{restrict}~ \rho~ s, \vqq{p_1} s) 
      \end{array}\]

    \item\noindent{\hskip-12 pt\ref{Fun Proj2}} Symmetric to Proj1.

    \item\noindent{\hskip-12 pt\ref{Fun If}} $\Gamma \vdash \Lif{V_1}{M_2}{M_3} : t$.  
      We have:
      \[\rho \vdash ... \Rightarrow C_1; 
      \Lif{l_1}{\Local{\locs(C_2)} C_2; p \assign 2}{\Local{\locs(C_3)}
      C_3; p \assign p_3}, p\]
      Our IHs are: $\aqq{M_i} = \Aarr (\metaF{lift}~ \rho) \Athen
      \impdt{C_i} \Athen \Aarr (\lambda s.~ \metaF{restrict}~ \rho~ s,
      \vqq{p_i} s)$.
      By the value lemma we have 
      $\impdt{V_1} = \Aarr f_1$ for some $f_1$ such that $\vqq{V_1} s = §\vqq{l_1} (f_1 (\metaF{lift}~ \rho~ s))$.
      We now calculate (at length):
      \[\begin{array}{rl}
        & \aqq{\Lif{V_1}{M_2}{M_3}} \\
        = & \Achoose (\lambda s.~ \vqq{V_1}s)\ \aqq{M_2}\ \aqq{M_3} \\

        & \\

        = & \Achoose (\lambda s.~\vqq{l_1} (f_1~ (\metaF{lift}~ \rho~ s))) \\
        & \qquad (\Aarr (\metaF{lift}~ \rho) \Athen \impdt{C_2} \Athen 
          \Aarr (\lambda s.~ \metaF{restrict}~ \rho~ s, \vqq{p_2} s)) \\
        & \qquad (\Aarr (\metaF{lift}~ \rho) \Athen \impdt{C_3} \Athen 
         \Aarr (\lambda s.~ \metaF{restrict}~ \rho~ s, \vqq{p_3} s)) \\

        & \\

        = & \Aarr (\metaF{lift}~ \rho) \Athen \Achoose (\lambda s.\vqq{l_1} (f_1~ s)) \\
        & \qquad (\impdt{C_2} \Athen \Aarr (\lambda s.~ \metaF{restrict}~ \rho~ s, \vqq{p_2} s)) \\
        & \qquad (\impdt{C_3} \Athen \Aarr (\lambda s.~ \metaF{restrict}~ \rho~ s, \vqq{p_3} s))\\

        & \\

        = & \Aarr (\metaF{lift}~ \rho) \Athen \Achoose (\lambda s.\vqq{l_1} (f_1~ s)) \\
        & \qquad (\impdt{C_2} \Athen \impdt{p \assign p_2} \Athen 
            \Aarr (\lambda s.~ \metaF{restrict}~ \rho~ s, \vqq p s)) \\
        & \qquad (\impdt{C_3} \Athen \impdt{p \assign p_3} \Athen 
            \Aarr (\lambda s.~ \metaF{restrict}~ \rho~ s, \vqq p s))  \\

        & \\            

        = & \Aarr (\metaF{lift}~ \rho) \Athen \Achoose (\lambda s.\vqq{l_1} (f_1~ s)) \\
        & \qquad (\impdt{C_2} \Athen \impdt{p \assign p_2} \Athen \Aarr (\Pdrop \locs(C_2)) 
            \Athen \Aarr (\lambda s.~ \metaF{restrict}~ \rho~ s, \vqq p s)) \\
        & \qquad (\aqq{C_3} \Athen \aqq{p \assign p_3} \Athen \Aarr (\Pdrop \locs(C_3))
            \Athen \Aarr (\lambda s.~ \metaF{restrict}~ \rho~ s, \vqq p s)) \\
  
        & \\

        = & \Aarr (\metaF{lift}~ \rho) \Athen (\Achoose (\lambda s.\vqq{l_1} (f_1~ s)) \\
        & \qquad (\aqq{C_2} \Athen \aqq{p \assign p_2} \Athen \Aarr (\Pdrop \locs(C_2)))  \\
        & \qquad (\aqq{C_3} \Athen \aqq{p \assign p_3} \Athen \Aarr (\Pdrop \locs(C_3)))) \Athen \\
        & \Aarr (\lambda s.~ \metaF{restrict}~ \rho~ s, \vqq p s) \\
       
        & \\

        = & \Aarr (\metaF{lift}~ \rho) \Athen \aqq{C_1} \Athen
          (\Achoose (\lambda s.\vqq{l_1} s) \\
        & \qquad (\aqq{C_2;p \assign p_2} \Athen \Aarr (\Pdrop \locs(C_2))) \\
        & \qquad (\aqq{C_3;p \assign p_3} \Athen \Aarr (\Pdrop \locs(C_3)))) \Athen \\
        & \Aarr (\lambda s.~ \metaF{restrict}~ \rho~ s, \vqq p s) \\
 
        = & \Aarr (\metaF{lift}~ \rho) \Athen \aqq{C_1} \Athen
          (\Achoose (\lambda s.\vqq{l_1} s) \\
        & \qquad (\aqq{\Local{\locs(C_2)} C_2;p \assign p_2}) \\
        & \qquad (\aqq{\Local{\locs(C_3)} C_3;p \assign p_3})) \Athen\\
        & \Aarr (\lambda s.~ \metaF{restrict}~ \rho~ s, \vqq p s)
      \end{array}\]

    \item\noindent{\hskip-12 pt\ref{Fun Let}} $\Gamma \vdash \kw{let}~ x = M_1 ~\kw{in}~ M_2 :
      t_2$; by inversion, $\Gamma \vdash M_1 : t_1$ and $\Gamma,x:t_1
      \vdash M_2 : t_2$.

      Let $\rho' = \rho\{x \mapsto p_1\}$ and $\Sigma_1=(\locs(C_1)\setminus\fv(p_1))$.  We have:
      \[\begin{array}{l}
        \rho \vdash M_1 \Rightarrow C_1, p_1 \\
        \rho' \vdash M_2 \Rightarrow C_2, p_2 \\
        \rho \vdash \kw{let}~ x = M_1 ~\kw{in}~ M_2 \Rightarrow 
            (\Local{\Sigma_1} C_1);C_2, p_2
      \end{array}\]
      As our IHs:
      \begin{eqnarray*}
        \aqq{M_1} &=& \Aarr (\metaF{lift}~ \rho) \Athen \aqq{C_1} \Athen 
              \Aarr (\lambda s.~ \metaF{restrict}~ \rho~ s, \vqq{p_1} s) \\
        \aqq{M_2} &=& \Aarr (\metaF{lift}~ \rho') \Athen \aqq{C_2} \Athen 
              \Aarr (\lambda s.~ \metaF{restrict}~ \rho' s, \vqq{p_2} s)
      \end{eqnarray*}

      We derive:
      \[\begin{array}{rl}
        & \aqq{\kw{let}~ x = M_1 ~\kw{in}~ M_2} \\

        = & \aqq{M_1} \Athen \Aarr (\Padd x) \Athen \aqq{M_2} \Athen 
        \Aarr (\lambda s,y.~ \Pdrop x~ s,y) \\

        = & \Aarr (\metaF{lift}~ \rho) \Athen \aqq{C_1} \Athen 
        \Aarr (\lambda s.~ \metaF{restrict}~ \rho~ s, \vqq{p_1} s) \Athen 
        \Aarr (\Padd x) \Athen\\
        & \aqq{M_2} \Athen \Aarr (\lambda s,y.~ \Pdrop x~ s,y) \\

        = &  \Aarr (\metaF{lift}~ \rho) \Athen \aqq{C_1} \Athen 
        \Aarr (\lambda s.~ \metaF{restrict}~ \rho~ s, \vqq{p_1} s) \Athen \\
        & \Aarr (\Padd x) \Athen \Aarr (\metaF{lift}~ \rho') \Athen \\
        & \aqq{C_2} \Athen \Aarr (\lambda s.~ \metaF{restrict}~ \rho' s, \vqq{p_2} s) \Athen  \Aarr (\lambda s,y.~ \Pdrop x~ s,y) \\

        = & \Aarr (\metaF{lift}~ \rho) \Athen \aqq{C_1} \Athen \Aarr (\Pdrop(\dom(\Sigma_1))) \Athen 
        \\ &\aqq{C_2} \Athen \Aarr (\lambda s.~ \metaF{restrict}~ \rho' s, \vqq{p_2} s) \Athen \Aarr (\lambda s,y.~ \Pdrop x~ s,y) \\

        = & \Aarr (\metaF{lift}~ \rho) \Athen \aqq{C_1} \Athen \Aarr (\Pdrop(\dom(\Sigma_1)))\Athen \\ &\aqq{C_2} \Athen
        \Aarr (\lambda s.~ \metaF{restrict}~ \rho~ s, \vqq{p_2} s) \\

        = & \Aarr (\metaF{lift}~ \rho) \Athen \aqq{(\Local{\Sigma_1} C_1);C_2} \Athen
        \Aarr (\lambda s.~ \metaF{restrict}~ \rho~ s, \vqq{p_2} s)
      \end{array}\]

    \item\noindent{\hskip-12 pt\ref{Fun Random}} $\Gamma \vdash \kw{random} (D(V)) : b$, where $D :
      (b_1,...,b_n) \to b_{n+1}, \Gamma \vdash V : (b_1,...,b_n)$.  We have
      $\rho \vdash V \Rightarrow C,p$ and $\rho \vdash D(V)
      \Rightarrow C;l \assign D(p),l$.
      By the value lemma, $\aqq{C} = \Aarr f$ and $\vqq{V}s = \vqq p (f~
      (\metaF{lift}~ \rho~ s))$.
      We derive:
      \[\begin{array}{rl}
        & \aqq{\kw{random} (D(V))} \\
        = & \Aextend (\lambda s.~\mu_{D(\vqq{V}s)}) \\
        = & \Aextend (\lambda s.~\mu_{D(p (f (\metaF{lift}~ \rho~ s)))}) \\
        = & \Aarr (\metaF{lift}~ \rho) \Athen \Aextend (\lambda s.~\mu_{D(p (f s))}) \Athen 
        \Aarr (\lambda s,v.~ \metaF{restrict}~ \rho~ s,v) \\
        = & \Aarr (\metaF{lift}~ \rho) \Athen \Aarr f \Athen \Aextend (\lambda s.~\mu_{D(\vqq p s)}) 
        \Athen \Aarr (\lambda s,v.~ \metaF{restrict}~ \rho~ s,v) \\
        = & \Aarr (\metaF{lift}~ \rho) \Athen \aqq{C} \Athen \Aextend (\lambda s.~\mu_{D(\vqq p s)}) 
        \Athen \Aarr (\lambda s,v.~ \metaF{restrict}~ \rho~ s,v) \\
        = & \Aarr (\metaF{lift}~ \rho) \Athen \aqq{C} \Athen \Aextend (\lambda s.~\mu_{D(\vqq p s)}) \Athen \\
        & \Aarr (\Padd l) \Athen \Aarr (\lambda s.~ \metaF{restrict}~ \rho~ s,\vqq l s) \\
        = & \Aarr (\metaF{lift}~ \rho) \Athen \aqq{C;l \assign D(p)} 
        \Athen \Aarr (\lambda s.~ \metaF{restrict}~ \rho~ s, \vqq l s)
      \end{array}\]

    \item\noindent{\hskip-12 pt\ref{Fun Observe}} $\Gamma \vdash \kw{observe}~ V : \unitT$ and
      $\Gamma \vdash V : b$ for some base type $b$.
      We have $\rho \vdash V \Rightarrow C,l$.
      By the value lemma: $\aqq{C} = \Aarr f$ and $\vqq{V}s = \vqq l (f~
      (\metaF{lift}~ \rho~ s))$.
      \[\begin{array}{rl}
         & \aqq{\kw{observe}~ V} \\ 
        = & \Aobserve (\lambda s.~\vqq{V}s) \Athen \Aarr (\lambda s.~(s,()) \\
        = & \Aobserve (\lambda s.~l (f (\metaF{lift}~ \rho~ s))) \Athen \Aarr (\lambda s.~s,()) \\
        = & \Aarr (\metaF{lift}~ \rho) \Athen \Aobserve (\lambda s.~ \vqq l (f~ s))
        \Athen \Aarr (\lambda s.~\metaF{restrict}~ \rho~ s,()~ s) \\
        = & \Aarr (\metaF{lift}~ \rho) \Athen \Aarr f \Athen \Aobserve (\lambda s.~ \vqq l s) 
        \Athen \Aarr (\lambda s.~\metaF{restrict}~ \rho~ s,()~ s) \\
        = & \Aarr (\metaF{lift}~ \rho) \Athen \aqq{C} \Athen \Aobserve (\lambda s.~ \vqq l s) 
        \Athen \Aarr (\lambda s.~\metaF{restrict}~ \rho~ s,()~ s) \\
        = & \Aarr (\metaF{lift}~ \rho) \Athen \aqq{C;\Aobserve l} 
        \Athen \Aarr (\lambda s.~\metaF{restrict}~ \rho~ s,()~ s)
        \rlap{\hbox to 83 pt{\hfill\qEd}}
      \end{array}\]
    \end{enumerate}
\pagebreak
\end{FULL} 

\end{document}